\documentclass[11pt]{article}

\usepackage[colorlinks, linkcolor=blue, citecolor=blue]{hyperref}            
\usepackage{color}
\usepackage{graphicx,subfigure,amsmath,amssymb,amsfonts,bm,epsfig,epsf,url,dsfont}
\usepackage{amsthm}
\usepackage{tikz}
\usepackage{bbm}      
\usepackage{booktabs}
\usepackage{cases}
\usepackage{fullpage}
\usepackage[small,bf]{caption}
\usepackage{natbib}
\usepackage[top=1in,bottom=1in,left=1in,right=1in]{geometry}
\usepackage{fancybox}
\usepackage[bottom]{footmisc}

\newcommand{\bR}{\mathbb{R}}
\newcommand{\mD}{\mathcal{D}}
\newcommand{\mL}{\mathcal{L}}

\newcommand{\mG}{\mathcal{G}}
\newcommand{\mN}{\mathcal{N}}

\newcommand{\mP}{\mathcal{P}}
\newcommand{\mQ}{\mathcal{Q}}
\newcommand{\mR}{\mathcal{R}}

\newcommand{\mU}{\mathcal{U}}

\newcommand{\TV}{d_{\text{TV}}}
\newcommand{\KL}{d_{\text{KL}}}

\newcommand{\bP}{\mathbb{P}}
\newcommand{\bE}{\mathbb{E}}

\newcommand{\pr}[1]{\textsc{#1}}

\newtheorem{theorem}{Theorem}[section]
\newtheorem{definition}[theorem]{Definition}

\newtheorem{fact}[theorem]{Fact}
\newtheorem{proposition}[theorem]{Proposition}
\newtheorem{corollary}[theorem]{Corollary}
\newtheorem{conjecture}[theorem]{Conjecture}
\newtheorem{lemma}[theorem]{Lemma}

\newtheorem{remark}[theorem]{Remark}

\newcommand{\ind}{\mathds{1}}
\newcommand{\la}{\langle}
\newcommand{\ra}{\rangle}

\newcommand{\cS}{\mathcal{S}}
\newcommand{\cF}{\mathcal{F}}
\newcommand{\cD}{\mathcal{D}}
\newcommand{\cZ}{\mathcal{Z}}
\newcommand{\cX}{\mathcal{X}}

\newcommand{\Dh}{\widehat{D}}

\renewcommand{\P}{\mathbb{P}}
\newcommand{\E}{\mathbb{E}}

\newcommand{\Ot}{\tilde{O}}

\newcommand{\inv}{^{-1}}

\newcommand{\lr}{\mathsf{LR}}
\newcommand{\lrd}{\lr^{\leq D}}

\newcommand{\kpds}{{$k$-\textsc{pds}}}
\newcommand{\kpc}{{$k$-\textsc{pc}}}

\newenvironment{fminipage}%
  {\begin{Sbox}\begin{minipage}}%
  {\end{minipage}\end{Sbox}\fbox{\TheSbox}}

\newenvironment{algbox}[0]{\vskip 0.2in
\noindent 
\begin{fminipage}{6.3in}
}{
\end{fminipage}
\vskip 0.2in
}

\begin{document}

\title{Average-Case Lower Bounds for Learning Sparse Mixtures, \\ Robust Estimation and Semirandom Adversaries}

\author{Matthew Brennan
\thanks{Massachusetts Institute of Technology. Department of EECS. Email: \texttt{brennanm@mit.edu}.}
\and 
Guy Bresler
\thanks{Massachusetts Institute of Technology. Department of EECS. Email: \texttt{guy@mit.edu}.}
}
\date{\today}

\maketitle

\begin{abstract}

This paper develops several average-case reduction techniques to show new hardness results for three central high-dimensional statistics problems, implying a statistical-computational gap induced by robustness, a detection-recovery gap and a universality principle for these gaps. A main feature of our approach is to map to these problems via a common intermediate problem that we introduce, which we call Imbalanced Sparse Gaussian Mixtures.
We assume the planted clique conjecture for a version of the planted clique problem where the position of the planted clique is mildly constrained, and from this obtain the following computational lower bounds that are tight against efficient algorithms:
\begin{itemize}
\item \textbf{Robust Sparse Mean Estimation:} Estimating an unknown $k$-sparse mean of $n$ samples from a $d$-dimensional Gaussian is a gapless problem, with the truncated empirical mean achieving the optimal sample complexity of $n = \tilde{\Theta}(k)$. However, if an $\epsilon$-fraction of these samples are corrupted, the best known polynomial time algorithms require $n = \tilde{\Omega}(k^2)$ samples. We give a reduction showing that this sample complexity is necessary, providing the first average-case complexity evidence for a conjecture of \cite{li2017robust, balakrishnan2017computationally}. Our reduction also shows that this statistical-computational gap persists even for algorithms estimating within a suboptimal $\ell_2$ rate of $\tilde{O}(\sqrt{\epsilon})$, rather than the minimax rate of $O(\epsilon)$.
\item \textbf{Semirandom Community Recovery:} The problem of finding a $k$-subgraph with elevated edge density $p$ within an Erd\H{o}s-R\'{e}nyi graph on $n$ vertices with edge density $q = \Omega(p)$ is a canonical problem believed to have different computational limits for detection and for recovery. The conjectured detection threshold has been established through reductions from planted clique, but the recovery threshold has remained elusive.
 We give a reduction showing that the detection and recovery thresholds coincide when the graph is perturbed by a \emph{semirandom adversary}, in the case where $q$ is constant. 
 This yields the first average-case evidence towards the recovery conjectures of \cite{chen2016statistical, hajek2015computational, brennan2018reducibility, brennan2019universality}.
\item \textbf{Universality of $k$-to-$k^2$ Gaps in Sparse Mixtures:} Extending the techniques for our other two reductions, we also show a universality principle for computational lower bounds at $n = \tilde{\Omega}(k^2)$ samples for learning $k$-sparse mixtures, under mild conditions on the likelihood ratios of the planted marginals.
Our result extracts the underlying problem structure leading to a gap, independently of distributional specifics, and is a first step towards explaining the ubiquity of $k$-to-$k^2$ gaps in high-dimensional statistics.
\end{itemize}
Our reductions produce structured and Gaussianized versions of an input graph problem, and then rotate these high-dimensional Gaussians by matrices carefully constructed from hyperplanes in $\mathbb{F}_r^t$. For our universality result, we introduce a new method to perform an algorithmic change of measure tailored to sparse mixtures. We also provide evidence that the mild promise in our variant of planted clique does not change the complexity of the problem.
\end{abstract}

\pagebreak

\tableofcontents

\pagebreak


\section{Introduction}



A primary aim of the field of mathematical statistics is to determine how much data is needed for various estimation tasks, and to analyze the performance of practical algorithms. 
For a century, the focus has been on \emph{information-theoretic} limits. However, the study of high-dimensional structured estimation problems over the last two decades has revealed that the much more relevant quantity -- the amount of data needed by \emph{computationally efficient} algorithms -- may be significantly higher than what is achievable without computational constraints. 

Because data in real-world problems is not adversarially generated, the mathematical analysis of estimation problems typically assumes a probabilistic model on the data. In computer science, combinatorial problems with random inputs have been studied since the 1970's \cite{karp1977probabilistic,kuvcera1977expected}. In the 1980's, Levin's theory of average-case complexity \cite{levin1986average} crystallized the notion of an average-case reduction and obtained completeness results.
Average-case hardness reductions are notoriously delicate, requiring that a distribution over instances in a conjecturally hard problem be mapped precisely to the target distribution, making gadget-based reductions from worst-case complexity ineffective. For this reason, much of the recent work showing hardness for statistical problems shows hardness for \emph{restricted models of computation} (or equivalently, classes of algorithms), such as statistical query algorithms \cite{feldman2013statistical,feldman2015complexity,diakonikolas2017statistical,diakonikolas2019efficient}, sum of squares \cite{barak2016nearly,hopkins2017power}, classes of circuits \cite{razborov1997natural,rossman2008constant,rossman2014monotone}, local algorithms \cite{gamarnik2017limits,linial1992locality}, message-passing \cite{zdeborova2016statistical,lesieur2015mmse,lesieur2016phase,krzakala2007gibbs,ricci2018typology}, and others. Another line of work attempts to understand computational limitations via properties of the energy landscape of solutions to estimation problems \cite{achlioptas2008algorithmic, gamarnik2017high,arous2017landscape, arous2018algorithmic,ros2019complex,chen2019suboptimality, gamarnik2019landscape}. 

This paper develops several average-case reduction techniques to show new hardness results for three central high-dimensional statistics problems. We assume the planted clique conjecture for a version of the planted clique problem where the position of the planted clique is mildly constrained. As discussed below, planted clique is known to be hard for all well-studied restricted models of computation, and we confirm in Section~\ref{sec:evidencekpc} for statistical query algorithms and low-degree polynomial tests 
that this hardness remains unchanged for our modification.
Aside from the advantage of being future-proof against new classes of algorithms, showing that a problem of interest is hard by reducing from planted clique
effectively \emph{subsumes} hardness in the restricted models and thus gives much stronger evidence for hardness.

%

\subsection{Results}

We show within the framework of average-case complexity that robustness, in the form of model misspecification or constrained adversaries, introduces statistical-computational gaps in several natural problems. In doing so, we develop a number of new techniques for average-case reductions and arrive at a universality class of problems with statistical-computational gaps. 
 \emph{Robust problem formulations} are somewhere between worst-case and average-case, and indeed are strictly harder to solve than the corresponding purely average-case formulations, which is captured in our reductions. The first problem we consider is robust sparse mean estimation, the problem of sparse mean estimation in Huber's contamination model \cite{huber1992robust, huber1965robust}.  Lower bounds for robust sparse mean estimation have been shown for statistical query algorithms \cite{diakonikolas2017statistical}, but to the best of our knowledge our reductions are the first average-case hardness results that capture \emph{hardness induced by robustness}. The second problem, semirandom community recovery, is the planted dense subgraph problem under a semirandom (monotone) adversary \cite{feige2001heuristics,feige2000finding}. It is surprising that reductions from planted clique can capture the shift in hardness threshold due to semirandom adversaries, since the semirandom version of
 planted clique has the same threshold as the original \cite{feige2000finding}. 

Both of these robust problems are mapped to via an intermediate problem that we introduce, Imbalanced Sparse Gaussian Mixtures (ISGM). We also develop a new technique to reduce from ISGM to a \emph{universality class} of sparse mixtures, which shows that the $k$ to $k^2$ statistical-computational gaps in many statistical estimation problems are not coincidental and in particular do not depend on specifics of the distributions involved. The fact that all of our reductions go through ISGM unifies the techniques associated to each of the reductions and illustrates the utility of constructing reductions that use natural problems as intermediates, a perspective espoused in \cite{brennan2018reducibility}. By constructing precise distributional maps preserving the canonical simple-versus-simple problem formulations, the maps can be composed, which greatly facilitates reducing to new problems. 

This work is closely related to a recent line of research showing lower bounds for average-case problems in computer science and statistics based on the planted clique (\pr{pc}) conjecture. The \pr{pc} conjecture asserts that there is no polynomial time algorithm to detect a clique of size $k = o(\sqrt{n})$ randomly planted in the Erd\H{o}s-R\'{e}nyi graph $\mG(n, 1/2)$. The \pr{pc} conjecture has been used to show lower bounds for testing $k$-wise independence \cite{alon2007testing}, biclustering \cite{ma2015computational, cai2015computational, caiwu2018}, community detection \cite{hajek2015computational}, RIP certification \cite{wang2016average, koiran2014hidden}, matrix completion \cite{chen2015incoherence}, sparse PCA \cite{berthet2013optimal, berthet2013complexity, wang2016statistical, gao2017sparse, brennan2019optimal}, universal submatrix detection \cite{brennan2019universality} and to give a web of average-case reductions to a number of additional problems in \cite{brennan2018reducibility}. The present paper considerably grows the set of statistical problems to which there are precise distribution-preserving average-case reductions.  

In this work, we assume the hardness of \pr{pc} up to $k = o(\sqrt{n})$ under a mild promise on the vertex set of its hidden clique. Specifically, we assume that we are given a partition of its vertex set into $k$ sets of equal size with the guarantee that the hidden clique has exactly one vertex in each set. This mild promise does not asymptotically affect any reasonable notion of entropy of the distribution over the hidden vertex set, scaling it by $1-o(1)$. In Section \ref{sec:evidencekpc}, we provide evidence that this $k$-partite promise does not affect the complexity of planted clique, by showing hardness up to $k = o(\sqrt{n})$ for low-degree polynomial tests and statistical query algorithms. In the cryptography literature, this form of promise is referred to as a small amount of information leakage about the secret. The hardness of the Learning With Errors (LWE) problem has been shown to be robust to leakage \cite{goldwasser10}, and it is left as an interesting open problem to show that a similar statement holds true for \textsc{pc}. Our main results are that this $k$-partite \pr{pc} conjecture ($k\pr{-pc}$) implies the statistical-computational gaps outlined below.

\paragraph{Robust Sparse Mean Estimation.} In sparse mean estimation, the observations $X_1, X_2, \dots, X_n$ are $n$ independent samples from $\mN(\mu, I_d)$ where $\mu$ is an unknown $k$-sparse vector in $\mathbb{R}^d$ and the task is to estimate $\mu$ closely in $\ell_2$ norm. This is a gapless problem, with the efficiently computable truncated empirical mean achieving the optimal sample complexity of $n = \tilde{\Theta}(k)$. However, if an $\epsilon$-fraction of these samples are adversarially corrupted, the best known polynomial time algorithms require $n = \tilde{\Omega}(k^2)$ samples while $n = \tilde{\Omega}(k)$ remains the information-theoretically optimal sample complexity. We give a reduction showing that robust sparse mean estimation is $k$\pr{-pc} hard if $n = \tilde{o}(k^2)$. This yields the first average-case evidence for this statistical-computational gap, which was originally conjectured in \cite{li2017robust, balakrishnan2017computationally}. Our reduction also shows that this statistical-computational gap persists even for algorithms estimating within a suboptimal $\ell_2$ rate of $\tilde{O}(\sqrt{\epsilon})$, rather than the minimax rate of $O(\epsilon)$. The reductions showing these lower bounds can be found in Sections \ref{sec:redisgm} and \ref{sec:robust}. 

\paragraph{Semirandom Community Recovery.} In the planted dense subgraph model of single community recovery, one observes a graph sampled from $\mG(n, q)$ with a random subgraph on $k$ vertices replaced with a sample from $\mG(k, p)$, where $p > q$ are allowed to vary with $n$ and satisfy that $p = O(q)$. Detection and recovery of the hidden community in this model have been studied extensively \cite{arias2014community, butucea2013detection, verzelen2015community, hajek2015computational,chen2016statistical, hajek2016information, montanari2015finding, candogan2018finding} and this model has emerged as a canonical example of a problem with a detection-recovery computational gap. While it is possible to efficiently detect the presence of a hidden subgraph of size $k=\tilde \Omega(\sqrt{n})$ if $(p - q)^2/q(1 - q) = \tilde{\Omega}(n^2/k^4)$, the best known polynomial time algorithms to \emph{recover} the subgraph require a higher signal of $(p - q)^2/q(1 - q) = \tilde{\Omega}(n/k^2)$.

In each of \cite{hajek2015computational, brennan2018reducibility} and \cite{brennan2019universality}, it has been conjectured that the recovery problem is hard below this threshold of $\tilde{\Theta}(n/k^2)$. This \pr{pds} Recovery Conjecture was even used in \cite{brennan2018reducibility} as a hardness assumption to show detection-recovery gaps in other problems including biased sparse PCA and Gaussian biclustering. A line of work has tightly established the conjectured detection threshold through reductions from the \pr{pc} conjecture \cite{hajek2015computational, brennan2018reducibility, brennan2019universality}, while the recovery threshold has remained elusive. Planted clique maps naturally to the detection threshold in this model, so it seems unlikely that the \pr{pc} conjecture could also yield lower bounds at the tighter recovery threshold, given that recovery and detection are known to be equivalent for \pr{pc} \cite{alon2007testing}.

We show that the $k\pr{-pc}$ conjecture implies the \pr{pds} Recovery Conjecture for \textit{semirandom} community recovery in the regime where $q = \Theta(1)$. Specifically, we show that the computational barrier in the detection problem shifts to the recovery threshold when perturbed by a semirandom adversary. This yields the first average-case evidence towards the \pr{pds} recovery conjecture as stated in \cite{chen2016statistical, hajek2015computational, brennan2018reducibility, brennan2019universality}. The reduction showing this lower bound is in Section \ref{sec:semirandom}.

\paragraph{Universality of $k$-to-$k^2$ Gaps in Sparse Mixtures.} 

Several sparse estimation problems exhibit $k$-to-$k^2$ statistical-computational gaps, such as sparse PCA. 
By extending our reductions for the problems above and introducing a new gadget performing an algorithmic change of measure, we also show a universality principle for computational lower bounds at $n = \tilde{\Omega}(k^2)$ samples for learning $k$-sparse mixtures. This implies that all problems with the same \emph{structure} exhibit the same gap, independent of the specifics of the distributions, and is a first step towards explaining the ubiquity of $k$-to-$k^2$ gaps in high-dimensional statistics.

The reduction gadget is a 3-input variant of the 2-input rejection kernel framework introduced in \cite{brennan2018reducibility} and extended in \cite{brennan2019universality}. Its guarantees are general and yield lower bounds that only require high probability bounds on the likelihood ratios of the planted marginals. The universality class of lower bounds we obtain includes the $k$-to-$k^2$ gap in the spiked covariance model of sparse PCA, recovering lower bounds obtained in \cite{berthet2013optimal, berthet2013complexity, wang2016statistical, gao2017sparse, brennan2018reducibility, brennan2019optimal}. It also includes the sparse Gaussian mixture models considered in \cite{azizyan2013minimax, verzelen2017detection, fan2018curse} and shows computational lower bounds for these problems that are tight against algorithmic achievability.

As noted earlier, average-case reductions are notoriously brittle due to the difficulty in exactly mapping to natural target distributions. Universality principles from average-case complexity must strongly overcome this obstacle by mapping precisely to every natural target distribution within an entire universality class. The rejection kernels we introduce provide a simple recipe for obtaining these universality results from an initial reduction to a single well-chosen member of the universality class, in the context of learning sparse mixtures. Our reduction can be found in Section \ref{sec:universality}. 


\paragraph{Techniques.}
To obtain our lower bounds from $k\pr{-pc}$, we introduce several techniques for average-case reductions. Our main approach is simple, but powerful: we rotate Gaussianized forms of the input graph by carefully designed matrices with orthonormal rows. By choosing the right rotation matrices, we can manipulate the planted signal in a precisely controlled manner and map approximately in total variation to our target distributions. The matrices we use in our reductions are constructed from hyperplanes in $\mathbb{F}_r^t$ for certain prime numbers $r$. In order to arrive at a Gaussianized form of the input graph and permit this approach, we need a number of average-case reduction primitives tailored to the $k$-partite promise in $k\pr{-pc}$, as presented in Section \ref{sec:redisgm}. Analyzing the total variation guarantees of these reduction primitives proves to be technically involved. Our techniques are outlined in more detail in Section \ref{subsec:techniques}.

\subsection{Related Work}
\paragraph{Robustness.}
The study of robust estimation under model misspecification began with Huber's contamination model \cite{huber1992robust, huber1965robust} and observations of Tukey \cite{tukey1975mathematics}. Classical robust estimators have typically either been computationally intractable or heuristic \cite{huber2011robust, tukey1975mathematics, yatracos1985rates}. Recently, the first dimension-independent error rates for robust mean estimation were obtained in \cite{diakonikolas2016robust} and a logarithmic dependence was obtained in \cite{lai2016agnostic}. This sparked an active line of research into robust algorithms for high-dimensional problems and has led to algorithms for robust variants of clustering, learning of mixture models, linear regression, estimation of sparse functionals and more \cite{awasthi2014power, diakonikolas2018robustly, klivans2018efficient, diakonikolas2019efficient, li2017robust, balakrishnan2017computationally, charikar2017learning}. Another notion of robustness that has been studied in computer science is against \emph{semirandom} adversaries, who may modify a random instance of a problem in certain constrained ways that heuristically appear to increase the signal strength \cite{blum1995coloring}. For example, in a semirandom variant of planted clique, the adversary takes an instance of planted clique and may only remove edges outside of the clique. Algorithms robust to semirandom adversaries have been exhibited for the stochastic block model \cite{feige2001heuristics}, planted clique \cite{feige2000finding}, unique games \cite{kolla2011play}, correlation clustering \cite{mathieu2010correlation, makarychev2015correlation}, graph partitioning \cite{makarychev2012approximation} and clustering mixtures of Gaussians \cite{vijayaraghavan2018clustering}.

The papers \cite{feige2001heuristics} and \cite{david2016effect}, show lower bounds in semirandom problems from worst-case hardness, and \cite{moitra2016robust} shows that semirandom adversaries can shift the information-theoretic threshold for recovery in the 2-block stochastic block model. Recent work also shows lower bounds from worst-case hardness for the optimal $\ell_2$ error that can be attained by polynomial time algorithms in sub-Gaussian mean estimation without the sparsity constraint \cite{hopkins2019hard}. Our work builds on the framework for average-case reductions introduced in \cite{brennan2018reducibility, brennan2019universality, brennan2019optimal}, making use of several of the average-case primitives in these papers.

\paragraph{Restricted models of computation.} In addition to the aforementioned work showing statistical-computational gaps through average-case reductions from the \pr{pc} conjecture, there is a line of work showing lower bounds in restricted models of computation. Lower bounds in the sum of squares (SOS) hierarchy have been shown for a variety of average-case problems, including planted clique \cite{deshpande2015improved, raghavendra2015tight, hopkins2016integrality, barak2016nearly}, sparse PCA \cite{ma2015sum}, sparse spiked Wigner and tensor PCA \cite{hopkins2017power}, maximizing random tensors \cite{bhattiprolu2017sum} and random CSPs \cite{kothari2017sum}. Tight lower bounds have been shown in the statistical query model for planted clique \cite{feldman2013statistical}, random CSPs \cite{feldman2015complexity} and learning Gaussian graphical models \cite{lu2018edge}. Lower bounds against meta-algorithms based on low-degree polynomials were shown for the $k$-block stochastic block model in \cite{hopkins2017efficient} and for random optimization problems related to PCA in \cite{bandeira2019computational}. 

More closely related to our results are the statistical query lower bounds of \cite{diakonikolas2017statistical} and \cite{fan2018curse}, which show statistical query lower bounds for some of the problems that we consider. Specifically, \cite{diakonikolas2017statistical} shows that statistical query algorithms require $n = \tilde{\Omega}(k^2)$ samples to solve robust sparse mean estimation and \cite{fan2018curse} shows tight statistical query lower bounds for sparse mixtures of Gaussians. We obtain lower bounds for sparse mixtures of Gaussians as an intermediate in reducing to robust sparse mean estimation. It also falls within our universality class for sparse mixtures.

\paragraph{Average-case reductions.} There have been a number of average-case reductions in the literature starting with different assumptions than the \pr{pc} conjecture. Hardness conjectures for random CSPs have been used to show hardness in improper learning complexity \cite{daniely2014average}, learning DNFs \cite{daniely2016complexityDNF} and hardness of approximation \cite{feige2002relations}.

Another related reference is the reduction in \cite{cai2015computational}, which proves a detection-recovery gap in the context of sub-Gaussian submatrix localization based on the hardness of finding a planted $k$-clique in a random $n/2$-regular graph. This degree-regular formulation of $\pr{pc}$ was previously considered in \cite{deshpande2015finding} and differs in a number of ways from \pr{pc}. For example, it is unclear how to generate a sample from the degree-regular variant in polynomial time. We remark that the reduction of \cite{cai2015computational}, when instead applied the usual formulation of \pr{pc} produces a matrix with highly dependent entries. Specifically, the sum of the entries of the output matrix has variance $n^2/\mu$ where $\mu \ll 1$ is the mean parameter for the submatrix localization instance whereas an output matrix with independent entries of unit variance would have a sum of entries of variance $n^2$. Note that, in general, any reduction beginning with $\pr{pc}$ that also preserves the natural $H_0$ hypothesis cannot show the existence of a detection-recovery gap, as any lower bounds for localization would also apply to detection.

\subsection{Outline of the Paper}

In Section \ref{sec:summary}, we formulate the estimation and recovery tasks we consider as detection problems, present our main results and give an overview of our techniques. In Section \ref{sec:avgreductionstv}, we establish our computational model and give some preliminaries on reductions in total variation. In Section \ref{sec:redisgm}, we give our main reduction to the intermediate problem of imbalanced sparse Gaussian mixtures. In Section \ref{sec:robust}, we deduce our lower bounds for robust sparse mean estimation. In Section \ref{sec:semirandom}, we give our reduction to semirandom community recovery. In Section \ref{sec:universality}, we introduce symmetric 3-ary rejection kernels and apply them with our reduction to sparse Gaussian mixtures to produce our universal lower bound. In Section \ref{sec:evidencekpc}, we provide evidence for the $k\pr{-pc}$ conjecture based on low-degree polynomials and lower bounds against statistical query algorithms.

\subsection{Notation}

We adopt the following notation. Let $\mL(X)$ denote the distribution law of a random variable $X$ and given two laws $\mL_1$ and $\mL_2$, let $\mL_1 + \mL_2$ denote $\mL(X + Y)$ where $X \sim \mL_1$ and $Y \sim \mL_2$ are independent. Given a distribution $\mathcal{P}$, let $\mathcal{P}^{\otimes n}$ denote the distribution of $(X_1, X_2, \dots, X_n)$ where the $X_i$ are i.i.d. according to $\mathcal{P}$. Similarly, let $\mathcal{P}^{\otimes m \times n}$ denote the distribution on $\mathbb{R}^{m \times n}$ with i.i.d. entries distributed as $\mathcal{P}$. Given a finite or measurable set $\mathcal{X}$, let $\text{Unif}[\mathcal{X}]$ denote the uniform distribution on $\mathcal{X}$. Let $\TV$, $\KL$ and $\chi^2$ denote total variation distance, KL divergence and $\chi^2$ divergence, respectively. Let $\mN(\mu, \Sigma)$ denote a multivariate normal random vector with mean $\mu \in \mathbb{R}^d$ and covariance matrix $\Sigma$, where $\Sigma$ is a $d \times d$ positive semidefinite matrix. Let $[n] = \{1, 2, \dots, n\}$ and $\mG_n$ be the set of simple graphs on $n$ vertices. Let $\mathbf{1}_S$ denote the vector $v \in \mathbb{R}^n$ with $v_i = 1$ if $i \in S$ and $v_i = 0$ if $i \not \in S$ where $S \subseteq [n]$. Let $\pr{mix}_{\epsilon}(\mD_1, \mD_2)$ denote the $\epsilon$-mixture distribution formed by sampling $\mD_1$ with probability $(1 - \epsilon)$ and $\mD_2$ with probability $\epsilon$.

\section{Problem Formulations and Main Lower Bounds}
\label{sec:summary}

\subsection{Detection Problems, Robustness and Adversaries}

We begin by describing our general setup for detection problems and the notions of robustness and types adversaries that we consider. In a detection task $\mP$, the algorithm is given a set of observations and tasked with distinguishing between two hypotheses:
\begin{itemize}
\item a \emph{uniform} hypothesis $H_0$ corresponding to the natural noise distribution for the problem; and
\item a \emph{planted} hypothesis $H_1$, under which observations are generated from the same noise distribution but with a latent planted structure.
\end{itemize}
Both $H_0$ and $H_1$ can either be simple hypotheses consisting of a single distribution or a composite hypothesis consisting of multiple distributions. Our problems typically are such that either: (1) both $H_0$ and $H_1$ are simple hypotheses; or (2) both $H_0$ and $H_1$ are composite hypotheses consisting of the set of distributions that can be induced by some constrained adversary. The robust estimation literature contains a number of adversaries capturing different notions of model misspecification. We consider the following three central classes of adversaries:
\begin{enumerate}
\item \textbf{$\epsilon$-corruption}: A set of samples $(X_1, X_2, \dots, X_n)$ is an $\epsilon$-corrupted sample from a distribution $\mD$ if they can be generated by giving a set of $n$ samples drawn i.i.d. from $\mD$ to an adversary who then changes at most $\epsilon n$ of them arbitrarily.
\item \textbf{Huber's contamination model}: A set of samples $(X_1, X_2, \dots, X_n)$ is an $\epsilon$-contaminated of $\mD$ in Huber's model if
$$X_1, X_2, \dots, X_n \sim_{\text{i.i.d.}} \pr{mix}_{\epsilon}(\mD, \mD_O)$$
where $\mD_O$ is an unknown outlier distribution chosen by an adversary. Here, $\pr{mix}_{\epsilon}(\mD, \mD_O)$ denotes the $\epsilon$-mixture distribution formed by sampling $\mD$ with probability $(1 - \epsilon)$ and $\mD_O$ with probability $\epsilon$.
\item \textbf{Semirandom adversaries}: Suppose that $\mD$ is a distribution over collections of observations $\{ X_i \}_{i \in I}$ such that an unknown subset $P \subseteq I$ of indices correspond to a planted structure. A sample $\{ X_i \}_{i \in I}$ is semirandom if it can be generated by giving a sample from $\mD$ to an adversary who is allowed decrease $X_i$ for any $i \in I \backslash P$. Some formulations of semirandom adversaries in the literature also permit increases in $X_i$ for any $i \in P$. Our lower bounds apply to both adversarial setups.
\end{enumerate}
All adversaries in these models of robustness are computationally unbounded and have access to randomness. Given a single distribution $\mD$ over a set $X$, any one of these three adversaries produces a set of distributions $\pr{adv}(\mD)$ that can be obtained after corruption. When formulated as detection problems, the hypotheses $H_0$ and $H_1$ are of the form $\pr{adv}(\mD)$ for some $\mD$. We remark that $\epsilon$-corruption can simulate contamination in Huber's model at a slightly smaller $\epsilon'$ within $o(1)$ total variation. This is because a sample from Huber's model has $\text{Bin}(n, \epsilon')$ samples from $\mD_O$. An adversary resampling $\min\{\text{Bin}(n, \epsilon'), \epsilon n\}$ samples from $\mD_O$ therefore simulates Huber's model within a total variation distance bounded by standard concentration for the Binomial distribution.

As discussed in \cite{brennan2018reducibility} and \cite{hajek2015computational}, when detection problems need not be composite by definition, average-case reductions to natural simple vs. simple hypothesis testing formulations are stronger and technically more difficult. In these cases, composite hypotheses typically arise because a reduction gadget precludes mapping to the natural simple vs. simple hypothesis testing formulation. We remark that simple vs. simple formulations are the hypothesis testing problems that correspond to average-case decision problems $(L, \mathcal{D})$ as in Levin's theory of average-case complexity. A survey of average-case complexity can be found in \cite{bogdanov2006average}.

Given an observation $X$, an algorithm $\mathcal{A}(X) \in \{0, 1\}$ \emph{solves} the detection problem \emph{with nontrivial probability} if there is an $\epsilon > 0$ such that its worst Type I$+$II error over the hypotheses $H_0$ and $H_1$ satisfies that
$$\limsup_{n \to \infty} \left( \sup_{P \in H_0} \bP_{X \sim P}[\mathcal{A}(X) = 1] + \sup_{P \in H_1} \bP_{X \sim P}[\mathcal{A}(X) = 0] \right) \le 1 - \epsilon$$
where $n$ is the parameter indicating the size of $X$. We refer to this quantity as the asymptotic Type I$+$II error of $\mathcal{A}$ for the problem $\mP$. If the asymptotic Type I$+$II error of $\mathcal{A}$ is zero, then we say $\mathcal{A}$ \emph{solves} the detection problem $\mP$. A simple consequence of this definition is that if $\mathcal{A}$ achieves asymptotic Type I$+$II error $1 - \epsilon$ for a composite testing problem with hypotheses $H_0$ and $H_1$, then it also achieves this same error on the simple problem with hypotheses $H_0$ and $H_1' : X \sim P$ where $P$ is any mixture of the distributions in $H_1$. When stating detection problems, we adopt the convention that any parameters such as $n$ implicitly refers to a sequence $n = (n_t)$. For notational convenience, we drop the index $t$. The asymptotic Type I$+$II error of a test for a parameterized detection problem is defined as $t \to \infty$.

\subsection{Formulations as Detection Problems}

In this section, we formulate robust sparse mean estimation, semirandom community recovery and our general learning sparse mixtures as detection problems. More precisely, for each problem $\mP$ that we consider, we introduce a detection variant $\mP'$ such that a blackbox for $\mP$ also solves $\mP'$.

\paragraph{Robust Sparse Mean Estimation.} In robust sparse mean estimation, the observed vectors $X_1, X_2, \dots, X_n$ are an $\epsilon$-corrupted set of $n$ samples from $\mN(\mu, I_d)$ where $\mu$ is an unknown $k$-sparse vector in $\mathbb{R}^d$. The task is to estimate $\mu$ in the $\ell_2$ norm by outputting $\hat{\mu}$ with $\| \hat{\mu} - \mu \|_2$ small. Without $\epsilon$-corruption, the information-theoretically optimal number of samples is $n = \Theta(k \log d/\rho^2)$ in order to estimate $\mu$ within $\ell_2$ distance $\rho$, and this is efficiently achieved by truncating the empirical mean. As discussed in \cite{li2017robust, balakrishnan2017computationally}, for $\| \mu - \mu' \|_2$ sufficiently small, it holds that $\TV\left( \mN(\mu, I_d), \mN(\mu', I_d) \right) = \Theta(\| \mu - \mu' \|_2)$. Furthermore, an $\epsilon$-corrupted set of samples can simulate distributions within $O(\epsilon)$ total variation from $\mN(\mu, I_d)$. Therefore $\epsilon$-corruption can simulate $\mN(\mu', I_d)$ if $\|\mu' - \mu\|_2 = O(\epsilon)$ and it is impossible to estimate $\mu$ with $\ell_2$ distance less than this $O(\epsilon)$.

This implies that the minimax rate of estimation for $\mu$ is $O(\epsilon)$, even for very large $n$. As shown in \cite{li2017robust, balakrishnan2017computationally}, the information-theoretic threshold for estimating at this rate in the $\epsilon$-corrupted model remains at $n = \Theta(k \log d/\epsilon^2)$ samples. However, the best known polynomial-time algorithms from \cite{li2017robust, balakrishnan2017computationally} require $n = \tilde{\Theta}(k^2 \log d/\epsilon^2)$ samples to estimate $\mu$ within $\epsilon \sqrt{\log \epsilon^{-1}}$ in $\ell_2$. One of our main results is to show that this $k$-to-$k^2$ statistical-computational gap induced by robustness follows from the $k$-partite planted clique conjecture. We show this by giving an average-case reduction to the following detection formulation of robust sparse mean estimation.

\begin{definition}[Detection Formulation of Robust Sparse Mean Estimation]
For any $\tau = \omega(\epsilon)$, the hypothesis testing problem $\pr{rsme}(n, k, d, \tau, \epsilon)$ has hypotheses
\begin{align*}
H_0 : (X_1, X_2, \dots, X_n) &\sim_{\textnormal{i.i.d.}} \mN(0, I_d) \\
H_1 : (X_1, X_2, \dots, X_n) &\sim_{\textnormal{i.i.d.}} \pr{mix}_{\epsilon}\left( \mN(\mu_R, I_d), \mD_O \right)
\end{align*}
where $\mD_O$ is any adversarially chosen outlier distribution on $\mathbb{R}^d$ and where $\mu_R \in \mathbb{R}^d$ is any random $k$-sparse vector satisfying $\| \mu_R \|_2 \ge \tau$ holds almost surely.
\end{definition}

This is a formulation of robust sparse mean estimation in Huber's contamination model, and therefore lower bounds for this problem imply corresponding lower bounds under $\epsilon$-corruption. We also directly consider a detection variant $\pr{c-rsme}$ formulated in $\epsilon$-corruption model. The condition that $\tau = \omega(\epsilon)$ ensures that any algorithm
achieving $\ell_2$ error $\| \mu - \mu' \|_2$ at the minimax rate of $O(\epsilon)$ can also solve the detection problem. Our lower bounds also apply to this detection problem with the 
requirement that $\tau = \omega(r)$ for $r$ much larger than $\epsilon$, up to approximately $\tilde{\Theta}(\sqrt{\epsilon})$. Lower bounds in this case show gaps for estimators achieving a suboptimal rate of estimation that is only $O(r)$ instead of $O(\epsilon)$.

\paragraph{Planted Dense Subgraph and Community Recovery.}
In single community recovery, one observes a graph $G$ drawn from the planted dense subgraph (\pr{pds}) distribution $\mG(n, k, p, q)$ with edge densities $p = p(n) > q = q(n)$ that can vary with $n$. To sample from $\mG(n, k, p, q)$, first $G$ is sampled from $\mG(n, q)$ and then a $k$-subset of $[n]$ is chosen uniformly at random and the induced subgraph of $G$ on $S$ is replaced with an independently sampled copy of $\mG(k, p)$. The regime of interest is when $p$ and $q$ are converging to one another at some rate and satisfy that $p/q = O(1)$. The task is to recover the latent index set $S$, either exactly by outputting $\hat{S} = S$ or partially by outputting $\hat{S}$ with $|\hat{S} \cap S| = \Omega(k)$. As shown in \cite{chen2016statistical}, a polynomial-time convexified maximum likelihood algorithm achieves exact recovery when
$$\pr{snr} = \frac{(p - q)^2}{q(1 - q)} \gtrsim \frac{n}{k^2}$$
This is the best known polynomial time algorithm and has a statistical-computational gap from the information-theoretic limit of $\pr{snr} = \tilde{\Theta}(1/k)$.

\newcommand{\colorImp}{gray}
\newcommand{\colorEasy}{green}
\newcommand{\colorHard}{blue}

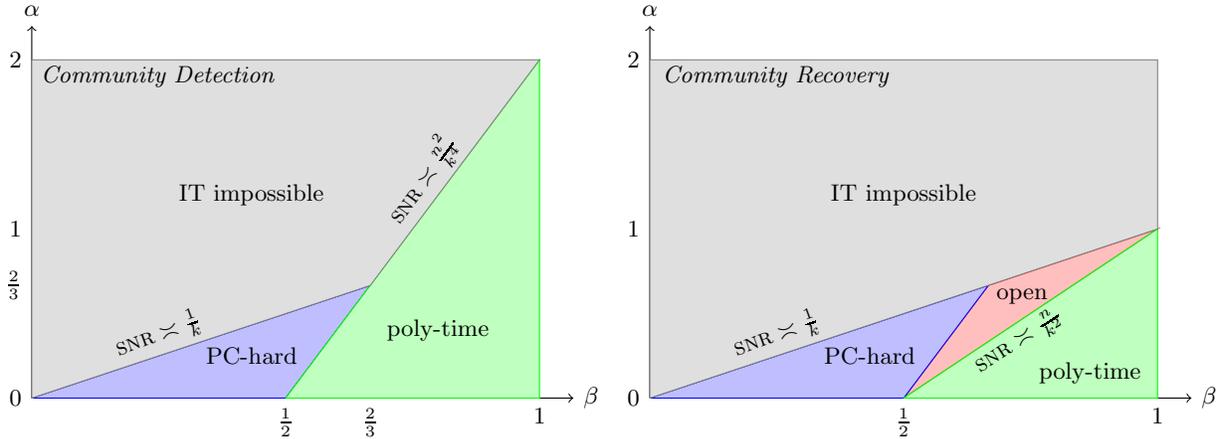
\begin{figure*}[t!]
\centering
\begin{tikzpicture}[scale=0.45]
\tikzstyle{every node}=[font=\footnotesize]
\def\xmin{0}
\def\xmax{16}
\def\ymin{0}
\def\ymax{11}

\draw[->] (\xmin,\ymin) -- (\xmax,\ymin) node[right] {$\beta$};
\draw[->] (\xmin,\ymin) -- (\xmin,\ymax) node[above] {$\alpha$};

\node at (15, 0) [below] {$1$};
\node at (7.5, 0) [below] {$\frac{1}{2}$};
\node at (0, 0) [left] {$0$};
\node at (0, 10) [left] {$2$};
\node at (0, 5) [left] {$1$};
\node at (0, 3.33) [left] {$\frac{2}{3}$};
\node at (10, 0) [below] {$\frac{2}{3}$};

\filldraw[fill=\colorHard!25, draw=\colorHard] (0, 0) -- (7.5, 0) -- (10, 3.33) -- (0, 0);
\filldraw[fill=\colorEasy!25, draw=\colorEasy] (7.5, 0) -- (15, 10) -- (15, 0) -- (7.5, 0);
\filldraw[fill=\colorImp!25, draw=\colorImp] (0, 0) -- (10, 3.33) -- (15, 10) -- (0, 10) -- (0, 0);

\node at (3.75, 9.5) {\textit{Community Detection}};
\node at (12.2, 6)[rotate=54, anchor=south] {$\pr{snr} \asymp \frac{n^2}{k^4}$};
\node at (4, 1.2)[rotate=20, anchor=south] {$\pr{snr} \asymp \frac{1}{k}$};
\node at (6.5, 6) {IT impossible};
\node at (12, 2) {poly-time};
\node at (6.5, 1.25) {PC-hard};
\end{tikzpicture}
\begin{tikzpicture}[scale=0.45]
\tikzstyle{every node}=[font=\footnotesize]
\def\xmin{0}
\def\xmax{16}
\def\ymin{0}
\def\ymax{11}

\draw[->] (\xmin,\ymin) -- (\xmax,\ymin) node[right] {$\beta$};
\draw[->] (\xmin,\ymin) -- (\xmin,\ymax) node[above] {$\alpha$};

\node at (15, 0) [below] {$1$};
\node at (7.5, 0) [below] {$\frac{1}{2}$};
\node at (0, 0) [left] {$0$};
\node at (0, 10) [left] {$2$};
\node at (0, 5) [left] {$1$};

\filldraw[fill=red!25, draw=red] (7.5, 0) -- (15, 5) -- (10, 3.33) -- (7.5, 0);
\filldraw[fill=\colorHard!25, draw=\colorHard] (0, 0) -- (7.5, 0) -- (10, 3.33) -- (0, 0);
\filldraw[fill=\colorEasy!25, draw=\colorEasy] (7.5, 0) -- (15, 5) -- (15, 0) -- (7.5, 0);
\filldraw[fill=\colorImp!25, draw=\colorImp] (0, 0) -- (15, 5) -- (15, 10) -- (0, 10) -- (0, 0);

\node at (3.75, 9.5) {\textit{Community Recovery}};
\node at (11.3, 1.1)[rotate=33, anchor=south] {$\pr{snr} \asymp \frac{n}{k^2}$};
\node at (4, 1.2)[rotate=20, anchor=south] {$\pr{snr} \asymp \frac{1}{k}$};
\node at (7.5, 6) {IT impossible};
\node at (13, 0.75) {poly-time};
\node at (6.5, 1.25) {PC-hard};
\node at (11, 3) {open};
\end{tikzpicture}

\caption{Prior computational and statistical barriers in the detection and recovery of a single hidden community from the \pr{pc} conjecture \cite{hajek2015computational, brennan2018reducibility, brennan2019universality}. The axes are parameterized by $\alpha$ and $\beta$ where $\pr{snr} = \frac{(p - q)^2}{q(1 - q)} = \tilde{\Theta}(n^{-\alpha})$ and $k = \tilde{\Theta}(n^\beta)$. The red region is conjectured to be computationally hard but no $\textsc{pc}$ reductions showing this hardness are known.}
\label{fig:pdsdetrecgap}
\end{figure*}

The natural detection variant of this problem is planted dense subgraph with varying $p$ and $q$, which has the hypotheses
$$H_0: G \sim \mG(n, q) \quad \textnormal{and} \quad H_1 : G \sim \mG(n, k, p, q)$$
Unlike the recovery problem, this detection problem can be solved at the lower threshold of $\pr{snr} \gtrsim n^2/k^4$ by thresholding the total number of edges in the observed graph. No polynomial time algorithm beating this threshold by a $\text{poly}(n, k)$ factor is known. The information-theoretic threshold for detection is $\pr{snr} = \tilde{\Theta}(\min\{1/k, n^2/k^4 \})$ and thus the problem has no statistical-computational gap when $k \gtrsim n^{2/3}$. Average-case lower bounds based on planted clique at the conjectured computational threshold of $\pr{snr} \gtrsim n^2/k^4$ were shown in the regime $p = cq = \Theta(n^{-\alpha})$ for some constant $c > 1$ in \cite{hajek2015computational} and generally for $p$ and $q$ with $p/q = O(1)$ in \cite{brennan2018reducibility, brennan2019universality}. Conjectured phase diagrams for the detection and recovery variants are shown in Figure \ref{fig:pdsdetrecgap}.

This pair of problems for recovery and detection of a single community has emerged as a canonical example of a problem with a detection-recovery computational gap, which is significant in all regimes of $q$ as long as $p/q = O(1)$. \cite{hajek2015computational, brennan2018reducibility} and \cite{brennan2019universality} all conjecture that the recovery problem is hard below the threshold of $\pr{snr} = \tilde{\Theta}(n/k^2)$. This \pr{pds} Recovery Conjecture was even used in \cite{brennan2018reducibility} as a hardness assumption to show detection-recovery gaps in other problems including biased sparse PCA and Gaussian biclustering. It is unlikely that the \pr{pds} recovery conjecture can be shown to follow from the $\pr{pc}$ conjecture, since it is a very different problem from detection which does have tight $\pr{pc}$ hardness. In fact, a reduction in total variation from $\pr{pc}$ to $\pr{pds}$ at the recovery threshold that faithfully maps corresponding hypotheses $H_i$ is impossible given the $\pr{pc}$ conjecture, since the \pr{pds} detection problem is easy for $n/k^2 \gg \pr{snr} \gtrsim n^2/k^4$.

We show the intriguing result that the \pr{pds} Recovery Conjecture is true for \textit{semirandom} community recovery in the regime where $q = \Theta(1)$ given the $k\pr{-pc}$ conjecture. To do this, we give an average-case reduction to the following semirandom detection formulation of community recovery.

\begin{definition}[Detection Formulation of Semirandom Community Recovery]
The hypothesis testing problem $\pr{semi-cr}(n, k, p, q)$ has observation $G \in \mG_n$ and hypotheses
\begin{align*}
&H_0 : G \sim P \quad \textnormal{for some } P \in \pr{adv}\left(\mG(n, q)\right) \\
&H_1 : G \sim P \quad \textnormal{for some } P \in \pr{adv}\left(\mG(n, k, p, q)\right)
\end{align*}
where $\pr{adv}\left(\mG(n, k, p, q)\right)$ is the set of distributions induced by a semirandom adversary that can only remove edges outside of the planted dense subgraph $S$. In the non-planted case, the set $\pr{adv}\left(\mG(n, q)\right)$ corresponds to an adversary that can remove any edges.
\end{definition}

An algorithm $\mathcal{A}$ solving semirandom community recovery exactly can threshold the edge density within the output set of vertices $\hat{S}$. The semirandomness of the adversary along with concentration bounds ensures that this solves $\pr{semi-cr}$. 
We remark that the convexified maximum likelihood algorithm from \cite{chen2016statistical} continues to solve the community recovery problem at the same threshold under a semirandom adversary by a simple monotonicity argument.

\paragraph{Learning Sparse Mixtures.} 
Our third main result is a universality principle for statistical-computational gaps in learning sparse mixtures. The sparse mixture setup we consider includes sparse PCA in the spiked covariance model, learning sparse mixtures of Gaussians, sparse group testing and distributions related to learning graphical models. To define the detection formulation of our generalized sparse mixtures problem, we will need the notion of computable pairs from \cite{brennan2019universality}. For brevity in this section, we defer this definition until Section \ref{subsec:universalitybounds}. Our general sparse mixtures detection problem is the following simple vs. simple hypothesis testing problem.

\begin{definition}[Generalized Learning Sparse Mixtures]
Let $n$ be a parameter and $\mD$ be a symmetric distribution on a subset of $\mathbb{R}$. Suppose that $\{\mP_{\mu}\}_{\mu \in \mathbb{R}}$ and $\mQ$ are distributions on a measurable space $(W, \mathcal{B})$ such that $(\mP_{\mu}, \mQ)$ is a computable pair for each $\mu \in \mathbb{R}$. Let the general sparse mixture $\pr{glsm}_{H_1}(n, k, d, \{\mP_{\mu}\}_{\mu \in \mathbb{R}}, \mQ, \mD)$ be the distribution on $X_1, X_2, \dots, X_n \in W^d$ sampled as follows:
\begin{enumerate}
\item Sample a single $k$-subset $S \subseteq [d]$ uniformly at random; and
\item For each $i \in [n]$, choose some $\mu \sim \mD$ and independently sample $(X_i)_j \sim \mP_\mu$ if $j \in S$ and $(X_i)_j \sim \mQ$ otherwise.
\end{enumerate}
The problem $\pr{glsm}(n, k, d, \{\mP_{\mu}\}_{\mu \in \mathbb{R}}, \mQ, \mD)$ has observations $X_1, X_2, \dots, X_n$ and hypotheses
\begin{align*}
&H_0 : X_1, X_2, \dots, X_n \sim_{\textnormal{i.i.d.}} \mQ^{\otimes d} \\
&H_1 : (X_1, X_2, \dots, X_n) \sim \pr{glsm}_{H_1}\left(n, k, d, \{\mP_{\mu}\}_{\mu \in \mathbb{R}}, \mQ, \mD\right)
\end{align*}
\end{definition}

\subsection{Main Average-Case Lower Bounds}

We now state our main average-case lower bounds for robust sparse mean estimation, semirandom community recovery and general learning of sparse mixtures. We first formally introduce the $k\pr{-pc}$ conjecture. Our reductions will all start with the more general problem of $k$-partite planted dense subgraph with constant edge densities, defined as follows. For simplicity of analysis we will always assume that $p$ and $q$ are fixed constants, however we remark that our reductions can handle the more general case where only $q$ is constant and $p$ is tending towards $q$. 

\begin{definition}[$k$-partite Planted Dense Subgraph]
Let $k$ divide $n$ and $E$ be known a partition of $[n]$ with $[n] = E_1 \cup E_2 \cup \cdots \cup E_k$ and $|E_i| = n/k$ for each $i$. Let $\mG_E(n, k, p, q)$ denote the distribution on $n$-vertex graphs formed by sampling $G \sim \mG(n, q)$ and planting an independently sampled copy of $H \sim G(k, p)$ on the $k$-vertex subgraph with one vertex chosen uniformly at random from each $E_i$ in $G$. The hypothesis testing problem $k\pr{-pds}(n, k, p, q)$ has hypotheses
$$H_0: G \sim \mG(n, q) \quad \textnormal{and} \quad H_1 : G \sim \mG_E(n, k, p, q)$$
\end{definition}

The $k$-partite variant of planted clique $k\pr{-pc}(n, k, p)$ is then $k\pr{-pds}(n, k, 1, p)$ in this notation. Note that the edges within each $E_i$ are irrelevant and independent of the hidden vertex set of the clique. We remark that $E$ can be any fixed partition of $[n]$ without changing the problem. The $k\pr{-pc}$ conjecture is formally stated as follows.

\begin{conjecture}[$k\pr{-pc}$ Conjecture]
Fix some constant $p \in (0, 1)$. Suppose that $\{ \mathcal{A}_{t} \}$ is a sequence of randomized polynomial time algorithms $\mathcal{A}_t : \mG_{n_t} \to \{0, 1\}$ where $n_t$ and $k_t$ are increasing sequences of positive integers satisfying that $k_t = o(\sqrt{n_t})$ and $k_t$ divides $n_t$ for each $t$. Then if $G_t$ is an instance of $k\pr{-pc}(n_t, k_t, p)$, it holds that
$$\liminf_{t \to \infty} \left( \bP_{H_0}\left[\mathcal{A}_t(G) = 1\right] + \bP_{H_1}\left[\mathcal{A}_t(G) = 0\right] \right) \ge 1$$ 
\end{conjecture}

The $k\pr{-pds}$ conjecture at a fixed pair of constant edge densities $0 < q < p \le 1$ can be formulated analogously and generalizes the $k\pr{-pc}$ conjecture. There is a plethora of evidence in the literature for the ordinary $\pr{pc}$ conjecture. Spectral algorithms, approximate message passing, semidefinite programming, nuclear norm minimization and several other polynomial-time combinatorial approaches all appear to fail to solve $\pr{pc}$ exactly when $k = o(n^{1/2})$ \cite{alon1998finding, feige2000finding, mcsherry2001spectral, feige2010finding, ames2011nuclear, dekel2014finding, deshpande2015finding, chen2016statistical}. Lower bounds against low-degree sum of squares relaxations \cite{barak2016nearly} and statistical query algorithms \cite{feldman2013statistical} have also been shown up to $k = o(n^{1/2})$. In Section \ref{sec:evidencekpc}, we show that lower bounds for ordinary $\pr{pc}$ against low-degree polynomials and statistical query algorithms extend easily to $k\pr{-pc}$.

We now state our main lower bounds based on either the $k\pr{-pc}$ conjecture or the $k\pr{-pds}$ conjecture at a fixed pair of constant edge densities $0 < q < p \le 1$. These theorem statements are reproduced in Sections \ref{sec:robust}, \ref{sec:semirandom} and \ref{sec:universality}, respectively. We first give our general lower bound for \pr{rsme}, which also applies to weak algorithms only able to estimate up to $\ell_2$ rates of approximately $\sqrt{\epsilon}$.

\begin{theorem}[General Lower Bound for $\pr{rsme}$]
Let $(n, k, d, \epsilon)$ be any parameters satisfying that $k^2 = o(d)$, $\epsilon \in (0, 1)$ and $n$ satisfies that $n = o(\epsilon^3 k^2)$ and $n = \Omega(k)$. If $c > 0$ is some fixed constant, then there is a parameter $\tau = \Omega(\sqrt{\epsilon/(\log n)^{1 + c}})$ such that any randomized polynomial time test for $\pr{rsme}(n, k, d, \tau, \epsilon)$ has asymptotic Type I$+$II error at least 1 assuming either the $k$\pr{-pc} conjecture or the $k$\pr{-pds} conjecture for some fixed edge densities $0 < q < p \le 1$.
\end{theorem}

Specializing this theorem to the case where $\epsilon = 1/\text{polylog}(n)$, we establish the tight $k$-to-$k^2$ gap for $\pr{rsme}$ up to $\text{polylog}(n)$ factors, as stated in the following corollary.

\begin{corollary}[Tight $k$-to-$k^2$ gap for $\pr{rsme}$]
Let $(n, k, d, \epsilon)$ be any parameters satisfying that $k^2 = o(d)$, $\epsilon = \Theta((\log n)^{-c})$ for some constant $c > 1$ and $n$ satisfies that $n = o(k^2 (\log n)^{-3c})$ and $n = \Omega(k)$. Then there is a parameter $\tau = \omega(\epsilon)$ such that any randomized polynomial time test for $\pr{rsme}(n, k, d, \tau, \epsilon)$ has asymptotic Type I$+$II error at least 1 assuming either the $k$\pr{-pc} conjecture or the $k$\pr{-pds} conjecture for some fixed edge densities $0 < q < p \le 1$.
\end{corollary}

In Section \ref{sec:robust}, we also show how to alleviate the dependence of our general sample complexity lower bound of $n = \Omega(\epsilon^3 k^2)$ on $\epsilon$ as a tradeoff with how far $\tau$ is above the minimax rate of $O(\epsilon)$. We now state our lower bounds for $\pr{semi-cr}$ and universality principle for $\pr{glsm}$.

\begin{theorem}[The \pr{pds} Recovery Conjecture holds for $\pr{semi-cr}$]
Fix any constant $\beta \in [1/2, 1)$. Suppose that $\mathcal{B}$ is a randomized polynomial time test for $\pr{semi-cr}(n, k, p, q)$ for all $(n, k, p, q)$ with $k = \Theta(n^\beta)$ and
$$\frac{(p - q)^2}{q(1 - q)} \leq \overline{\nu} \quad \textnormal{and} \quad \min\{q, 1 - q \} = \Omega(1) \quad \textnormal{where} \quad \overline{\nu} = o\left(\frac{n}{k^2 \log n}\right)$$
Then $\mathcal{B}$ has asymptotic Type I$+$II error at least 1 assuming either the $k$\pr{-pc} conjecture or the $k$\pr{-pds} conjecture for some fixed edge densities $0 < q' < p' \le 1$.
\end{theorem}

\begin{theorem}[Universality of $n = \tilde{\Omega}(k^2)$ for \pr{glsm}]
Let $(n, k, d)$ be parameters such that $n = o(k^2)$ and $k^2 = o(d)$. Suppose that $(\mD, \mQ, \{ \mP_{\nu} \}_{\nu \in \mathbb{R}})$ satisfy
\begin{itemize}
\item $\mD$ is a symmetric distribution about zero and $\bP_{\nu \sim \mD}[\nu \in [-1, 1]] = 1 - o(n^{-1})$; and
\item for all $\nu \in [-1, 1]$, it holds that
$$\frac{1}{\sqrt{k \log n}} \gg \left| \frac{d\mP_{\nu}}{d\mQ} (x) - \frac{d\mP_{-\nu}}{d\mQ} (x) \right| \quad \textnormal{and} \quad \frac{1}{k \log n} \gg \left|\frac{d\mP_{\nu}}{d\mQ} (x) + \frac{d\mP_{-\nu}}{d\mQ} (x) - 2 \right|$$
with probability at least $1 - o(n^{-3} d^{-1})$ over each of $\mP_{\nu}, \mP_{-\nu}$ and $\mQ$.
\end{itemize}
Then if $\mathcal{B}$ is a randomized polynomial time test for $\pr{glsm}\left(n, k, d, \{ \mP_{\nu} \}_{\nu \in \mathbb{R}}, \mQ, \mD \right)$. Then $\mathcal{B}$ has asymptotic Type I$+$II error at least 1 assuming either the $k$\pr{-pc} conjecture or the $k$\pr{-pds} conjecture for some fixed edge densities $0 < q < p \le 1$.
\end{theorem}

Note that $\mD$ and the indices of $\mP_{\nu}$ can be reparameterized without changing the underlying problem. The assumption that $\mD$ is symmetric and mostly supported on $[-1, 1]$ is for convenience. When the likelihood ratios are relatively concentrated the dependence of the two conditions above on $n$ and $d$ is nearly negligible. These conditions almost exclusively depend on $k$, implying that they will not require a stronger dependence between $n$ and $k$ to produce hardness than the $n = \tilde{o}(k^2)$ condition that arises from our reductions. Thus these conditions do show a universality principle for the computational sample complexity of $n = \tilde{\Theta}(k^2)$. 

\subsection{Overview of Techniques}
\label{subsec:techniques}

We now overview the average-case reduction techniques we introduce to show our lower bounds.

\paragraph{Rotating Gaussianized Instances by $H_{r, t}$.} We give a simplified overview of our rotation method in the case of robust sparse mean estimation. Observe that if $M = \tau \cdot \mathbf{1}_S \mathbf{1}_S^\top + \mN(0, 1)^{\otimes n \times n}$, $S$ is a $k$-subset and $R$ is an orthogonal matrix then $MR^\top$ is distributed as $\tau \cdot \mathbf{1}_S (R\mathbf{1}_S)^\top + \mN(0, 1)^{\otimes n \times n}$ since the rows of the noise term have i.i.d. $\mN(0, 1)$ entries and thus are isotropic. This simple property allows us to manipulate the latent structure $S$ upon a rotation by $R$ while at the same time, crucially preserve the independence among the entries of the noise distribution. Our key insight is to carefully construct $R$ based on the geometry of $\mathbb{F}_r^t$. Let $P_1, P_2, \dots, P_{r^t}$ be an enumeration of the points in $\mathbb{F}_r^t$ and $V_1, V_2, \dots, V_\ell$, where $\ell = \frac{r^t - 1}{r - 1}$, be an enumeration of the hyperplanes in $\mathbb{F}_r^t$. Now take $H_{r, t}$ to be the $\ell \times r^t$ matrix
$$(H_{r, t})_{ij} = \frac{1}{\sqrt{r^t(r - 1)}} \cdot \left\{ \begin{matrix} 1 & \text{if } P_j \not \in V_i \\ 1 - r & \text{if } P_j \in V_i \end{matrix} \right.$$
It can be verified that the rows of $H_{r, t}$ are an orthonormal basis of the subspace they span. This choice of $H_{r, t}$ has two properties crucial to our reductions: (1) it contains exactly two values; and (2) each of its columns, other than the column corresponding to $P_i = 0$, approximately contains a $1 - r^{-1}$ vs. $r^{-1}$ mixture of these two values. Taking $R$ to rotate by $H_{r, t}$ on blocks of indices of $M$, each containing exactly one element of $S$, then produces an instance of robust sparse mean estimation with $\epsilon \approx r^{-1}$ and planted mean vector approximately given by $\tau \cdot \mathbf{1}_S / \sqrt{r^t(r - 1)}$. A proper choice of parameters causes this to approximately be $\mu = \mathbf{1}_S \cdot \sqrt{\epsilon/wk \log n}$ for a slow-growing function $w$ and thus $\| \mu \|_2 = \sqrt{\epsilon/\log n}$, which for polylogarithmically small $\epsilon$ is larger than the $O(\epsilon)$ minimax rate. Carrying out this strategy involves a number of technical obstacles related to the counts of each type of entry, the fact that the $H_{r, t}$ are not square, dealing with the column corresponding to zero and other distributional issues that arise.

In our reduction to semirandom community recovery we rotate along both rows and columns and in smaller blocks, and also only require $r = 3$. Our universality result uses exactly the reduction described above for $r = 2$ as a subroutine.

\paragraph{$k$-Partite Average-Case Primitives.} The discussion above overlooks several issues, notably including: (1) mapping from the input graph problem to a matrix of the form $\tau \cdot \mathbf{1}_S \mathbf{1}_S^\top + \mN(0, 1)^{\otimes n \times n}$; and (2) ensuring that we can rotate by $H_{r, t}$ on blocks, each containing exactly one element of $S$. Property (2) is essential to mapping in distribution to an instance of robust sparse mean estimation. Achieving (1) involves more than an algorithmic change of measure -- the initial adjacency matrix of the graph is symmetric and lacks diagonal entries. Symmetry can be broken using a cloning gadget, but planting the diagonal entries so that they are properly distributed is harder. In particular, doing this while maintaining the $k$-partite promise in order to achieve (2) requires embedding in a larger submatrix and an involved analysis of the distances between product distributions with binomial marginals. We break this first Gaussianization step into a number of primitives, extending the framework introduced in \cite{brennan2018reducibility} and \cite{brennan2019universality}.

\paragraph{3-ary Symmetric Rejection Kernels.} To show our universality result, we need to carry out an algorithmic change of measure with three target distributions. To do this, we introduce a new more involved variant of the rejection kernels from \cite{brennan2018reducibility} and \cite{brennan2019universality} in order to map from three inputs to three outputs. The analysis of these rejection kernels is made possible due to symmetries in the initial distributions that crucially rely on using the rotations reduction described above as a subroutine. We remark that the fact that $H_{r, t}$ contains exactly two distinct values turns out to be especially important here to ensure that there are only three possible input distributions to these rejection kernels.

\section{Average-Case Reductions in Total Variation}
\label{sec:avgreductionstv}

\subsection{Reductions in Total Variation and the Computational Model}

We give approximate reductions in total variation to show that lower bounds for one hypothesis testing problem imply lower bounds for another. These reductions yield an exact correspondence between the asymptotic Type I$+$II errors of the two problems. This is formalized in the following lemma, which is Lemma 3.1 from \cite{brennan2018reducibility} stated in terms of composite hypotheses $H_0$ and $H_1$. The main quantity in the statement of the lemma can be interpreted as the smallest total variation distance between the reduced object $\mathcal{A}(X)$ and the closest mixture of distributions from either $H_0'$ or $H_1'$. The proof of this lemma is short and follows from the definition of total variation.

\begin{lemma}[Lemma 3.1 in \cite{brennan2018reducibility}] \label{lem:3a}
Let $\mP$ and $\mP'$ be detection problems with hypotheses $H_0, H_1$ and  $H_0', H_1'$, respectively. Let $X$ be an instance of $\mathcal{P}$ and let $Y$ be an instance of $\mP'$. Suppose there is a polynomial time computable map $\mathcal{A}$ satisfying
$$\sup_{P \in H_0} \inf_{\pi \in \Delta(H_0')} \TV\left( \mL_{P}(\mathcal{A}(X)), \bE_{P' \sim \pi} \, \mL_{P'}(Y) \right) + \sup_{P \in H_1} \inf_{\pi \in \Delta(H_1')} \TV\left( \mL_{P}(\mathcal{A}(X)), \bE_{P' \sim \pi} \, \mL_{P'}(Y) \right) \le \delta$$
If there is a randomized polynomial time algorithm solving $\mP'$ with Type I$+$II error at most $\epsilon$, then there is a randomized polynomial time algorithm solving $\mP$ with Type I$+$II error at most $\epsilon + \delta$.
\end{lemma}

If $\delta = o(1)$, then given a blackbox solver $\mathcal{B}$ for $\mathcal{P}'_D$, the algorithm that applies $\mathcal{A}$ and then $\mathcal{B}$ solves $\mathcal{P}_D$ and requires only a single query to the blackbox. We now outline the computational model and conventions we adopt throughout this paper. An algorithm that runs in randomized polynomial time refers to one that has access to $\text{poly}(n)$ independent random bits and must run in $\text{poly}(n)$ time where $n$ is the size of the instance of the problem. For clarity of exposition, in our reductions we assume that explicit expressions can be exactly computed and that we can sample a biased random bit $\text{Bern}(p)$ in polynomial time. We also assume that the oracles described in Definition \ref{def:computable} can be computed in $\text{poly}(n)$ time. For simplicity of exposition, we assume that we can sample $\mN(0, 1)$ in $\text{poly}(n)$ time. 

\subsection{Properties of Total Variation}

The analysis of our reductions will make use of the following well-known facts and inequalities concerning total variation distance.

\begin{fact} \label{tvfacts}
The distance $\TV$ satisfies the following properties:
\begin{enumerate}
\item (Tensorization) Let $P_1, P_2, \dots, P_n$ and $Q_1, Q_2, \dots, Q_n$ be distributions on a measurable space $(\mathcal{X}, \mathcal{B})$. Then
$$\TV\left( \prod_{i = 1}^n P_i, \prod_{i = 1}^n Q_i \right) \le \sum_{i = 1}^n \TV\left( P_i, Q_i \right)$$
\item (Conditioning on an Event) For any distribution $P$ on a measurable space $(\mathcal{X}, \mathcal{B})$ and event $A \in \mathcal{B}$, it holds that
$$\TV\left( P(\cdot | A), P \right) = 1 - P(A)$$
\item (Conditioning on a Random Variable) For any two pairs of random variables $(X, Y)$ and $(X', Y')$ each taking values in a measurable space $(\mathcal{X}, \mathcal{B})$, it holds that
$$\TV\left( \mL(X), \mL(X') \right) \le \TV\left( \mL(Y), \mL(Y') \right) + \bE_{y \sim Y} \left[ \TV\left( \mL(X | Y = y), \mL(X' | Y' = y) \right)\right]$$
where we define $\TV\left( \mL(X | Y = y), \mL(X' | Y' = y) \right) = 1$ for all $y \not \in \textnormal{supp}(Y')$.
\end{enumerate}
\end{fact}

Given an algorithm $\mathcal{A}$ and distribution $\mP$ on inputs, let $\mathcal{A}(\mP)$ denote the distribution of $\mathcal{A}(X)$ induced by $X \sim \mP$. If $\mathcal{A}$ has $k$ steps, let $\mathcal{A}_i$ denote the $i$th step of $\mathcal{A}$ and $\mathcal{A}_{i\text{-}j}$ denote the procedure formed by steps $i$ through $j$. Each time this notation is used, we clarify the intended initial and final variables when $\mathcal{A}_{i}$ and $\mathcal{A}_{i\text{-}j}$ are viewed as Markov kernels. The next lemma from \cite{brennan2019universality} encapsulates the structure of all of our analyses of average-case reductions. Its proof is simple and included in Appendix \ref{sec:appendix} for completeness.

\begin{lemma}[Lemma 4.2 in \cite{brennan2019universality}] \label{lem:tvacc}
Let $\mathcal{A}$ be an algorithm that can be written as $\mathcal{A} = \mathcal{A}_m \circ \mathcal{A}_{m-1} \circ \cdots \circ \mathcal{A}_1$ for a sequence of steps $\mathcal{A}_1, \mathcal{A}_2, \dots, \mathcal{A}_m$. Suppose that the probability distributions $\mP_0, \mP_1, \dots, \mP_m$ are such that $\TV(\mathcal{A}_i(\mP_{i-1}), \mP_i) \le \epsilon_i$ for each $1 \le i \le m$. Then it follows that
$$\TV\left( \mathcal{A}(\mP_0), \mP_m \right) \le \sum_{i = 1}^m \epsilon_i$$
\end{lemma}

\section{Reduction to Imbalanced Sparse Mixtures}
\label{sec:redisgm}

\begin{figure}[t!]
\begin{algbox}
\textbf{Algorithm} $k$\textsc{-pds-to-isgm}

\vspace{1mm}

\textit{Inputs}: $k$\pr{-pds} instance $G \in \mG_N$ with dense subgraph size $k$ that divides $N$, partition $E$ of $[N]$ and edge probabilities $0 < q < p \le 1$, a slow growing function $w(N) = \omega(1)$, target $\pr{isgm}$ parameters $(n, k, d, \mu, \epsilon)$ satisfying that $\epsilon = 1/r$ for some prime number $r$, $wn \le k \cdot \frac{r^t - 1}{r - 1}$ for some $t \in \mathbb{N}$, $d \ge m$ and $kr^t \ge m$ where $m$ is the smallest multiple of $k$ larger than $\left( \frac{p}{Q} + 1 \right) N$ where $Q = 1 - \sqrt{(1 - p)(1 - q)} + \mathbf{1}_{\{ p = 1\}} \left( \sqrt{q} - 1 \right)$, $\mu \le c/\sqrt{r^t(r - 1) \log(kmr^t)}$ for a sufficiently small constant $c > 0$ and $n, kr^t \le \text{poly}(N)$

\begin{enumerate}
\item \textit{Symmetrize and Plant Diagonals}: Compute $M_{\text{PD1}} \in \{0, 1\}^{m \times m}$ with partition $F'$ of $[m]$ as
$$M_{\text{PD1}} \gets \pr{To-}k\textsc{-Partite-Submatrix}(G)$$
applied with initial dimension $N$, edge probabilities $p$ and $q$ and target dimension $m$.
\item \textit{Pad}: Form $M_{\text{PD2}} \in \{0, 1\}^{m \times kr^t}$ by adding $kr^t - m$ new columns sampled i.i.d. from $\text{Bern}(Q)^{\otimes m}$ to $M_{\text{PD1}}$. Let $F_i$ be $F_i'$ with $r^t - m/k$ of the new columns. Randomly permute the row indices of $M_{\text{PD2}}$ and the column indices of $M_{\text{PD2}}$ within each part $F_i$.
\item \textit{Gaussianize}: Compute $M_{\text{G}} \in \mathbb{R}^{m \times kr^t}$ as
$$M_{\text{G}} \gets \textsc{Gaussianize}(M_{\text{PD2}})$$
applied with probabilities $p$ and $Q$ and $\mu_{ij} = \sqrt{r^t(r - 1)} \cdot \mu$ for all $(i, j) \in [m] \times [kr^t]$.
\item \textit{Construct Rotation Matrix}: Form the $\ell \times r^t$ matrix $H_{r, t}$ where $\ell = \frac{r^t - 1}{r - 1}$ as follows
\begin{enumerate}
\item[(1)] Let $V_1, V_2, \dots, V_{\ell}$ be an enumeration of the hyperplanes of $\mathbb{F}_r^t$ and $P_1, P_2, \dots, P_{r^t}$ be an enumeration of the points in $\mathbb{F}_r^t$.
\item[(2)] For each pair $(i, j) \in [\ell] \times [r^t]$, set the $(i, j)$th entry of $H_{r, t}$ to be
$$(H_{r, t})_{ij} = \frac{1}{\sqrt{r^t(r - 1)}} \cdot \left\{ \begin{matrix} 1 & \text{if } P_j \not \in V_i \\ 1 - r & \text{if } P_j \in V_i \end{matrix} \right.$$
\end{enumerate}
\item \textit{Sample Rotation}: Fix a partition $[k\ell] = F_1' \cup F_2' \cup \cdots \cup F_k'$ into $k$ parts each of size $\ell$ and compute the matrix $M_{\text{R}} \in \mathbb{R}^{m \times k\ell}$ where
$$(M_{\text{R}})_{F_i'} = (M_{\text{G}})_{F_i} H_{r, t}^\top \quad \text{for each } i \in [k]$$
where $Y_{F_i}$ denotes the submatrix of $Y$ restricted to the columns with indices in $F_i$.
\item \textit{Permute and Output}: Form $X \in \mathbb{R}^{d \times n}$ by choosing $n$ columns of $M_{\text{R}}$ uniformly at random, randomly embedding the resulting matrix as $m$ rows of $X$ and sampling the remaining $d - m$ rows of $X$ i.i.d. from $\mN(0, I_n)$. Output the columns $(X_1, X_2, \dots, X_n)$ of $X$.
\end{enumerate}
\vspace{0.5mm}

\end{algbox}
\caption{Reduction from $k$-partite planted dense subgraph to exactly imbalanced sparse Gaussian mixtures.}
\label{fig:isgmreduction}
\end{figure}

In this section, we give our reduction from $k$\pr{-pds} to the key intermediate problem \pr{isgm}, which we will reduce from in subsequent sections to obtain several of our main computational lower bounds. We also introduce several average-case reduction subroutines that will be used in our reduction to semirandom community recovery. The problem \pr{isgm}, imbalanced sparse Gaussian mixtures, is a simple vs. simple hypothesis testing problem defined formally below. A similar distribution was also used in \cite{diakonikolas2017statistical} to construct an instance of robust sparse mean estimation inducing the tight statistical-computational gap in the statistical query model.

\begin{definition}[Imbalanced Sparse Gaussian Mixtures]
Given some $\mu \in \mathbb{R}$ and $\epsilon \in (0, 1)$, let $\mu'$ be such that $\epsilon \cdot \mu' + (1 - \epsilon) \cdot \mu = 0$. The distribution $\pr{isgm}_{H_1}(n, k, d, \mu, \epsilon)$ over $X = (X_1, X_2, \dots, X_n)$ where $X_i \in \mathbb{R}^d$ is sampled as follows:
\begin{enumerate}
\item choose a $k$-subset $S \subseteq [d]$ uniformly at random;
\item sample $X_1, X_2, \dots, X_n$ i.i.d. from the mixture $\pr{mix}_{\epsilon}\left( \mN(\mu \cdot \mathbf{1}_S, I_d), \mN(\mu' \cdot \mathbf{1}_S, I_d) \right)$.
\end{enumerate}
The imbalanced sparse Gaussian mixture detection problem $\pr{isgm}(n, k, d, \mu, \epsilon)$ has observations $X = (X_1, X_2, \dots, X_n)$ and hypotheses
$$H_0 : X \sim \mN(0, I_d)^{\otimes n} \quad \textnormal{and} \quad H_1 : X \sim \pr{isgm}_{H_1}(n, k, d, \mu, \epsilon)$$
\end{definition}

Figure \ref{fig:isgmreduction} outlines the steps of our reduction from $k$\pr{-pds} to \pr{isgm}, using subroutines that will be introduced in the next two subsections. The reduction makes use of the framework for average-case reductions set forth in the sequence of work \cite{brennan2018reducibility, brennan2019universality, brennan2019optimal} for its initial steps transforming $k$\pr{-pc} into a Gaussianized submatrix problem, while preserving the partition-promise structure of $k$\pr{-pds}. These steps are discussed in Section \ref{subsec:tosubmatrix}.

In Section \ref{subsec:rotations}, we introduce the key insight of the reduction, which is to rotate the resulting Gaussianized submatrix problem by a carefully chosen matrix $H_{r, t}$ constructed using hyperplanes in $\mathbb{F}_r^t$, to arrive at an instance of \pr{isgm}. The matrix $H_{r, t}$ is a Grassmanian construction related to projective constructions of block designs and is close to a Hadamard matrix when $r = 2$. Three properties of $H_{r, t}$ are essential to our reduction: (1) $H_{r, t}$ has orthonormal rows; (2) $H_{r, t}$ contains only two distinct values; and (3) each column of $H_{r, t}$ has approximately an $1/r$ fraction of its entries negative. These properties are established and used in the analysis of our reduction in Section \ref{subsec:rotations}.

The next theorem encapsulates the total variation guarantees of the reduction $k$\textsc{-pds-to-isgm}. A key parameter is the prime number $r$, which is essential to our construction of the matrices $H_{r, t}$. In applications of the theorem to robust sparse mean estimation, $r$ will grow with $n$. To show the tightest possible statistical-computational gaps for robust sparse mean estimation, we ideally would want to take $N$ such that $N = \Theta(kr^t)$. When $r$ is growing with $n$, this induces number theoretic constraints on our choices of parameters that require careful attention. Because of this, we have kept the statement of our next theorem technically precise and in terms of all of the free parameters of the reduction $k$\pr{-pds-to-isgm}. Ignoring these number theoretic constraints, the reduction $k$\pr{-pds-to-isgm} can be interpreted as essentially mapping an instance of $k\pr{-pds}(N, k, p, q)$ with $k = o(\sqrt{N})$ to $\pr{isgm}(n, k, d, \mu, \epsilon)$ where $\epsilon \in (0, 1)$ is arbitrary and can vary with $n$. The target parameters $n, d$ and $\mu$ satisfy that
$$d = \Omega(N), \quad n = o(\epsilon N) \quad \text{and} \quad \mu \asymp \frac{1}{\sqrt{\log N}} \cdot \sqrt{\frac{\epsilon k}{N}}$$
All of our applications will handle the number theoretic constraints to set parameters so that they nearly satisfy these conditions. The slow-growing function $w(N)$ is so that Step 6 subsamples the produced samples by a large enough factor to enable an application of finite de Finetti's theorem. Our lower bounds for robust sparse mean estimation, semirandom community recovery and universality of lower bounds for sparse mixtures will set $r$ to be growing, $r = 3$ and $r = 2$, respectively. We now state our total variation guarantees for $k$\textsc{-pds-to-isgm}.

\begin{theorem}[Reduction from $k$\pr{-pds} to \pr{isgm}] \label{thm:isgmreduction}
Let $N$ be a parameter, $r = r(N) \ge 2$ be a prime number and $w(N) = \omega(1)$ be a slow-growing function. Fix initial and target parameters as follows:
\begin{itemize}
\item \textnormal{Initial} $k\pr{-pds}$ \textnormal{Parameters:} number of vertices $N$, dense subgraph size $k$ that divides $N$, fixed constant edge probabilities $0 < q < p \le 1$ with $q = N^{-O(1)}$ and a partition $E$ of $[N]$.
\item \textnormal{Target} $\pr{isgm}$ \textnormal{Parameters:} $(n, k, d, \mu, \epsilon)$ where $\epsilon = 1/r$ and there is a parameter $t = t(N) \in \mathbb{N}$ such that
$$wn \le \frac{k(r^t - 1)}{r - 1}, \quad m \le d, kr^t \le \textnormal{poly}(N) \quad \textnormal{and}$$
$$0 \le \mu \le \frac{\delta}{2 \sqrt{3\log (kmr^t) + 2\log (p - Q)^{-1}}} \cdot \frac{1}{\sqrt{r^t(r - 1)}}$$
where $m$ is the smallest multiple of $k$ larger than $\left( \frac{p}{Q} + 1 \right) N$, where $Q = 1 - \sqrt{(1 - p)(1 - q)} + \mathbf{1}_{\{ p = 1\}} \left( \sqrt{q} - 1 \right)$ and $\delta = \min \left\{ \log \left( \frac{p}{Q} \right), \log \left( \frac{1 - Q}{1 - p} \right) \right\}$.
\end{itemize}
Let $\mathcal{A}(G)$ denote $k$\textsc{-pds-to-isgm} applied to the graph $G$ with these parameters. Then $\mathcal{A}$ runs in randomized polynomial time and it follows that
$$\TV\left( \mathcal{A}\left( k\pr{-pds}(N, k, p, q) \right), \pr{isgm}(n, k, d, \mu, \epsilon) \right) = O\left( w^{-1} + \frac{k^2}{wN} + \frac{k}{\sqrt{N}} + e^{-\Omega(N^2/kn)} + N^{-1} \right)$$
under both $H_0$ and $H_1$ as $N \to \infty$.
\end{theorem}

Before proceeding to prove this theorem, we establish some notation that will be adopted throughout this section. Given a partition $F$ of $[N]$ with $[N] = F_1 \cup F_2 \cup \cdots \cup F_k$, let $\mU_N(F)$ denote the distribution of $k$-subsets of $[N]$ formed by choosing one member element of each of $F_1, F_2, \dots, F_k$ uniformly at random. Let $\mU_{N, k}$ denote the uniform distribution on $k$-subsets of $[N]$. Let $\mG(n, S, p, q)$ denote the distribution of planted dense subgraph instances from $\mG(n, k, p, q)$ conditioned on the subgraph being planted on the vertex set $S$ where $|S| = k$. Given $S \subseteq [m]$, $T \subseteq [n]$ and two distributions $\mP$ and $\mQ$ over $X$, let $\mathcal{M}(m, n, S, T, \mP, \mQ)$ denote the distribution of matrices in $X^{m \times n}$ with independent entries where $M_{ij} \sim \mP$ if $(i, j) \in S \times T$ and $M_{ij} \sim \mQ$ otherwise.

For simplicity of notation, when either $S$ or $T$ is a distribution $\mD$ on subsets of $[m]$ or $[n]$, we let this denote the mixture over $\mathcal{M}(m, n, S, T, \mP, \mQ)$ induced by sampling this set from $\mD$. We will adopt the same convention for $S$ in $\mG(n, S, p, q)$. We also let $\mathcal{M}(n, S, \mP, \mQ)$ be a shorthand for the distribution when $m = n$ and $S = T$. For simplicity, we also replace $\mP$ and $\mQ$ with $p$ and $q$ when $\mP = \text{Bern}(p)$ and $\mQ = \text{Bern}(q)$. Similarly, let $\mathcal{V}(n, S, \mP, \mQ)$ denote the distribution of vectors $v$ in $X^n$ with independent entries such that $v_i \sim \mP$ if $i \in S$ and $v_i \sim \mQ$ otherwise. We adopt the analogous shorthands for $\mathcal{V}(n, S, \mP, \mQ)$.

\subsection{Planting Diagonals, Cloning and Gaussianization}
\label{subsec:tosubmatrix}

In this section, we present several reductions from \cite{brennan2018reducibility, brennan2019universality, brennan2019optimal} that are used as subroutines in $k$\textsc{-pds-to-isgm}. We also introduce $\pr{To-}k\textsc{-Partite-Submatrix}$, which is a modified variant of the reduction $\pr{To-Submatrix}$ from \cite{brennan2019universality} that maps from the $k$-partite variant of planted dense subgraph. We remark that the proof of the total variation guarantees of $\pr{To-}k\textsc{-Partite-Submatrix}$ is technically more involved than that of $\pr{To-Submatrix}$ in \cite{brennan2019universality}.

We begin with the subroutine $\textsc{Graph-Clone}$, shown in Figure \ref{fig:clone}. This subroutine was introduced in \cite{brennan2019universality} and produces several independent samples from a planted subgraph problem given a single sample. Its properties as a Markov kernel are stated in the next lemma, which is proven by showing the two explicit expressions for $\bP[x^{ij} = v]$ in Step 1 define valid probability distributions and then explicitly writing the mass functions of $\mathcal{A}\left( \mG(n, q) \right)$ and $\mathcal{A}\left( \mG(n, S, p, q) \right)$.

\begin{figure}[t!]
\begin{algbox}
\textbf{Algorithm} \textsc{Graph-Clone}

\vspace{1mm}

\textit{Inputs}: Graph $G \in \mG_n$, the number of copies $t$, parameters $0 < q < p \le 1$ and $0 < Q < P \le 1$ satisfying $\frac{1 - p}{1 - q} \le \left( \frac{1 - P}{1 - Q} \right)^t$ and $\left( \frac{P}{Q} \right)^t \le \frac{p}{q}$

\begin{enumerate}
\item Generate $x^{ij} \in \{0, 1\}^t$ for each $1 \le i < j \le n$ such that:
\begin{itemize}
\item If $\{i, j \} \in E(G)$, sample $x^{ij}$ from the distribution on $\{0, 1\}^t$ with
$$\bP[x^{ij} = v] = \frac{1}{p - q} \left[ (1 - q) \cdot P^{|v|_1} (1 - P)^{t - |v|_1} - (1 - p) \cdot Q^{|v|_1} (1 - Q)^{t - |v|_1} \right]$$
\item If $\{i, j \} \not \in E(G)$, sample $x^{ij}$ from the distribution on $\{0, 1\}^t$ with
$$\bP[x^{ij} = v] = \frac{1}{p - q} \left[ p \cdot Q^{|v|_1} (1 - Q)^{t - |v|_1} - q \cdot P^{|v|_1} (1 - P)^{t - |v|_1} \right]$$
\end{itemize}
\item Output the graphs $(G_1, G_2, \dots, G_t)$ where $\{i, j\} \in E(G_k)$ if and only if $x^{ij}_k = 1$.
\end{enumerate}
\vspace{1mm}

\end{algbox}
\caption{Subroutine $\textsc{Graph-Clone}$ for producing independent samples from planted graph problems from \cite{brennan2019universality}.}
\label{fig:clone}
\end{figure}

\begin{lemma}[Graph Cloning -- Lemma 5.2 in \cite{brennan2019universality}] \label{lem:graphcloning}
Let $t \in \mathbb{N}$, $0 < q < p \le 1$ and $0 < Q < P \le 1$ satisfy that
$$\frac{1 - p}{1 - q} \le \left( \frac{1 - P}{1 - Q} \right)^t \quad \text{and} \quad \left( \frac{P}{Q} \right)^t \le \frac{p}{q}$$
Then the algorithm $\mathcal{A} = \textsc{Graph-Clone}$ runs in $\textnormal{poly}(t, n)$ time and satisfies that for each $S \subseteq [n]$,
$$\mathcal{A}\left( \mG(n, q) \right) \sim \mG(n, Q)^{\otimes t} \quad \text{and} \quad \mathcal{A}\left( \mG(n, S, p, q) \right) \sim \mG(n, S, P, Q)^{\otimes t}$$
\end{lemma}

\begin{figure}[t!]
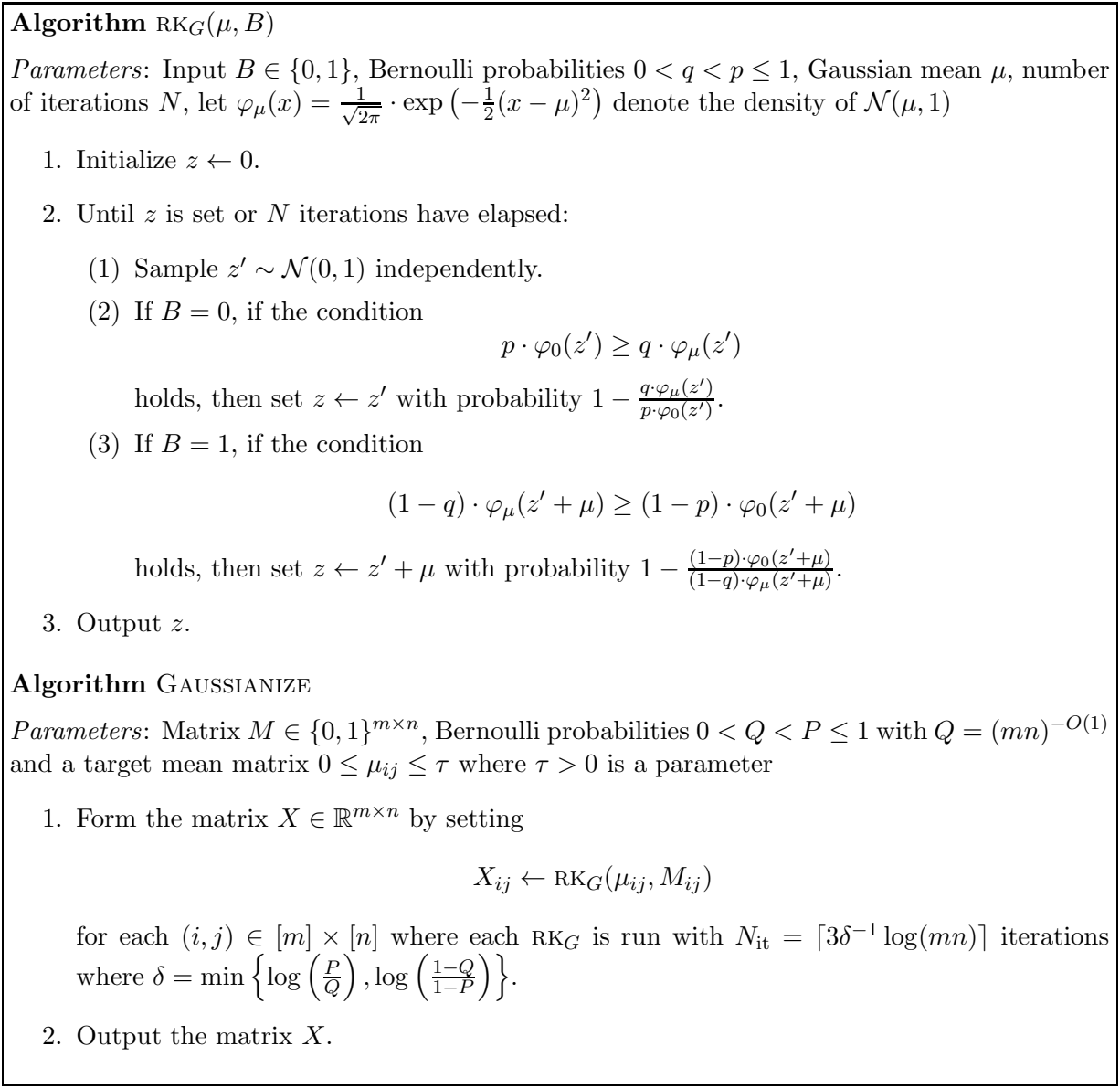

\begin{algbox}
\textbf{Algorithm} $\textsc{rk}_G(\mu, B)$

\vspace{2mm}

\textit{Parameters}: Input $B \in \{0, 1\}$, Bernoulli probabilities $0 < q < p \le 1$, Gaussian mean $\mu$, number of iterations $N$, let $\varphi_\mu(x) = \frac{1}{\sqrt{2\pi}} \cdot \exp\left(- \frac{1}{2}(x - \mu)^2 \right)$ denote the density of $\mN(\mu, 1)$
\begin{enumerate}
\item Initialize $z \gets 0$.
\item Until $z$ is set or $N$ iterations have elapsed:
\begin{enumerate}
\item[(1)] Sample $z' \sim \mN(0, 1)$ independently.
\item[(2)] If $B = 0$, if the condition
$$p \cdot \varphi_0(z') \ge q \cdot \varphi_{\mu}(z')$$
holds, then set $z \gets z'$ with probability $1 - \frac{q \cdot \varphi_\mu(z')}{p \cdot \varphi_0(z')}$.
\item[(3)] If $B = 1$, if the condition
$$(1 - q) \cdot \varphi_\mu(z' + \mu) \ge (1 - p) \cdot \varphi_0(z' + \mu)$$
holds, then set $z \gets z' + \mu$ with probability $1 - \frac{(1 - p) \cdot \varphi_0(z' + \mu)}{(1 - q) \cdot \varphi_\mu(z' + \mu)}$.
\end{enumerate}
\item Output $z$.
\end{enumerate}
\vspace{1mm}
\textbf{Algorithm} $\textsc{Gaussianize}$

\vspace{2mm}

\textit{Parameters}: Matrix $M \in \{0, 1\}^{m \times n}$, Bernoulli probabilities $0 < Q < P \le 1$ with $Q = (mn)^{-O(1)}$ and a target mean matrix $0 \le \mu_{ij} \le \tau$ where $\tau > 0$ is a parameter
\begin{enumerate}
\item Form the matrix $X \in \mathbb{R}^{m \times n}$ by setting
$$X_{ij} \gets \textsc{rk}_{G}(\mu_{ij}, M_{ij})$$
for each $(i, j) \in [m] \times [n]$ where each $\textsc{rk}_{G}$ is run with $N_{\text{it}} = \lceil 3\delta^{-1} \log (mn) \rceil$ iterations where $\delta = \min \left\{ \log \left( \frac{P}{Q} \right), \log \left( \frac{1 - Q}{1 - P} \right) \right\}$.
\item Output the matrix $X$.
\end{enumerate}
\vspace{1mm}
\end{algbox}
\caption{Gaussian instantiation of the rejection kernel algorithm from \cite{brennan2018reducibility} and the reduction $\textsc{Gaussianize}$ for mapping from Bernoulli to Gaussian planted submatrix problems from \cite{brennan2019optimal}.}
\label{fig:rej-kernel}
\end{figure}

We also will require the subroutine $\textsc{Gaussianize}$ from \cite{brennan2019optimal}, shown in Figure \ref{fig:rej-kernel}, which maps a planted Bernoulli submatrix problem to a corresponding submatrix problems with independent Gaussian entries. To describe this subroutine, we first will need the univariate rejection kernel framework introduced in \cite{brennan2018reducibility}. A multivariate extension of this framework is used in \cite{brennan2019universality}, but we will only require the univariate case in this section. The next lemma states the total variation guarantees of the Gaussian rejection kernels, which are also shown in Figure \ref{fig:rej-kernel}.

The proof of this lemma consists of showing that the distributions of the outputs $\pr{rk}_G(\mu, \text{Bern}(p))$ and $\pr{rk}_G(\mu, \text{Bern}(q))$ are close to $\mN(\mu, 1)$ and $\mN(0, 1)$ conditioned to lie in the set of $x$ with $\frac{1 - p}{1 - q} \le \frac{\varphi_\mu(x)}{\varphi_0(x)} \le \frac{p}{q}$ and then showing that this event occurs with probability close to one. We will use the notation $\textsc{rk}_G(B)$ to denote the random variable output by a run of the procedure $\textsc{rk}_G$ using independently generated randomness.

\begin{lemma}[Gaussian Rejection Kernels -- Lemma 5.4 in \cite{brennan2018reducibility}] \label{lem:5c}
Let $n$ be a parameter and suppose that $p = p(n)$ and $q = q(n)$ satisfy that $0 < q < p \le 1$, $\min(q, 1 - q) = \Omega(1)$ and $p - q \ge n^{-O(1)}$. Let $\delta = \min \left\{ \log \left( \frac{p}{q} \right), \log \left( \frac{1 - q}{1 - p} \right) \right\}$. Suppose that $\mu = \mu(n) \in (0, 1)$ satisfies that
$$\mu \le \frac{\delta}{2 \sqrt{6\log n + 2\log (p-q)^{-1}}}$$
Then the map $\textsc{rk}_{\text{G}}$ with $N = \left\lceil 6\delta^{-1} \log n \right\rceil$ iterations can be computed in $\text{poly}(n)$ time and satisfies
$$\TV\left(\textsc{rk}_{\text{G}}(\mu, \textnormal{Bern}(p)), \mN(\mu, 1) \right) = O(n^{-3}) \quad \text{and} \quad \TV\left(\textsc{rk}_{\text{G}}(\mu, \textnormal{Bern}(q)), \mN(0, 1) \right) = O(n^{-3})$$
\end{lemma}

We now state the total variation guarantees of $\textsc{Gaussianize}$. The instantiation of $\textsc{Gaussianize}$ here generalizes that in \cite{brennan2019optimal} to rectangular matrices, but has the same proof. The procedure applies a Gaussian rejection kernel entrywise and its total variation guarantees follow by a simple by applying the tensorization property of $\TV$ from Fact \ref{tvfacts}. 

\begin{lemma}[Gaussianization -- Lemma 4.5 in \cite{brennan2019optimal}] \label{lem:gaussianize}
Given parameters $m$ and $n$, let $0 < Q < P \le 1$ be such that $P - Q = (mn)^{-O(1)}$ and $\min(Q, 1 - Q) = \Omega(1)$, let $\mu_{ij}$ be such that $0 \le \mu_{ij} \le \tau$ for each $i \in [m]$ and $j \in [n]$ where the parameter $\tau > 0$ satisfies that
$$\tau \le \frac{\delta}{2 \sqrt{3\log (mn) + 2\log (P - Q)^{-1}}} \quad \text{where} \quad \delta = \min \left\{ \log \left( \frac{P}{Q} \right), \log \left( \frac{1 - Q}{1 - P} \right) \right\}$$
The algorithm $\mathcal{A} = \textsc{Gaussianize}$ runs in $\textnormal{poly}(mn)$ time and satisfies that
\begin{align*}
\TV\left( \mathcal{A}(\mathcal{M}(m, n, S, T, P, Q)), \, \mu \circ \mathbf{1}_S \mathbf{1}_T^\top + \mN(0, 1)^{\otimes m \times n} \right) &= O\left((mn)^{-1/2}\right) \\
\TV\left( \mathcal{A}\left(\textnormal{Bern}(Q)^{\otimes m \times n}\right), \, \mN(0, 1)^{\otimes m \times n} \right) &= O\left((mn)^{-1/2}\right)
\end{align*}
for all subsets $S \subseteq [m]$ and $T \subseteq [n]$ where $\circ$ denotes the Hadamard product between two matrices.
\end{lemma}

\begin{figure}[t!]
\begin{algbox}
\textbf{Algorithm} $\pr{To-}k\textsc{-Partite-Submatrix}$

\vspace{1mm}

\textit{Inputs}: $k$\pr{-pds} instance $G \in \mG_N$ with clique size $k$ that divides $N$ and partition $E$ of $[N]$, edge probabilities $0 < q < p \le 1$ with $q = N^{-O(1)}$ and target dimension $n \ge \left(\frac{p}{Q} + 1 \right)N$ where $Q = 1 - \sqrt{(1 - p)(1 - q)} + \mathbf{1}_{\{p = 1\}} \left( \sqrt{q} - 1 \right)$ and $k$ divides $n$
\begin{enumerate}
\item Apply $\textsc{Graph-Clone}$ to $G$ with edge probabilities $P = p$ and $Q = 1 - \sqrt{(1 - p)(1 - q)} + \mathbf{1}_{\{p = 1\}} \left( \sqrt{q} - 1 \right)$ and $t = 2$ clones to obtain $(G_1, G_2)$.
\item Let $F$ be a partition of $[n]$ with $[n] = F_1 \cup F_2 \cup \cdots \cup F_k$ and $|F_i| = n/k$. Form the matrix $M_{\text{PD}} \in \{0, 1\}^{n \times n}$ as follows:
\begin{enumerate}
\item[(1)] For each $t \in [k]$, sample $s_1^t \sim \text{Bin}(N/k, p)$ and $s_2^t \sim \text{Bin}(n/k, Q)$ and let $S_t$ be a subset of $F_t$ with $|S_t| = N/k$ selected uniformly at random. Sample $T_1^t \subseteq S_t$ and $T_2^t \subseteq F_t \backslash S_t$ with $|T_1^t| = s_1^t$ and $|T_2^t| = \max\{s_2^t - s_1^t, 0 \}$ uniformly at random.
\item[(2)] Now form the matrix $M_{\text{PD}}$ such that its $(i, j)$th entry is
$$(M_{\text{PD}})_{ij} = \left\{ \begin{array}{ll} \mathbf{1}_{\{\pi_t(i), \pi_t(j)\} \in E(G_1)} & \text{if } i < j \text{ and } i, j \in S_t \\ \mathbf{1}_{\{\pi_t(i), \pi_t(j)\} \in E(G_2)} & \text{if } i > j \text{ and } i, j \in S_t \\ \mathbf{1}_{\{ i \in T_1^t \}} & \text{if } i = j \text{ and } i, j \in S_t \\ \mathbf{1}_{\{i \in T_2^t\}} & \text{if } i = j \text{ and } i, j \in F_t \backslash S_t \\ \sim_{\text{i.i.d.}} \text{Bern}(Q) & \text{if } i \neq j \text{ and } (i, j) \not \in S_t^2 \text{ for a } t \in [k] \end{array} \right.$$
where $\pi_t : S_t \to E_t$ is a bijection chosen uniformly at random.
\end{enumerate}
\item Output the matrix $M_{\text{PD}}$ and the partition $F$.
\end{enumerate}
\vspace{1mm}

\end{algbox}
\caption{Subroutine $\pr{To-}k\textsc{-Partite-Submatrix}$ for mapping from an instance of $k$-partite planted dense subgraph to a $k$-partite Bernoulli submatrix problem.}
\label{fig:tosubmatrix}
\end{figure}

We now introduce the procedure $\pr{To-}k\textsc{-Partite-Submatrix}$, which is shown in Figure \ref{fig:tosubmatrix}. This reduction clones the upper half of the adjacency matrix of the input graph problem to produce an independent lower half and plants diagonal entries while randomly embedding into a larger matrix to hide the diagonal entries in total variation. $\pr{To-}k\textsc{-Partite-Submatrix}$ is similar to $\textsc{To-Submatrix}$ in \cite{brennan2019universality} and $\textsc{To-Bernoulli-Submatrix}$ in \cite{brennan2019optimal} but ensures that the random embedding step accounts for the $k$-partite promise of the input $k$\pr{-pds} instance.

We begin with the following lemma, which is a key computation in the proof of correctness for $\pr{To-}k\textsc{-Partite-Submatrix}$. We remark that the total variation upper bound in this lemma is tight in the following sense. When all of the $P_i$ are the same, the expected value of the sum of the coordinates of the first distribution is $k(P_i - Q)$ higher than that of the second. The standard deviation of the second sum is $\sqrt{kmQ(1 - Q)}$ and thus when $k(P_i - Q)^2 \gg mQ(1 - Q)$, the total variation below tends to one.

\begin{lemma} \label{lem:bernproduct}
If $k, m \in \mathbb{N}$, $P_1, P_2, \dots, P_k \in [0, 1]$ and $Q \in (0, 1)$, then
$$\TV\left( \otimes_{i = 1}^k \left( \textnormal{Bern}(P_i) + \textnormal{Bin}(m - 1, Q) \right), \textnormal{Bin}(m, Q)^{\otimes k} \right) \le \sqrt{\sum_{i = 1}^k \frac{(P_i - Q)^2}{2mQ(1 - Q)}}$$
\end{lemma}

\begin{proof}
Given some $P \in [0, 1]$, we begin by computing $\chi^2\left( \textnormal{Bern}(P) + \textnormal{Bin}(m - 1, Q), \textnormal{Bin}(m, Q) \right)$. For notational convenience, let $\binom{a}{b} = 0$ if $b > a$ or $b < 0$. It follows that
\allowdisplaybreaks
\begin{align*} 
&1 + \chi^2\left( \textnormal{Bern}(P) + \textnormal{Bin}(m - 1, Q), \textnormal{Bin}(m, Q) \right) \\
&\quad \quad = \sum_{t = 0}^{m} \frac{\left((1 - P) \cdot \binom{m - 1}{t} Q^t (1 - Q)^{m - 1 - t} + P \cdot \binom{m - 1}{t - 1} Q^{t - 1} (1 - Q)^{m - t} \right)^2}{\binom{m}{t} Q^t (1 - Q)^{m - t}} \\
&\quad \quad = \sum_{t = 0}^{m} \binom{m}{t} Q^t (1 - Q)^{m - t} \left( \frac{m - t}{m} \cdot \frac{1 - P}{1 - Q} + \frac{t}{m} \cdot \frac{P}{Q} \right)^2 \\
&\quad \quad = \bE\left[ \left( \frac{m - X}{m} \cdot \frac{1 - P}{1 - Q} + \frac{X}{m} \cdot \frac{P}{Q} \right)^2 \right] \\
&\quad \quad = \bE\left[ \left( 1 + \frac{X - mQ}{m} \cdot \frac{P - Q}{Q(1 - Q)} \right)^2 \right] \\
&\quad \quad = 1 + \frac{2(P - Q)}{mQ(1 - Q)} \cdot \bE[X - mQ] + \frac{(P - Q)^2}{m^2Q^2(1 - Q)^2} \cdot \bE\left[(X - Qm)^2\right] \\
&\quad \quad = 1 + \frac{(P - Q)^2}{mQ(1 - Q)}
\end{align*}
where $X \sim \textnormal{Bin}(m, Q)$ and the second last equality follows from $\bE[X] = Qm$ and $\bE[(X - Qm)^2] = \text{Var}[X] = Q(1 - Q)m$. The concavity of $\log$ implies that $\KL(\mP, \mQ) \le \log\left( 1 + \chi^2(\mP, \mQ) \right) \le \chi^2(\mP, \mQ)$ for any two distributions with $\mP$ absolutely continuous with respect to $\mQ$. Pinsker's inequality and tensorization of $\KL$ now imply that
\begin{align*}
&2 \cdot \TV\left( \otimes_{i = 1}^k \left( \textnormal{Bern}(P_i) + \textnormal{Bin}(m - 1, Q) \right), \textnormal{Bin}(m, Q)^{\otimes k} \right)^2 \\
&\quad \quad \le \KL\left( \otimes_{i = 1}^k \left( \textnormal{Bern}(P_i) + \textnormal{Bin}(m - 1, Q) \right), \textnormal{Bin}(m, Q)^{\otimes k} \right) \\
&\quad \quad = \sum_{i = 1}^k \KL\left( \textnormal{Bern}(P_i) + \textnormal{Bin}(m - 1, Q), \textnormal{Bin}(m, Q) \right) \\
&\quad \quad \le \sum_{i = 1}^k \chi^2\left( \textnormal{Bern}(P_i) + \textnormal{Bin}(m - 1, Q), \textnormal{Bin}(m, Q) \right) = \sum_{i = 1}^k \frac{(P_i - Q)^2}{mQ(1 - Q)}
\end{align*}
which completes the proof of the lemma.
\end{proof}

We now use this lemma to establish an analogue of Lemma 6.4 from \cite{brennan2019universality} in the $k$-partite case to analyze the planted diagonal entries in Step 2 of $\pr{To-}k\textsc{-Partite-Submatrix}$.

\begin{lemma}[Planting $k$-Partite Diagonals] \label{lem:plantingdiagonals}
Suppose that $0 < Q < P \le 1$ and $n \ge \left( \frac{P}{Q} + 1 \right) N$ is such that both $N$ and $n$ are divisible by $k$ and $k \le QN/4$. Suppose that for each $t \in [k]$,
$$z_1^t \sim \textnormal{Bern}(P), \quad z_2^t \sim \textnormal{Bin}(N/k - 1, P) \quad \textnormal{and} \quad z_3^t \sim \textnormal{Bin}(n/k, Q)$$
are independent. If $z_4^t = \max \{ z_3^t - z_1^t - z_2^t, 0 \}$, then it follows that
\begin{align*}
\TV\left( \otimes_{t = 1}^k \mL(z_1^t, z_2^t + z_4^t), \left( \textnormal{Bern}(P) \otimes \textnormal{Bin}(n/k - 1, Q) \right)^{\otimes k} \right) &\le 4k \cdot \exp \left( - \frac{Q^2N^2}{48Pkn} \right) + \sqrt{\frac{C_Q k^2}{2n}} \\
\TV\left( \otimes_{t = 1}^k \mL(z_1^t + z_2^t + z_4^t), \textnormal{Bin}(n/k, Q)^{\otimes k} \right) &\le 4k \cdot \exp \left( - \frac{Q^2N^2}{48Pkn} \right)
\end{align*}
where $C_Q = \max \left\{ \frac{Q}{1 - Q}, \frac{1 - Q}{Q} \right\}$.
\end{lemma}

\begin{proof}
Throughout this argument, let $v$ denote a vector in $\{0, 1\}^k$. Now define the event
$$\mathcal{E} = \bigcap_{t = 1}^k \left\{ z_3^t = z_1^t + z_2^t + z_4^t \right\}$$
Now observe that if $z_3^t \ge Qn/k - QN/2k + 1$ and $z_2^t \le P(N/k - 1) + QN/2k$ then it follows that $z_3^t \ge 1 + z_2^t \ge v_t + z_2^t$ for any $v_t \in \{0, 1\}$ since $Qn \ge (P+Q)N$. Now union bounding the probability that $\mathcal{E}$ does not hold conditioned on $z_1$ yields that
\allowdisplaybreaks
\begin{align*}
\bP\left[ \mathcal{E}^C \Big| z_1 = v \right] &\le \sum_{t = 1}^k \bP\left[ z_3^t < v_t + z_2^t \right] \\
&\le \sum_{t = 1}^k \bP\left[ z_3^t < \frac{Qn}{k} - \frac{QN}{2k} + 1 \right] + \sum_{t = 1}^k \bP\left[ z_2^t > P\left(\frac{N}{k} - 1\right) + \frac{QN}{2k} \right] \\
&\le k \cdot \exp\left( - \frac{\left(QN/2k - 1 \right)^2}{3Qn/k} \right) + k \cdot \exp\left( - \frac{\left(QN/2k \right)^2}{2P(N/k - 1)} \right) \\
&\le 2k \cdot \exp \left( - \frac{Q^2N^2}{48Pkn} \right)
\end{align*}
where the third inequality follows from standard Chernoff bounds on the tails of the binomial distribution. Marginalizing this bound over $v \sim \mL(z_1) = \text{Bern}(P)^{\otimes k}$, we have that
$$\bP\left[ \mathcal{E}^C \right] = \bE_{v \sim \mL(z_1)} \bP\left[ \mathcal{E}^C \Big| z_1 = v \right] \le 2k \cdot \exp \left( - \frac{Q^2N^2}{48Pkn} \right)$$
Now consider the total variation error induced by conditioning each of the product measures $\otimes_{t = 1}^k \mL(z_1^t + z_2^t + z_4^t)$ and $\otimes_{t = 1}^k \mL(z_3^t)$ on the event $\mathcal{E}$. Note that under $\mathcal{E}$, by definition, we have that $z_3^t = z_1^t + z_2^t + z_4^t$ for each $t \in [k]$. By the conditioning property of $\TV$ in Fact \ref{tvfacts}, we have
\begin{align*}
\TV\left( \otimes_{t = 1}^k \mL(z_1^t + z_2^t + z_4^t), \mL\left( \left(z_3^t : t \in [k]\right) \Big| \mathcal{E} \right) \right) &\le \bP\left[ \mathcal{E}^C \right] \\
\TV\left( \otimes_{t = 1}^k \mL(z_3^t), \mL\left( \left(z_3^t : t \in [k]\right) \Big| \mathcal{E} \right) \right) &\le \bP\left[ \mathcal{E}^C \right]
\end{align*}
The fact that $\otimes_{t = 1}^k \mL(z_3^t) = \text{Bin}(n/k, Q)^{\otimes k}$ and the triangle inequality now imply that
$$\TV\left( \otimes_{t = 1}^k \mL(z_1^t + z_2^t + z_4^t), \textnormal{Bin}(n/k, Q)^{\otimes k} \right) \le 2 \cdot \bP\left[ \mathcal{E}^C \right] \le 4k \cdot \exp \left( - \frac{Q^2N^2}{48Pkn} \right)$$
which proves the second inequality in the statement of the lemma. It suffices to establish the first inequality. A similar conditioning step as above shows that for all $v \in \{0, 1\}^k$, we have that
\begin{align*}
\TV\left( \otimes_{t = 1}^k \mL\left(v_t + z_2^t + z_4^t \Big| z_1^t = v_t\right), \mL\left( \left(v_t + z_2^t + z_4^t : t \in [k]\right) \Big| z_1 = v \text{ and } \mathcal{E} \right) \right) &\le \bP\left[ \mathcal{E}^C \Big| z_1 = v \right] \\
\TV\left( \otimes_{t = 1}^k \mL\left(z_3^t \Big| z_1^t = v_t \right), \mL\left( \left(z_3^t : t \in [k]\right) \Big| z_1 = v \text{ and } \mathcal{E} \right) \right) &\le \bP\left[ \mathcal{E}^C \Big| z_1 = v \right]
\end{align*}
The triangle inequality and the fact that $z_3 \sim \text{Bin}(n/k, Q)^{\otimes k}$ is independent of $z_1$ implies that
$$\TV\left( \otimes_{t = 1}^k \mL\left(v_t + z_2^t + z_4^t \Big| z_1^t = v_t\right), \text{Bin}(n/k, Q)^{\otimes k} \right) \le 4k \cdot \exp \left( - \frac{Q^2N^2}{48Pkn} \right)$$
By Lemma \ref{lem:bernproduct} applied with $P_t = v_t \in \{0, 1\}$, we also have that
$$\TV\left( \otimes_{t = 1}^k \left( v_t + \text{Bin}(n/k - 1, Q) \right), \text{Bin}(n/k, Q)^{\otimes k} \right) \le \sqrt{\sum_{t = 1}^k \frac{k(v_t - Q)^2}{2nQ(1 - Q)}} \le \sqrt{\frac{C_Q k^2}{2n}}$$
The triangle now implies that for each $v \in \{0, 1\}^k$,
\allowdisplaybreaks
\begin{align*}
&\TV\left( \otimes_{t = 1}^k \mL\left(z_2^t + z_4^t \Big| z_1^t = v_t\right), \text{Bin}(n/k - 1, Q)^{\otimes k} \right) \\
&\quad \quad = \TV\left( \otimes_{t = 1}^k \mL\left(v_t + z_2^t + z_4^t \Big| z_1^t = v_t\right), \otimes_{t = 1}^k \left( v_t + \text{Bin}(n/k - 1, Q) \right) \right) \\
&\quad \quad \le 4k \cdot \exp \left( - \frac{Q^2N^2}{48Pkn} \right) + \sqrt{\frac{C_Q k^2}{2n}}
\end{align*}
We now marginalize over $v \sim \mL(z_1) = \text{Bern}(P)^{\otimes k}$. The conditioning on a random variable property of $\TV$ in Fact \ref{tvfacts} implies that
\begin{align*}
&\TV\left( \otimes_{t = 1}^k \mL(z_1^t, z_2^t + z_4^t), \left( \textnormal{Bern}(P) \otimes \textnormal{Bin}(n/k - 1, Q) \right)^{\otimes k} \right) \\
&\quad \quad \le \bE_{v \sim \text{Bern}(P)^{\otimes k}} \, \TV\left( \otimes_{t = 1}^k \mL\left(z_2^t + z_4^t \Big| z_1^t = v_t\right), \text{Bin}(n/k - 1, Q)^{\otimes k} \right)
\end{align*}
which, when combined with the inequalities above, completes the proof of the lemma.
\end{proof}

We now combine these lemmas to analyze $\pr{To-}k\textsc{-Partite-Submatrix}$. The next lemma is a $k$-partite variant of Theorem 6.1 in \cite{brennan2019universality} and involves several technical subtleties that do not arise in the non $k$-partite case. After applying $\textsc{Graph-Clone}$, the adjacency matrix of the input graph $G$ is still missing its diagonal entries. The main difficulty in producing these diagonal entries is to ensure that entries corresponding to vertices in the planted subgraph are properly sampled from $\text{Bern}(p)$. To do this, we randomly embed the original $N \times N$ adjacency matrix in a larger $n \times n$ matrix with i.i.d. entries from $\text{Bern}(Q)$ and sample all diagonal entries corresponding to entries of the original matrix from $\text{Bern}(p)$. The diagonal entries in the new $n - N$ columns are chosen so that the supports on the diagonals within each $F_t$ each have size $\text{Bin}(n/k, Q)$. Even though this causes the sizes of the supports on the diagonals in each $F_t$ to have the same distribution under both $H_0$ and $H_1$, the randomness of the embedding and the fact that $k = o(\sqrt{n})$ ensures that this is hidden in total variation. Showing this involves some technical subtleties captured in the above two lemmas and the next lemma.

\begin{lemma}[Reduction to $k$-Partite Bernoulli Submatrix Problems] \label{lem:submatrix}
Let $0 < q < p \le 1$ and $Q = 1 - \sqrt{(1 - p)(1 - q)} + \mathbf{1}_{\{p = 1\}} \left( \sqrt{q} - 1 \right)$. Suppose that $n$ and $N$ are such that
$$n \ge \left( \frac{p}{Q} + 1 \right) N \quad \text{and} \quad k \le QN/4$$
Also suppose that $q = N^{-O(1)}$ and both $N$ and $n$ are divisible by $k$. Let $E = (E_1, E_2, \dots, E_k)$ and $F = (F_1, F_2, \dots, F_k)$ be partitions of $[N]$ and $[n]$, respectively. Then it follows that the algorithm $\mathcal{A} = \pr{To-}k\textsc{-Partite-Submatrix}$ runs in $\textnormal{poly}(N)$ time and satisfies
\begin{align*}
\TV\left( \mathcal{A}(\mG(N, \mU_N(E), p, q)), \, \mathcal{M}\left(n, \mU_n(F), p, Q \right) \right) &\le 4k \cdot \exp \left( - \frac{Q^2N^2}{48pkn} \right) + \sqrt{\frac{C_Q k^2}{2n}} \\
\TV\left( \mathcal{A}(\mG(N, q)), \, \textnormal{Bern}\left( Q \right)^{\otimes n \times n} \right) &\le 4k \cdot \exp \left( - \frac{Q^2N^2}{48pkn} \right)
\end{align*}
where $C_Q = \max \left\{ \frac{Q}{1 - Q}, \frac{1 - Q}{Q} \right\}$.
\end{lemma}

\begin{proof}
Fix some subset $R \subseteq [N]$ such that $|R \cap E_i| = 1$ for each $i \in [k]$. We will first show that $\mathcal{A}$ maps an input $G \sim G(N, R, p, q)$ approximately in total variation to $\mathcal{M}\left(n, \mU_n(F), p, Q \right)$. By AM-GM, we have that
$$\sqrt{pq} \le \frac{p + q}{2} = 1 - \frac{(1 - p) + (1 - q)}{2} \le 1 - \sqrt{(1 - p)(1 - q)}$$
If $p \neq 1$, it follows that $P = p > Q = 1 - \sqrt{(1 - p)(1 - q)}$. This implies that $\frac{1 - p}{1 - q} = \left( \frac{1 - P}{1 - Q} \right)^2$ and the inequality above rearranges to $\left( \frac{P}{Q} \right)^2 \le \frac{p}{q}$. If $p = 1$, then $Q = \sqrt{q}$ and $\left( \frac{P}{Q} \right)^2 = \frac{p}{q}$. Furthermore, the inequality $\frac{1 - p}{1 - q} \le \left( \frac{1 - P}{1 - Q} \right)^2$ holds trivially. Therefore we may apply Lemma \ref{lem:graphcloning}, which implies that $(G_1, G_2) \sim \mG(N, R, p, Q)^{\otimes 2}$.

Let the random set $U = \{ \pi_1^{-1}(R \cap E_1), \pi_2^{-1}(R \cap E_2), \dots, \pi_k^{-1}(R \cap E_k) \}$ denote the support of the $k$-subset of $[n]$ that $R$ is mapped to in the embedding step of $\pr{To-}k\textsc{-Partite-Submatrix}$. Now fix some $k$-subset $R' \subseteq [n]$ with $|R' \cap F_i| = 1$ for each $i \in [k]$ and consider the distribution of $M_{\text{PD}}$ conditioned on the event $U = R'$. Since $(G_1, G_2) \sim \mG(n, R, p, Q)^{\otimes 2}$, Step 2 of $\pr{To-}k\textsc{-Partite-Submatrix}$ ensures that the off-diagonal entries of $M_{\text{PD}}$, given this conditioning, are independent and distributed as follows:
\begin{itemize}
\item $M_{ij} \sim \text{Bern}(p)$ if $i \neq j$ and $i, j \in R'$; and
\item $M_{ij} \sim \text{Bern}(Q)$ if $i \neq j$ and $i \not \in R'$ or $j \not \in R'$.
\end{itemize}
which match the corresponding entries of $\mathcal{M}(n, R', p, Q)$. Furthermore, these entries are independent of the vector $\text{diag}(M_{\text{PD}}) = \left( (M_{\text{PD}})_{ii} : i \in [k] \right)$ of the diagonal entries of $M_{\text{PD}}$. It therefore follows that
$$\TV\left( \mL \left( M_{\text{PD}} \Big| U = R' \right), \mathcal{M}(n, R', p, Q) \right) = \TV\left( \mL \left( \text{diag}(M_{\text{PD}}) \Big| U = R' \right), \mathcal{V}(n, R', p, Q) \right)$$
Let $(S_1', S_2', \dots, S_k')$ be any tuple of fixed subsets such that $|S_t'| = N/k$, $S_i' \subseteq F_t$ and $R' \cap F_t \in S_t'$ for each $t \in [k]$. Now consider the distribution of $\text{diag}(M_{\text{PD}})$ conditioned on both $U = R'$ and $(S_1, S_2, \dots, S_k) = (S_1', S_2', \dots, S_k')$. It holds by construction that the $k$ vectors $\text{diag}(M_{\text{PD}})_{F_t}$ are independent for $t \in [k]$ and each distributed as follows:
\begin{itemize}
\item $\text{diag}(M_{\text{PD}})_{S_t'}$ is an exchangeable distribution on $\{0, 1\}^{N/k}$ with support of size $s_1^t \sim \text{Bin}(N/k, p)$, by construction. This implies that $\text{diag}(M_{\text{PD}})_{S_t'} \sim \text{Bern}(p)^{\otimes N/k}$. This can trivially be restated as $\left(M_{R' \cap F_t, R' \cap F_t}, \text{diag}(M_{\text{PD}})_{S_t' \backslash R'}\right) \sim \text{Bern}(p) \otimes \text{Bern}(p)^{\otimes N/k - 1}$.
\item $\text{diag}(M_{\text{PD}})_{F_t \backslash S_t'}$ is an exchangeable distribution on $\{0, 1\}^{N/k}$ with support of size $z_4^t = \max\{s_2^t - s_1^t, 0\}$. Furthermore, $\text{diag}(M_{\text{PD}})_{F_t \backslash S_t'}$ is independent of $\text{diag}(M_{\text{PD}})_{S_t'}$.
\end{itemize}
For each $t \in [k]$, let $z_1^t = M_{R' \cap F_t, R' \cap F_t} \sim \text{Bern}(p)$ and $z_2^t \sim \text{Bin}(N/k - 1, p)$ be the size of the support of $\text{diag}(M_{\text{PD}})_{S_t' \backslash R'}$. As shown discussed in the first point above, we have that $z_1^t$ and $z_2^t$ are independent and $z_1^t + z_2^t = s_1^t$.

Now consider the distribution of $\text{diag}(M_{\text{PD}})$ relaxed to only be conditioned on $U = R'$, and no longer on $(S_1, S_2, \dots, S_k) = (S_1', S_2', \dots, S_k')$. Conditioned on $U = R'$, the $S_t$ are independent and each uniformly distributed among all $N/k$ size subsets of $F_t$ that contain the element $R' \cap F_t$. In particular, this implies that the distribution of $\text{diag}(M_{\text{PD}})_{F_t \backslash R'}$ is an exchangeable distribution on $\{0, 1\}^{n/k - 1}$ with support size $z_2^t + z_4^t$ for each $t$. Note that any $v \sim \mathcal{V}(n, R', p, Q)$ also satisfies that $v_{F_t \backslash R'}$ is exchangeable. This implies that $\mathcal{V}(n, R', p, Q)$ and $\text{diag}(M_{\text{PD}})$ are identically distributed when conditioned on their entries with indices in $R'$ and on their support sizes within the $k$ sets of indices $F_t \backslash R'$. The conditioning property of Fact \ref{tvfacts} therefore implies that
\begin{align*}
&\TV\left( \mL \left( \text{diag}(M_{\text{PD}}) \Big| U = R' \right), \mathcal{V}(n, R', p, Q) \right) \\
&\quad \quad \le \TV\left( \otimes_{t = 1}^k \mL(z_1^t, z_2^t + z_4^t), \left( \textnormal{Bern}(p) \otimes \textnormal{Bin}(n/k - 1, Q) \right)^{\otimes k} \right) \\
&\quad \quad \le 4k \cdot \exp \left( - \frac{Q^2N^2}{48Pkn} \right) + \sqrt{\frac{C_Q k^2}{2n}}
\end{align*}
by the first inequality in Lemma \ref{lem:plantingdiagonals}. Now observe that $U \sim \mU_n(F)$ and thus marginalizing over $R' \sim \mL(U) = \mU_n(F)$ and applying the conditioning property of Fact \ref{tvfacts} yields that
$$\TV\left( \mathcal{A}(G(N, R, p, q)), \mathcal{M}(n, \mU_n(F), p, Q) \right) \le \bE_{R' \sim \mU_n(F)} \, \TV\left( \mL \left( M_{\text{PD}} \Big| U = R' \right), \mathcal{M}(n, R', p, Q) \right)$$
since $M_{\text{PD}} \sim \mathcal{A}(\mG(N, R, p, q))$. Applying an identical marginalization over $R \sim \mU_N(E)$ completes the proof of the first inequality in the lemma statement.

It suffices to consider the case where $G \sim \mG(N, q)$, which follows from an analogous but simpler argument. By Lemma \ref{lem:graphcloning}, we have that $(G_1, G_2) \sim \mG(N, Q)^{\otimes 2}$. It follows that the entries of $M_{\text{PD}}$ are distributed as $(M_{\text{PD}})_{ij} \sim_{\text{i.i.d.}} \text{Bern}(Q)$ for all $i \neq j$ independently of $\text{diag}(M_{\text{PD}})$. Now note that the $k$ vectors $\text{diag}(M_{\text{PD}})_{F_t}$ for $t \in [k]$ are each exchangeable and have support size $s_1^t + \max\{ s_2^t - s_1^t, 0 \} = z_1^t + z_2^t + z_4^t$ where $z_1^t \sim \text{Bern}(p)$, $z_2^t \sim \text{Bin}(N/k - 1, p)$ and $s_2^t \sim \text{Bin}(n/k, Q)$ are independent. By the same argument as above, we have that
\begin{align*}
\TV\left( \mL(M_{\text{PD}}), \text{Bern}(Q)^{\otimes n \times n} \right) &= \TV\left( \mL(\text{diag}(M_{\text{PD}})), \text{Bern}(Q)^{\otimes n} \right) \\
&= \TV\left( \otimes_{t = 1}^k \mL\left( z_1^t + z_2^t + z_4^t \right), \text{Bin}(n/k, Q) \right) \\
&\le 4k \cdot \exp \left( - \frac{Q^2N^2}{48Pkn} \right)
\end{align*}
by Lemma \ref{lem:plantingdiagonals}. Since $M_{\text{PD}} \sim \mathcal{A}(\mG(N, q))$, this completes the proof of the lemma.
\end{proof}

\subsection{Imbalanced Binary Orthogonal Matrices and Sample Rotations}
\label{subsec:rotations}

In this section, we analyze the matrices $H_{r, t}$ constructed based on the incidences between points and hyperplanes in $\mathbb{F}_r^t$. The definition of $H_{r, t}$ can be found in Step 4 of $k$\pr{-pds-to-isgm} in Figure \ref{fig:isgmreduction}. We remark that a classic trick counting the number of ordered $d$-tuples of linearly independent vectors in $\mathbb{F}_r^t$ shows that the number of $d$-dimensional subspaces of $\mathbb{F}_r^t$ is
$$|\text{Gr}(d, \mathbb{F}_r^t)| = \frac{(r^t - 1)(r^t - r) \cdots (r^t - r^{d - 1})}{(r^d - 1)(r^d - r) \cdots (r^d - r^{d - 1})}$$
This implies that the number of hyperplanes in $\mathbb{F}_r^t$ is $\ell = \frac{r^t - 1}{r - 1}$, which justifies that the number of rows of $H_{r, t}$ is as described in $k$\pr{-pds-to-isgm}. The matrices $H_{r, t}$ are used to rotate the Gaussianized submatrix produced from $k$\pr{-pds} in Step 5 of $k$\pr{-pds-to-isgm} to produce the exactly imbalanced mixture structure. The crucial properties of $H_{r, t}$ are that they have orthogonal rows, are binary in the sense that they contain two distinct real values and contain a fraction of approximately $1/r$ negative values per column. All three properties are essential to the correctness of the reduction $k$\pr{-pds-to-isgm}. We establish these properties in the simple lemma below.

\begin{lemma}[Imbalanced Binary Orthogonal Matrices] \label{lem:orthogonalmatrices}
If $t \ge 2$ and $r \ge 2$ is prime, then $\left( \frac{r^t - 1}{r - 1} \right) \times r^t$ real matrix $H_{r, t}$ has orthonormal rows and each column of $H_{r, t}$ other than the column corresponding to $P_i = 0$ contains exactly $\frac{r^{t - 1} - 1}{r - 1}$ entries equal to $\frac{1 - r}{\sqrt{r^t(r - 1)}}$.
\end{lemma}

\begin{proof}
Let $r_i$ denote the $i$th row of $H_{r, t}$. First observe that
$$\| r_i \|_2^2 = (r^t - |V_i|) \cdot \frac{1}{r^t(r - 1)} + |V_i| \cdot \frac{(1 - r)^2}{r^t(r - 1)} = 1$$
since $|V_i| = r^{t - 1}$. Furthermore $V_i \cap V_j$ is a $(t - 2)$-dimensional subspace of $\mathbb{F}_r^t$ if $i \neq j$, which implies that $|V_i \cap V_j| = r^{t - 2}$ and $|V_i \cup V_j| = |V_i| + |V_j| - |V_i \cap V_j| = 2r^{t - 1} - r^{t - 2}$. Therefore if $i \neq j$,
\begin{align*}
\langle r_i, r_j \rangle &= (r^t - |V_i \cup V_j|) \cdot \frac{1}{r^t(r - 1)} + (|V_i \cup V_j| - |V_i \cap V_j|) \cdot \frac{1 - r}{r^t(r - 1)} + |V_i \cap V_j| \cdot \frac{(1 - r)^2}{r^t(r - 1)} \\
&= (r - 1)^2 \cdot \frac{1}{r^2(r - 1)} - 2(r - 1) \cdot \frac{1}{r^2} + \frac{r - 1}{r^2} = 0
\end{align*}
which shows that the rows of $H_{r, t}$ are orthonormal. Fix any two nonzero vectors $P_i, P_j \in \mathbb{F}_r^t$ and consider any invertible linear transformation $M : \mathbb{F}_r^t \to \mathbb{F}_r^t$ such that $M(P_i) = P_j$. The map $M$ induces a bijection between the hyperplanes containing $P_i$ and the hyperplanes containing $P_j$. In particular, it follows that each nonzero point in $\mathbb{F}_r^t$ is contained in the same number of hyperplanes. Each hyperplane contains $r^{t - 1} - 1$ nonzero points and thus the total number of incidences is $(r^{t - 1} - 1) \cdot \frac{r^t - 1}{r - 1}$, which implies that the number per nonzero point is $\frac{r^{t - 1} - 1}{r - 1}$. Each such incidence corresponds to a negative entry in the column for that point in $H_{r, t}$, implying that the column for each nonzero point $P_i$ contains exactly $\frac{r^{t - 1} - 1}{r - 1}$ negative entries.
\end{proof}

We now proceed to establish the total variation guarantees for sample rotation and subsampling as in Steps 5 and 6 in $k\pr{-pds-to-isgm}$, using these properties of $H_{r, t}$. In the rest of this section, let $\mathcal{A}$ denote the reduction $k\pr{-pds-to-isgm}$ with input $(G, E)$ where $E$ is a partition of $[N]$ and output $(X_1, X_2, \dots, X_n)$. We will need the following convenient upper bound on the total variation between two binomial distributions.

\begin{lemma} \label{lem:bintv}
Given $P \in [0, 1]$, $Q \in (0, 1)$ and $n \in \mathbb{N}$, it follows that
$$\TV\left( \textnormal{Bin}(n, P), \textnormal{Bin}(n, Q) \right) \le |P - Q| \cdot \sqrt{\frac{n}{2Q(1 - Q)}}$$
\end{lemma}

\begin{proof}
By applying the data processing inequality for $\TV$ to the function taking the sum of the coordinates of a vector, we have that
\begin{align*}
2 \cdot \TV\left( \textnormal{Bin}(n, P), \textnormal{Bin}(n, Q) \right)^2 &\le 2 \cdot \TV\left( \textnormal{Bern}(P)^{\otimes n}, \textnormal{Bern}(Q)^{\otimes n} \right)^2 \\
&\le \KL\left( \textnormal{Bern}(P)^{\otimes n}, \textnormal{Bern}(Q)^{\otimes n} \right) \\
&= n \cdot \KL\left( \textnormal{Bern}(P), \textnormal{Bern}(Q) \right) \\
&\le n \cdot \chi^2\left( \textnormal{Bern}(P), \textnormal{Bern}(Q) \right) \\
&= n \cdot \frac{(P - Q)^2}{Q(1 - Q)}
\end{align*}
The second inequality is an application of Pinsker's, the first equality is tensorization of $\KL$ and the third inequality is the fact that $\chi^2$ upper bounds $\KL$ by the concavity of $\log$. This completes the proof of the lemma.
\end{proof}

Let $\text{Hyp}(N, K, n)$ denote a hypergeometric distribution with $n$ draws from a population of size $N$ with $K$ success states. We will also need the upper bound on the total variation between hypergeometric and binomial distributions given by
$$\TV\left( \text{Hyp}(N, K, n), \text{Bin}(n, K/N) \right) \le \frac{4n}{N}$$
This bound is a simple case of finite de Finetti's theorem and is proven in Theorem (4) in \cite{diaconis1980finite}. The following lemma analyzes Steps 5 and 6 of $\mathcal{A}$.

\begin{lemma}[Sample Rotation] \label{lem:samplerotation}
Let $F$ be a fixed partition of $[kr^t]$ into $k$ parts of size $r^t$ and let $S \subseteq [m]$ be a fixed $k$-subset. Let $\mathcal{A}_{\textnormal{5-6}}$ denote Steps 5 and 6 of $k\pr{-pds-to-isgm}$ with input $M_{\textnormal{G}}$ and output $(X_1, X_2, \dots, X_n)$. Then for all $\tau \in \mathbb{R}$,
$$\TV\left( \mathcal{A}_{\textnormal{5-6}} \left( \tau \cdot \mathbf{1}_S \mathbf{1}_{\mU_{kr^t}(F)}^\top + \mN(0, 1)^{\otimes m \times kr^t} \right), \pr{isgm}_{H_1}(n, k, d, \mu, \epsilon) \right) \le \frac{4}{w} + \frac{k^2}{wn} + \frac{k}{w\sqrt{2n(r - 1)}}$$
where $\epsilon = 1/r$ and $\mu = \frac{\tau}{\sqrt{r^t(r - 1)}}$. Furthermore, it holds that $\mathcal{A}_{\textnormal{5-6}} \left( \mN(0, 1)^{\otimes m \times kr^t} \right) \sim \mN(0, I_d)^{\otimes n}$.
\end{lemma}

\begin{proof}
Let $T$ be a fixed $k$-subset of $[kr^t]$ such that $|T \cap F_i| = 1$ for each $i \in [k]$. Let $\ell = \frac{r^t - 1}{r - 1}$ and $F'$ be a fixed partition of $[k\ell]$ into $k$ parts of size $\ell$. We first consider the case where the input $M_{\text{G}}$ to $\mathcal{A}_{\textnormal{5-6}}$ is of the form
$$M_{\text{G}} = \tau \cdot \mathbf{1}_S \mathbf{1}_{T}^\top + G \quad \text{where } G \sim \mN(0, 1)^{\otimes m \times kr^t}$$
Since $M_{\text{G}}$ has independent entries, the submatrices $(M_{\text{G}})_{F_i}$ for each $i \in [k]$ are independent. Now observe that if $(H_{r, t})_{j}$ denotes the $j$th column of $H_{r, t}$, then we have that
$$(M_{\text{R}})_{F'_i} = (M_{\text{G}})_{F_i} H_{r, t}^\top = \tau \cdot \mathbf{1}_{S} \mathbf{1}_{T \cap F_i}^\top H_{r, t}^\top + G_{F_i} H_{r, t}^\top \sim \mL\left( \tau \cdot \mathbf{1}_{S} (H_{r, t})_{T \cap F_i}^\top + \mN(0, 1)^{\otimes m \times \ell} \right)$$
The distribution statement above follows from the joint Gaussianity and isotropy of the rows of $G_{F_i}$. More precisely, the entries of $G_{F_i} H_{r, t}^\top$ are linear combinations of the entries of $G_{F_i}$, which implies that they are jointly Gaussian. The fact that the rows of $H_{r, t}$ are orthonormal implies that the entries of $G_{F_i} H_{r, t}^\top$ are uncorrelated and each have unit variance. Therefore it follows that $G_{F_i} H_{r, t}^\top \sim \mN(0, 1)^{\otimes m \times \ell}$. If $h_{T, F, F'} \in \mathbb{R}^{k\ell}$ denotes the vector with $(h_{T, F, F'})_{F'_i} = (H_{r, t})_{T \cap F_i}$ for each $i \in [k]$, then it follows that
$$M_{\text{R}} \sim \mL\left( \tau \cdot \mathbf{1}_S h_{T, F, F'}^\top + \mN(0, 1)^{\otimes m \times k\ell} \right)$$
Observe that the columns of $M_{\text{R}}$ are independent and either distributed according $\mN(\mu \cdot \mathbf{1}_S, I_m)$ or $\mN(\mu' \cdot \mathbf{1}_S, I_m)$ where $\mu' =  \tau(1 - r)/\sqrt{r^t(r - 1)}$ depending on whether the entry of $h_{T, F, F'}$ at the index corresponding to the column is $1/\sqrt{r^t(r - 1)}$ or $(1 - r)/\sqrt{r^t(r - 1)}$.

Now let $s_{T, F}$ denote the number of entries of $h_{T, F, F'}$ that are equal to $1/\sqrt{r^t(r - 1)}$. Define $\mR_{n}(s)$ to be the distribution on $\mathbb{R}^n$ with a sample $v \sim \mR_{n}(s)$ generated by first choosing an $s$-subset $U$ of $[n]$ uniformly at random and then setting $v_i = 1/\sqrt{r^t(r - 1)}$ if $i \in U$ and $v_i = (1 - r)/\sqrt{r^t(r - 1)}$ if $i \not \in U$. Note that the number of columns distributed as $\mN(\mu \cdot \mathbf{1}_S, I_m)$ in $M_{\text{R}}$ chosen to be in $X$ is distributed according to $\text{Hyp}(k\ell, s_{T, F}, n)$. Step 6 of $\mathcal{A}$ therefore ensures that
$$M_{\text{R}} \sim \mL\left( \tau \cdot \mathbf{1}_{\mU_{d, k}} \mR_n(\text{Hyp}(k\ell, s_{T, F}, n))^\top + \mN(0, 1)^{\otimes d \times n} \right)$$
Observe that the data matrix for a sample from $\pr{isgm}_{H_1}(n, k, d, \mu, \epsilon)$ can be expressed similarly as
$$\pr{isgm}_{H_1}(n, k, d, \mu, \epsilon) = \mL\left( \tau \cdot \mathbf{1}_{\mU_{d, k}} \mR_n(\text{Bin}(n, 1 - \epsilon))^\top + \mN(0, 1)^{\otimes d \times n} \right)$$
where again we set $\mu = \tau/\sqrt{r^t(r - 1)}$. Now consider mixing over inputs $M_{\text{G}}$ where $T \sim \mU_{kr^t}(F)$. It follows that the output $M_{\text{R}}$ under this input is distributed as
$$\mathcal{A}_{\textnormal{5-6}} \left( \tau \cdot \mathbf{1}_S \mathbf{1}_{\mU_{kr^t}(F)}^\top + \mN(0, 1)^{\otimes m \times kr^t} \right) \sim \mL\left( \tau \cdot \mathbf{1}_{\mU_{d, k}} \mR_n\left(\text{Hyp}\left(k\ell, s_{\mU_{kr^t}(F), F}, n\right)\right)^\top + \mN(0, 1)^{\otimes d \times n} \right)$$
The conditioning property of $\TV$ in Fact \ref{tvfacts} now implies that
\begin{align*}
&\TV\left( \mathcal{A}_{\textnormal{5-6}} \left( \tau \cdot \mathbf{1}_S \mathbf{1}_{\mU_{kr^t}(F)}^\top + \mN(0, 1)^{\otimes m \times kr^t} \right), \pr{isgm}_{H_1}(n, k, d, \mu, \epsilon) \right) \\
&\quad \quad \le \TV\left(\text{Bin}(n, 1 - \epsilon), \text{Hyp}\left(k\ell, s_{\mU_{kr^t}(F), F}, n\right) \right)
\end{align*}
By Lemma \ref{lem:orthogonalmatrices}, $(H_{r, t})_{T \cap F_i}$ contains $\frac{r^t - 1}{r - 1} - \frac{r^{t - 1} - 1}{r - 1} = r^{t - 1}$ entries equal to $1/\sqrt{r^t(r - 1)}$ as long $T \cap F_i$ is not the column index corresponding to the zero point in $\mathbb{F}_r^t$. If it does correspond to zero, then $(H_{r, t})_{T \cap F_i}$ contains no entries equal to $1/\sqrt{r^t(r - 1)}$. Now note that if $T \sim \mU_{kr^t}(F)$, then $T \cap F_i$ corresponds to zero with probability $1/r^t$ for each $i \in [k]$. Furthermore, these events are independent. This implies that $r^{1 - t} \cdot s_{\mU_{kr^t}(F), F}$ is the number of $F_i$ such that $T \cap F_i$ does not correspond to zero, and therefore distributed as $r^{1 - t} \cdot s_{\mU_{kr^t}(F), F} \sim \text{Bin}(k, 1 - 1/r^t)$. Thus
$$\bP\left[ s_{\mU_{kr^t}(F), F} = kr^{t - 1} \right] = \left( 1 - \frac{1}{r^t} \right)^k \ge 1 - \frac{k}{r^t}$$
Applying the conditioning on an event property of $\TV$ from Fact \ref{tvfacts} now yields that
$$\TV\left( \text{Hyp}\left(k\ell, s_{\mU_{kr^t}(F), F}, n\right), \text{Hyp}\left(k\ell, kr^{t - 1}, n\right) \right) \le \bP\left[ s_{\mU_{kr^t}(F), F} \neq kr^{t - 1} \right] \le \frac{k}{r^t} \le \frac{k^2}{wn}$$
since $wn \le k \ell$ by definition and $\ell \le r^t$. By the application of Theorem (4) in \cite{diaconis1980finite} to hypergeometric distributions above, we also have
$$\TV\left( \text{Hyp}(k\ell, kr^{t - 1}, n), \text{Bin}(n, r^{t - 1}/\ell) \right) \le \frac{4n}{k\ell} \le \frac{4}{w}$$
Recall that $\epsilon = 1/r$ and note that Lemma \ref{lem:bintv} implies that
\begin{align*}
\TV\left( \text{Bin}(n, r^{t - 1}/\ell), \text{Bin}(n, 1 - \epsilon) \right) &\le \left| \frac{r^{t - 1}}{\frac{r^t - 1}{r - 1}} - \left( 1 - \frac{1}{r} \right) \right| \cdot \sqrt{\frac{nr^2}{2(r - 1)}} \\
&= \frac{1}{r^t - 1} \cdot \sqrt{\frac{n(r - 1)}{2}} \\
&= \frac{1}{\ell} \cdot \sqrt{\frac{n}{2(r - 1)}} \le \frac{k}{w\sqrt{2n(r - 1)}}
\end{align*}
where the last inequality is again since $wn \le k\ell$. Applying the triangle inequality now yields that
$$\TV\left(\text{Bin}(n, 1 - \epsilon), \text{Hyp}\left(k\ell, s_{\mU_{kr^t}(F), F}, n\right) \right) \le \frac{4}{w} + \frac{k^2}{wn} + \frac{k}{w\sqrt{2n(r - 1)}}$$
Now consider applying the above argument with $\tau = 0$. It follows that
$$\mathcal{A}_{\textnormal{5-6}} \left( \mN(0, 1)^{\otimes m \times kr^t} \right) \sim \mN(0, 1)^{\otimes d \times n} = \mN(0, I_d)^{\otimes n}$$
which completes the proof of the lemma.
\end{proof}

We now combine these lemmas to complete the proof of Theorem \ref{thm:isgmreduction}.

\begin{proof}[Proof of Theorem \ref{thm:isgmreduction}]
We apply Lemma \ref{lem:tvacc} to the steps $\mathcal{A}_i$ of $\mathcal{A}$ under each of $H_0$ and $H_1$ to prove Theorem \ref{thm:isgmreduction}. Define the steps of $\mathcal{A}$ to map inputs to outputs as follows
$$(G, E) \xrightarrow{\mathcal{A}_1} (M_{\text{PD1}}, F') \xrightarrow{\mathcal{A}_2} (M_{\text{PD2}}, F) \xrightarrow{\mathcal{A}_3} (M_{\text{G}}, F) \xrightarrow{\mathcal{A}_{\text{5-6}}} (X_1, X_2, \dots, X_n)$$
We first prove the desired result in the case that $H_1$ holds. Consider Lemma \ref{lem:tvacc} applied to the steps $\mathcal{A}_i$ above and the following sequence of distributions
\allowdisplaybreaks
\begin{align*}
\mathcal{P}_0 &= \mG_E(N, k, p, q) \\
\mathcal{P}_1 &= \mathcal{M}(m, \mU_m(F'), p, Q) \\
\mathcal{P}_2 &=\mathcal{M}(m, kr^t, \mU_{m, k}, \mU_{kr^t}(F), p, Q) \\
\mathcal{P}_3 &=\mathcal{M}\left(m, kr^t, \mU_{m, k}, \mU_{kr^t}(F), \mN\left(\sqrt{r^t(r - 1)} \cdot \mu, 1\right), \mN(0, 1) \right) \\
&= \sqrt{r^t(r - 1)} \cdot \mu \cdot \mathbf{1}_{\mU_{m, k}} \mathbf{1}_{\mU_{kr^t}(F)}^\top + \mN(0, 1)^{\otimes m \times kr^t} \\
\mathcal{P}_{\text{5-6}} &= \pr{isgm}_{H_1}(n, k, d, \mu, \epsilon)
\end{align*}
As in the statement of Lemma \ref{lem:tvacc}, let $\epsilon_i$ be any real numbers satisfying $\TV\left( \mathcal{A}_i(\mP_{i-1}), \mP_i \right) \le \epsilon_i$ for each step $i$. A direct application of Lemma \ref{lem:submatrix} implies that we can take
$$\epsilon_1 = 4k \cdot \exp\left( - \frac{Q^2N^2}{48pkm} \right) + \sqrt{\frac{C_Q k^2}{2m}}$$
where $C_Q = \max\left\{ \frac{Q}{1 - Q}, \frac{1 - Q}{Q} \right\}$. Note that, by construction, the step $\mathcal{A}_2$ is exact and we can take $\epsilon_2 = 0$. Consider applying Lemma \ref{lem:gaussianize} and averaging over $S \sim \mU_{m, k}$ and $T \sim \mU_{kr^t}(F)$ using the conditioning property from Fact \ref{tvfacts}. This yields that we can take $\epsilon_3 = O((mkr^t)^{-1/2}) = O(N^{-1})$. Applying Lemma \ref{lem:samplerotation} while similarly averaging over $S \sim \mU_{m, k}$ yields that we can take
$$\epsilon_{\text{5-6}} = \frac{4}{w} + \frac{k^2}{wn} + \frac{k}{w\sqrt{2n(r - 1)}}$$
By Lemma \ref{lem:tvacc}, we therefore have that
$$\TV\left( \mathcal{A}\left( \mG_E(N, k, p, q) \right), \pr{isgm}_{H_1}(n, k, d, \mu, \epsilon) \right) = O\left( w^{-1} + \frac{k^2}{wN} + \frac{k}{\sqrt{N}} + e^{-\Omega(N^2/km)} + N^{-1} \right)$$
which proves the desired result in the case of $H_1$. Now consider the case that $H_0$ holds and Lemma \ref{lem:tvacc} applied to the steps $\mathcal{A}_i$ and the following sequence of distributions
\allowdisplaybreaks
\begin{align*}
\mathcal{P}_0 &= \mG(N, q) \\
\mathcal{P}_1 &= \text{Bern}(Q)^{\otimes m \times m} \\
\mathcal{P}_2 &= \text{Bern}(Q)^{\otimes m \times kr^t} \\
\mathcal{P}_3 &= \mN(0, 1)^{\otimes m \times kr^t} \\
\mathcal{P}_{\text{5-6}} &= \mN(0, I_d)^{\otimes n}
\end{align*}
As above, Lemmas \ref{lem:submatrix}, \ref{lem:gaussianize} and \ref{lem:samplerotation} imply that we can take
$$\epsilon_1 = 4k \cdot \exp\left( - \frac{Q^2N^2}{48pkm} \right), \quad \epsilon_2 = 0, \quad \epsilon_3 = O(N^{-1}) \quad \text{and} \quad \epsilon_{\text{5-6}} = 0$$
By Lemma \ref{lem:tvacc}, we therefore have that
$$\TV\left( \mathcal{A}\left( \mG(N, q) \right), \mN(0, I_d)^{\otimes n} \right) = O\left( e^{-\Omega(N^2/kn)} + N^{-1} \right)$$
which completes the proof of the theorem.
\end{proof}

\section{Lower Bounds for Robust Sparse Mean Estimation}
\label{sec:robust}

In this section, we apply the reduction $k$-\pr{pds-to-isgm} (Fig.~\ref{fig:isgmreduction}) to deduce our main statistical-computational gaps for robust sparse mean estimation. We begin by showing that a direct application of $k$-\pr{pds-to-isgm} yields a lower bound of $n = \Omega(\epsilon^3 k^2)$ for polynomial-time robust sparse estimation within $\ell_2$ distance $\tau = o(\sqrt{\epsilon/(\log n)^{1 + c}})$ for any fixed $c > 0$. When $\epsilon = 1/\text{polylog}(n)$, this yields the optimal sample lower bound of $n = \tilde{\Omega}(k^2)$ for estimation within the $\ell_2$ minimax rate of $\tau = \Tilde{\Theta}(\epsilon)$. When $\epsilon = (\log n)^{-\omega(1)}$, our reduction shows that robust sparse mean estimation continues to have a large statistical-computational gap even when the task only requires estimation within $\ell_2$ distance $\tau \approx \sqrt{\epsilon}$, which is far above the minimax rate.

We remark that $(n, k, d, \epsilon)$ in the next theorem is an arbitrary sequence of parameters satisfying the given conditions. Given these parameters, we construct $(N, k')$ such that $k' = o(\sqrt{N})$ and $k$-\pr{pds-to-isgm} reduces from $k\pr{-pc}$  with $N$ vertices and subgraph size $k'$ to $\pr{rsme}(n, k, d, \tau, \epsilon)$. An algorithm solving $\pr{rsme}(n, k, d, \tau, \epsilon)$ would then that yield that $k\pr{-pc}$ can be solved on this constructed sequence of parameters $(N, k')$, contradicting our form of the $k\pr{-pc}$ conjecture. When $N$ is appropriately close to a number of the form $k' r^t$, then the resulting $\tau$ satisfies
$$\tau \asymp \sqrt{\frac{\epsilon}{w \log n}}$$
where $w$ denotes an arbitrarily slow-growing function of $n$ tending to infinity. Note that this $\tau$ is much larger than $\epsilon$ for $\epsilon = o(1/\log n)$. When $N$ is far from any $k' r^t$, then $\tau$ can degrade to $\tau \asymp \epsilon/\sqrt{w \log n}$, which is never larger than $\epsilon$ and thus yields vacuous lower bounds for $\pr{rsme}$.\footnote{This is because our reduction sets $r = \Theta(\epsilon^{-1})$ and maps to the signal level $\tau \asymp \sqrt{k'/(r^{t+1} \log n)}$ where $r^t$ is the smallest power of $r$ greater than $(p/Q + 1)N/k'$ and $k' = o(\sqrt{N})$ for the starting \pr{pc} instance to be hard. If $(p/Q + 1)N/k'$ is far from the next smallest power of $r$, it is possible that $r^t \asymp Nr/k'$ which implies that $\tau = o(1/(r\sqrt{\log n})) = O(\epsilon/\sqrt{\log n})$. However, for our choice of parameters, it will hold that $r^t \asymp N/k'$ and $\tau$ will instead be close to $1/\sqrt{r\log n} \asymp \sqrt{\epsilon/\log n}$.} This is the number theoretic subtlety alluded to in Section \ref{sec:redisgm}. This is a nonissue for the right choices of $(N, k')$, as shown in the derivation below. For clarity, we go through this first lower bound in detail and include fewer details in subsequent similar theorems.

\begin{theorem}[General Lower Bound for $\pr{rsme}$] \label{thm:rsmefull}
Let $(n, k, d, \epsilon)$ be any parameters satisfying that $k^2 = o(d)$, $\epsilon \in (0, 1)$ and $n$ satisfies that $n = o(\epsilon^3 k^2)$ and $n = \Omega(k)$. If $c > 0$ is some fixed constant, then there is a parameter $\tau = \Omega(\sqrt{\epsilon/(\log n)^{1 + c}})$ such that any randomized polynomial time test for $\pr{rsme}(n, k, d, \tau, \epsilon)$ has asymptotic Type I$+$II error at least 1 assuming either the $k$\pr{-pc} conjecture or the $k$\pr{-pds} conjecture for some fixed edge densities $0 < q < p \le 1$.
\end{theorem}

\begin{proof}
This theorem will follow from a careful selection of parameters with which to apply $k$-\pr{pds-to-isgm} from Theorem \ref{thm:isgmreduction}. Assume the $k$\pr{-pds} conjecture for some fixed edge densities $0 < q < p \le 1$. Let $w = w(n)$ be an arbitrarily slow-growing function of $n$ and let $Q = 1 - \sqrt{(1 - p)(1 - q)} + \mathbf{1}_{\{ p = 1\}} \left( \sqrt{q} - 1 \right)$. Note that $Q \in (0, 1)$ and is constant. Now define parameters as follows:
\begin{enumerate}
\item Let $r$ be a prime number with $\epsilon^{-1} < r = O(\epsilon^{-1})$, which can be found in $\text{poly}(\epsilon^{-1})$ time.
\item Let $t$ be such that $r^t$ is the largest power of $r$ less than $wk(1 + \frac{p}{Q})$.
\item Set $k$\pr{-pds} parameters $k' = \lfloor r^t w^{-1} (1 + \frac{p}{Q})^{-1} \rfloor$ and $N = wk'^2$.
\item Set the mean parameter $\mu$ to be
$$\mu = \frac{\delta}{2 \sqrt{6\log (k'r^t) + 2\log (p - Q)^{-1}}} \cdot \frac{1}{\sqrt{r^t(r - 1)}}$$
where $\delta = \min \left\{ \log \left( \frac{p}{Q} \right), \log \left( \frac{1 - Q}{1 - p} \right) \right\}$.
\end{enumerate}
By construction, we have that $\left( \frac{p}{Q} + 1 \right) N \le k'r^t$. For a slow-growing enough choice of $w$ and large enough $n$, we have
\begin{align*}
&\left( \frac{p}{Q} + 1 \right) N + k' \le 2w k^2 \left( 1 + \frac{p}{Q} \right) \le d \qquad \text{and} \\
&wn \le \frac{1}{2} \cdot \epsilon^3 k^2 \le \frac{r^{2(t + 1)} \epsilon^3}{2w^2\left(1 + \frac{p}{Q} \right)^2} \le \frac{r^2(r - 1)\epsilon^3}{w\left(1 + \frac{p}{Q} \right)} \cdot \frac{k'(r^t - 1)}{r - 1} \le \frac{k'(r^t - 1)}{r - 1}
\end{align*}
by the definitions above. Now consider applying $k$-\pr{pds-to-isgm} to map from $k\pr{-pds}(N, k', p, q)$ to $\pr{isgm}(n, k', d, \mu, 1/r)$. The inequalities above guarantee that we have met the conditions needed to apply Theorem \ref{thm:isgmreduction}. Note that the total variation upper bound in Theorem \ref{thm:isgmreduction} tends to zero since $k'^2/N = w^{-1} = o(1)$ and $N^2/k'n = \Omega(k') = \omega(1)$. 

Now observe that $\pr{isgm}(n, k', d, \mu, 1/r)$ is an instance of $\pr{rsme}(n, k, d, \tau, \epsilon)$ since $k' \le k$ and $1/r < \epsilon$. More precisely, it is an instance with mean vector $\mu \cdot \mathbf{1}_S$, where $S$ is a $k'$-subset of $[d]$ chosen uniformly at random, and outlier distribution $\mD_O = \pr{mix}_{\epsilon^{-1} r^{-1}} \left( \mN(\mu \cdot \mathbf{1}_S, I_d), \mN(\mu' \cdot \mathbf{1}_S, I_d) \right)$ where $\mu'$ is such that $(1 - r^{-1}) \mu + r^{-1} \cdot \mu' = 0$. Now note that
$$\tau = \| \mu \cdot \mathbf{1}_S \|_2 = \mu \sqrt{k'} \asymp \frac{1}{\sqrt{\log(n)}} \cdot \frac{1}{\sqrt{r^t(r - 1)}} \cdot \sqrt{k'} \asymp \sqrt{\frac{\epsilon}{w \log n}}$$
since $\log(k'r^t) = \Theta(\log n)$ and $k' = \lceil r^t w^{-1} (1 + \frac{p}{Q})^{-1} \rceil$ implies that $r^t = \Theta(w k')$. This $\tau$ satisfies that $\tau = \Omega(\sqrt{\epsilon/(\log n)^{1 + c}})$ as long as $w = O((\log n)^c)$. Now suppose that some randomized polynomial time test $\mathcal{A}$  for $\pr{rsme}(n, k, d, \tau, \epsilon)$ has asymptotic Type I$+$II error less than 1. By Lemma \ref{lem:3a} and the reduction above, this implies that there is a randomized polynomial time test for $k\pr{-pds}$ on the sequence of inputs $(N, k', p, q)$ with asymptotic Type I$+$II error less than 1. This contradicts the $k$\pr{-pds} conjecture and proves the theorem.
\end{proof}

\begin{figure}[t!]
\begin{algbox}
\textbf{Algorithm} \pr{isgm-Sample-Cloning}

\vspace{1mm}

\textit{Inputs}: $\pr{isgm}$ samples $X_1, X_2, \dots, X_n \in \mathbb{R}^d$, blowup parameter $\ell$
\begin{enumerate}
\item Set $X_i^0 = X_i$ for each $1 \le i \le n$.
\item For $j = 1, 2, \dots, \ell$ do:
\begin{enumerate}
\item[(1)] Sample $G_1, G_2, \dots, G_{2^{j - 1}n} \sim_{\text{i.i.d.}} \mN(0, I_d)$.
\item[(2)] For each $1 \le i \le 2^{j - 1} n$, form $X_i^j$ and $X_{2^{j - 1}n + i}^j$ as
$$X_i^j = \frac{1}{\sqrt{2}} \left( X_i^{j - 1} + G_i \right) \quad \text{and} \quad X_{2^{j - 1}n + i}^j = \frac{1}{\sqrt{2}} \left( X_i^{j - 1} - G_i \right)$$
\end{enumerate}
\item Output a subset of $X_1^{\ell}, X_2^{\ell}, \dots, X_{2^\ell n}^\ell$ of size $n'$ chosen uniformly at random.
\end{enumerate}
\vspace{1mm}

\end{algbox}
\caption{Sample cloning subroutine in the reduction from a planted dense subgraph instance to robust sparse mean estimation.}
\label{fig:isgmcloning}
\end{figure}

For small $\epsilon$ with $\epsilon = (\log n)^{-\omega(1)}$, the mean parameter $\tau$ above is $\sqrt{\epsilon}$ up to subpolynomial factors in $\epsilon$. This value of $\tau$ is much larger than the $\omega(\epsilon)$ it needs to be to show lower bounds for $\pr{rsme}$. Thus the reduction $k$-\pr{pds-to-isgm} actually shows lower bounds at small $\epsilon$ for weak estimators that can only estimate up to $\ell_2$ distance $\tau \approx \sqrt{\epsilon}$. When $\epsilon = (\log n)^{c}$ where $c > 1$ in the theorem above, it yields the optimal $k$-to-$k^2$ gap in $\pr{rsme}$ up to polylogarithmic factors. This is stated in the corollary below.

\begin{corollary}[Optimal Statistical-Computational Gaps in $\pr{rsme}$] \label{thm:rsmeopt}
Let $(n, k, d, \epsilon)$ be any parameters satisfying that $k^2 = o(d)$, $\epsilon = \Theta((\log n)^{-c})$ for some constant $c > 1$ and $n$ satisfies that $n = o(k^2 (\log n)^{-3c})$ and $n = \Omega(k)$. Then there is a parameter $\tau = \omega(\epsilon)$ such that any randomized polynomial time test for $\pr{rsme}(n, k, d, \tau, \epsilon)$ has asymptotic Type I$+$II error at least 1 assuming either the $k$\pr{-pc} conjecture or the $k$\pr{-pds} conjecture for some fixed edge densities $0 < q < p \le 1$.
\end{corollary}

We remark that in intermediate parameter regimes where $\epsilon = (\log n)^{-\omega(1)}$ is not yet polynomially small in $n$, such as $\epsilon = e^{-\Theta(\sqrt{\log n})}$, our result essentially shows a $k$-to-$k^2$ statistical-computational gap for $\pr{rsme}$ at the weak $\ell_2$ estimation rate of $\tau = \tilde{\Theta}(\sqrt{\epsilon})$. It is in these parameter regimes where our lower bounds for $\pr{rsme}$ are strongest.

In the case where $\epsilon$ is polynomially small in $n$, the sample lower bound of $n = \Omega(\epsilon^3 k^2)$ in Theorem \ref{thm:rsmefull} degrades with $\epsilon$. We now show that the high $\tau \approx \sqrt{\epsilon}$ produced by our reduction can be traded off for sharper bounds in $n$ using the simple post-processing subroutine \pr{isgm-Sample-Cloning} in Figure \ref{fig:isgmcloning}. Its important properties are captured in the following lemma.

\begin{lemma}[Sample Cloning] \label{lem:samplecloning}
Let $\mathcal{A}$ denote $\pr{isgm-Sample-Cloning}$ applied with blowup parameter $\ell$ and let $(Y_1, Y_2, \dots, Y_{2^\ell n})$ be the output of $\mathcal{A}(X_1, X_2, \dots, X_n)$. Then we have that
\begin{itemize}
\item If $X_1, X_2, \dots, X_n$ are independent and exactly $m$ of the $n$ samples $X_1, X_2, \dots, X_n$ are distributed according to $\mN(\mu \cdot \mathbf{1}_S, I_d)$ and the rest from $\mN(\mu' \cdot \mathbf{1}_S, I_d)$, then the $Y_1, Y_2, \dots, Y_{2^\ell n}$ are independent and exactly $2^\ell m$ of $Y_1, Y_2, \dots, Y_{2^\ell n}$ are distributed according to $\mN(2^{-\ell/2}\mu \cdot \mathbf{1}_S, I_d)$ and the rest from $\mN(2^{-\ell/2}\mu' \cdot \mathbf{1}_S, I_d)$.
\item If $X_1, X_2, \dots, X_n \sim_{\text{i.i.d.}} \mN(0, I_d)$, then the $Y_1, Y_2, \dots, Y_{2^\ell n} \sim_{\text{i.i.d.}} \mN(0, I_d)$.
\end{itemize}
\end{lemma}

\begin{proof}
These are both simple consequences of the fact that if $X \sim \mN(\mu \cdot \mathbf{1}_S, I_d)$ and $G \sim \mN(0, I_d)$ are independent then
$$X_1 = \frac{1}{\sqrt{2}} \left( X + G_i \right) \quad \text{and} \quad X_2 = \frac{1}{\sqrt{2}} \left( X - G_i \right)$$
are independent and satisfy $(X_1, X_2) \sim \mN(\mu \cdot \mathbf{1}_S/\sqrt{2}, I_d)^{\otimes 2}$. Iteratively applying this fact proves the lemma.
\end{proof}

We now use $\pr{isgm-Sample-Cloning}$ to strengthen Theorem \ref{thm:rsmeopt} as follows. This requires a slightly more stringent choice of the parameter $k$ than in Theorem \ref{thm:rsmeopt} to initially improve the lower bound to $n = \Omega(\epsilon k^2)$ before applying $\pr{isgm-Sample-Cloning}$. This choice of $k$ renders the number-theoretic issue alluded to above trivial. We omit details that are the same as in Theorem \ref{thm:rsmeopt}. Note that $\pr{rsme}$ is formulated in Section \ref{sec:summary} in Huber's $\epsilon$-contamination model. Let $\pr{rsme-c}$ be the variant of $\pr{rsme}$ instead defined in the $\epsilon$-corruption model. Then we have the following theorem.

\begin{theorem}[Lower Bound Tradeoff with Estimation Accuracy for $\pr{rsme}$]
Fix some $\alpha \in (0, 1)$ and suppose that $\epsilon = O(n^{-c})$ for some constant $c > 0$. Assume either the $k$\pr{-pc} conjecture or the $k$\pr{-pds} conjecture for some fixed edge densities $0 < q < p \le 1$. Then any test solving $\pr{rsme-c}(n, k, d, \tau, \epsilon)$ with $\tau = \tilde{\Omega}(\epsilon^{1 - \alpha/2})$ has asymptotic Type I$+$II error at least 1 if $n = o(\epsilon^\alpha k^2)$, $k^2 = o(d)$ and $n = \Omega(k)$.
\end{theorem}

\begin{proof}
Set parameters identically as in Theorem \ref{thm:rsmeopt}, except let $k = r^t$ for the choice of $r$ and $t$ and let $1/r < \epsilon/2$ but $1/r = \Omega(\epsilon)$. All parameter calculations remain identical to Theorem \ref{thm:rsmeopt}, except it only needs to hold that $n = o(\epsilon k^2)$ instead of $n = o(\epsilon^3 k^2)$ to satisfy the conditions to apply Theorem \ref{thm:isgmreduction}. This is because if $n = o(\epsilon k^2)$ and $k = r^t$, then we have that
$$wn \le \frac{1}{4} \cdot \epsilon k^2 \le \frac{r^{2t- 1}}{2w^2\left(1 + \frac{p}{Q} \right)^2} \le \frac{k'(r^t - 1)}{r - 1}$$
for large enough $n$ since $w$ tends to infinity. Therefore the same application of Theorem \ref{thm:isgmreduction} yields that $k$-\pr{pds-to-isgm} produces an instance of $\pr{isgm}(n, k', d, \mu, 1/r)$ with
$\tau \asymp \sqrt{\epsilon/w\log n}$. Applying Lemma \ref{lem:samplecloning} yields that if $\pr{isgm-Sample-Cloning}$ is then applied with blowup factor $\ell$, we have that we arrive at an instance of $\pr{rsme-c}(2^\ell n, k, d, 2^{-\ell/2} \cdot \tau, \epsilon)$ within total variation $o(1)$. Note that here, we have used the fact that since $1/r < \epsilon/2$, then the concentration of $\text{Bin}(n, 1/r)$ implies with high probability that the number of corrupted sampled is at most $\epsilon n$ before applying $\pr{isgm-Sample-Cloning}$. Note that $\pr{isgm-Sample-Cloning}$ will preserve this fact. Suppose that $\ell$ is chosen such that $2^\ell = \Theta(\epsilon^{\alpha - 1})$, then it follows that $2^\ell n = o(\epsilon^\alpha k^2)$ and $2^{-\ell/2} \cdot \tau = \tilde{\Omega}(\epsilon^{1 - \alpha/2})$. Applying Lemma \ref{lem:3a} completes the proof of this theorem.
\end{proof}

\section{Lower Bounds for Semirandom Single Community Recovery}
\label{sec:semirandom}


\begin{figure}[t!]
\begin{algbox}
\textbf{Algorithm} $k$\textsc{-pds-to-semi-cr}

\vspace{1mm}

\textit{Inputs}: $k$\pr{-pds} instance $G \in \mG_N$ with dense subgraph size $k$ that divides $N$, partition $E$ of $[N]$ and edge probabilities $0 < q < p \le 1$, blowup factor $\ell$, target number of vertices $n \ge m$ where $m$ is the smallest multiple of $(3^\ell - 1)k$ larger than $\left( \frac{p}{Q} + 1 \right) N$ where $Q = 1 - \sqrt{(1 - p)(1 - q)} + \mathbf{1}_{\{ p = 1\}} \left( \sqrt{q} - 1 \right)$ and $n \le \text{poly}(N)$

\begin{enumerate}
\item \textit{Symmetrize and Plant Diagonals}: Compute $M_{\text{PD}} \in \{0, 1\}^{m \times m}$ with partition $F$ of $[m]$ as
$$M_{\text{PD}} \gets \pr{To-}k\textsc{-Partite-Submatrix}(G)$$
applied with initial dimension $N$, edge probabilities $p$ and $q$ and target dimension $m$.
\item \textit{Gaussianize}: Compute $M_{\text{G}} \in \mathbb{R}^{m \times m}$ as $M_{\text{G}} \gets \textsc{Gaussianize}(M_{\text{PD}})$ applied with probabilities $p$ and $Q$ and mean parameters
$$\mu_{ij} = \mu = \frac{1}{2 \sqrt{6\log m + 2\log (p - Q)^{-1}}} \cdot \min \left\{ \log \left( \frac{p}{Q} \right), \log \left( \frac{1 - Q}{1 - p} \right) \right\}$$
for each $i, j \in [m]$.
\item \textit{Partition into Blocks and Pad}: Form the rotation matrix $H_{3, \ell}$ as in Figure \ref{fig:isgmreduction}. Form the matrix $M_{\text{P}} \in \mathbb{R}^{m' \times m'}$ where $m' = 3^\ell ks$ where $s = m/(3^\ell - 1)k$ by embedding $M_{\text{G}}$ as its upper left submatrix and sampling all other entries i.i.d. from $\mN(0, 1)$. Further partition each $F_i$ into $s$ blocks of size $3^{\ell} - 1$ and add one of the new $ks$ indices to each $F_i$. Now reorder the indices in each block so that the new index corresponds to the index of the column for the point $P_i = 0$ in $H_{3, \ell}$. Let $[m'] = F_1' \cup F_2' \cup \cdots \cup F'_{ks}$ be the partition of the row and column indices of $M_{\text{P}}$ induced by the blocks.
\item \textit{Rotate in Blocks}: Let $F''$ be a partition of $[m'']$ into $ks$ equally sized parts where $m'' = \frac{1}{2} (3^{\ell} - 1) ks$. Now compute the matrix $M_{\text{R}} \in \mathbb{R}^{m'' \times m''}$ as
$$(M_{\text{R}})_{F''_i, F''_j} = H_{3, \ell} (M_{\text{P}})_{F_i', F_j'} H_{3, \ell}^\top \quad \text{for each } i, j \in [ks]$$
where $Y_{A, B}$ denotes the submatrix of $Y$ restricted to the entries with indices in $A \times B$.
\item \textit{Threshold and Output}: Now construct the graph $G'$ with vertex set $[m'']$ such that for each $i > j$ with $i, j \in [m'']$, we have
$$\{i, j \} \in E(G') \textnormal{ if and only if } (M_{\text{R}})_{ij} \ge \frac{\mu}{2 \cdot 3^\ell}$$
Add $n - m''$ new vertices to $G'$ such that each edge incident to a new vertex is included in $E(G')$ independently with probability $1/2$. Randomly permute the vertex labels of $G'$ and output the resulting graph.
\end{enumerate}
\vspace{0.5mm}

\end{algbox}
\caption{Reduction from $k$-partite planted dense subgraph to semirandom community recovery.}
\label{fig:semirandreduction}
\end{figure}

In this section, we prove our second main result showing the $k$\pr{-pc} and $k$\pr{-pds} conjectures imply the $\pr{pds}$ Recovery Conjecture under a semirandom adversary in the regime of constant ambient edge density. Our reduction from $k$\pr{-pds} to \pr{semi-cr} is shown in Figure \ref{fig:semirandreduction}. On a high level, $k$\pr{-pds-to-semi-cr} can be interpreted as rotating by $H_{3, \ell}$ to effectively spread the signal in the planted dense subgraph out by simultaneously expanding its size from $k$ to $\Theta(3^\ell k)$ while decreasing its planted edge density. Furthermore, this rotation spreads the signal at the sharper rate from the $\pr{pds}$ Recovery Conjecture as opposed to the slower detection rate. In doing so, the reduction also produces monotone noise in the rest of the graph that can be simulated by a semirandom adversary.

Our reduction begins with the same first two steps as in the reduction to \pr{isgm}, by symmetrizing, planting diagonals and Gaussianizing. The total variation guarantees of these steps were already established in Section \ref{subsec:tosubmatrix}. The third step breaks the resulting matrix into blocks within each part $F_i$ and adds one row and one column to each block such that: (1) the added row and column have index in the block corresponding to the column index of $P_i = 0$ in $H_{3, \ell}$; and (2) the entries of both are independently sampled from $\mN(0, 1)$. The fourth step rotates according to imbalanced binary orthogonal matrices $H_{3, \ell}$ along both rows and columns. The last step produces a graph by appropriately thresholding the above-diagonal entries of the resulting matrix.

Throughout the remainder of this section, let $\mathcal{A}$ denote the reduction $k\pr{-pds-to-semi-cr}$. We first formally introduce the key intermediate distributions on graphs that our reduction produces and which we will show can be simulated by a semirandom adversary.

\begin{definition}[Target Graph Distributions]
Given positive integers $k \le m \le n$ and $\mu_1, \mu_2, \mu_3 \in (0, 1)$, let $\pr{tg}_{H_1}(n, k, k', m, \mu_1, \mu_2, \mu_3)$ be the distribution over $G \in \mG_n$ sampled as follows:
\begin{enumerate}
\item Choose a subset $V \subseteq [n]$ of size $|V| = m$ uniformly at random and then choose two disjoint subsets $S \subseteq V$ and $S' \subseteq V$ of sizes $|S| = k$ and $|S'| = k'$, respectively, uniformly at random.
\item Include the edge $\{i, j\}$ in $E(G)$ independently with probability $p_{ij}$ where
$$p_{ij} = \left\{ \begin{array}{ll} 1/2 &\text{if } (i, j) \in S'^2 \text{ or } (i, j) \not \in V^2 \\ 1/2 - \mu_1 &\text{if } (i, j) \in V^2 \backslash (S \cup S')^2 \\ 1/2 - \mu_2 &\text{if } (i, j) \in S \times S' \text{ or } (i, j) \in S' \times S \\ 1/2 + \mu_3 &\text{if } (i, j) \in S^2 \end{array} \right.$$
\end{enumerate}
Furthermore, let $\pr{tg}_{H_0}(n, m, \mu_1)$ be the distribution over $G \in \mG_n$ sampled by choosing $V \subseteq [n]$ with $|V| = m$ uniformly at random and including $\{i, j\}$ in $E(G)$ with probability $1/2 - \mu_1$ if $(i, j) \in V^2$ and probability $1/2$ otherwise.
\end{definition}

We now establish the desired Markov transition properties for the block rotations and thresholding procedures in Steps 3, 4 and 5. We then will combine this with lemmas in Section \ref{subsec:tosubmatrix} to provet our lower bound for $\pr{semi-cr}$. We remark that the block-wise padding in Step 3 is needed in the next lemma. In the proof of Lemma \ref{lem:samplerotation}, we were able to condition on no planted indices corresponding to columns with $P_i = 0$ upon rotating without a loss in total variation since these correspondences occurred with low probability. Here, this is no longer possible because rotations are carried out block-wise and the number of blocks is much larger than the number of blocks in the partition. This issue is resolved by adding in a new index corresponding to $P_i = 0$ in each block that is guaranteed not to be planted. The fact that no planted index corresponds to $P_i = 0$ is zero ensures that the number of vertices in the planted subgraph of the resulting semirandom instance is $\frac{1}{2}(3^{\ell - 1} - 1)k$. This fact is used in the proof of the following lemma. Recall that $\Phi(x) = \int_{-\infty}^x e^{-t^2/2} dt$ is the standard normal cumulative distribution function.

\begin{lemma}[Block Rotations and Thresholding] \label{lem:rotthres}
Let $F$ be a fixed partition of $[m]$ where $m$ is divisible by $(3^\ell - 1)k$. Let $U \subseteq [m]$ be a fixed $k$-subset such that $|U \cap F_i| = 1$ for each $i \in [k]$. Let $\mathcal{A}_{\textnormal{3-5}}$ denote Steps 5 and 6 of $k\textsc{-pds-to-semi-cr}$ with input $M_{\textnormal{G}}$ and output $G'$. Then for all $\tau \in \mathbb{R}$,
\begin{align*}
&\mathcal{A}_{\textnormal{3-5}} \left( \mu \cdot \mathbf{1}_U \mathbf{1}_{U}^\top + \mN(0, 1)^{\otimes m \times m} \right) \sim \pr{tg}_{H_1}\left(n, \frac{1}{2}(3^{\ell - 1} - 1)k, 3^{\ell - 1}k, m, \mu_1, \mu_2, \mu_3\right) \\
&\mathcal{A}_{\textnormal{3-5}} \left( \mN(0, 1)^{\otimes m \times m} \right) \sim \pr{tg}_{H_0}(n, m, \mu_1)
\end{align*}
where $\mu_1, \mu_2, \mu_3 \in (0, 1)$ are equal to
$$\mu_1 = \Phi\left( \frac{1}{2} \mu \cdot 3^{-\ell} \right) - \frac{1}{2}, \quad \text{and} \quad \mu_2 = \mu_3 = \Phi\left( \frac{1}{2} \mu \cdot 3^{-\ell + 1} \right) - \frac{1}{2}$$
\end{lemma}

\begin{proof}
First consider the case in which $M_{\textnormal{G}} \sim \mL\left( \mu \cdot \mathbf{1}_U \mathbf{1}_{U}^\top + \mN(0, 1)^{\otimes m \times m} \right)$. It follows that $M_{\text{P}} = \mu \cdot \mathbf{1}_{U'} \mathbf{1}_{U'}^\top + G$ where $G \sim \mN(0, 1)^{\otimes m' \times m'}$ and $U'$ is the image of $U$ under the embedding and index reordering in Step 3. Let $[m'] = F_1' \cup F_2' \cup \cdots \cup F_{ks}'$ be the partition of the row and column indices of $M_{\text{P}}$ induced by the blocks. Note that $|F_i'| = 3^\ell$ for each $i \in [ks]$. Since $F'$ is a refinement of $F$, it holds that $|U' \cap F_i'| \le 1$ for each $i \in [ks]$. Furthermore, if $W$ is the set of $m' - m$ indices added in Step 3, then it holds that $W$ and $U'$ are disjoint.

Let $[m''] = F''_1 \cup F''_2 \cup \cdots \cup F''_{ks}$ be the partition in Step 4 of $[m'']$ into blocks of size $\frac{1}{2}(3^\ell - 1)$. Now define the matrix $\mathcal{H} \in \mathbb{R}^{m'' \times m'}$ to be such that:
\begin{itemize}
\item $\mathcal{H}$ restricted to the indices $F_i'' \times F_i'$ is a copy of $H_{3, \ell}$ with $\mathcal{H}_{F_i'', F_i'} = H_{3, \ell}$ for each $i \in [ks]$
\item $\mathcal{H}_{ij} = 0$ if $(i, j)$ is not in $F_a'' \times F_a'$ for some $a \in [ks]$
\end{itemize}
The rotation step setting $(M_{\text{R}})_{F''_i, F''_j} = H_{3, \ell} (M_{\text{P}})_{F_i', F_j'} H_{3, \ell}^\top$ for each $i, j \in [ks]$ can equivalently be expressed as $M_{\text{R}} = \mathcal{H} M_{\text{P}} \mathcal{H}^\top$. Therefore we have that
$$M_{\text{R}} = \mathcal{H} M_{\text{P}} \mathcal{H}^\top = \mu \cdot \mathcal{H} \mathbf{1}_{U'} \mathbf{1}_{U'}^\top \mathcal{H}^\top + \mathcal{H}G\mathcal{H}^\top \sim \mL\left( \mu \cdot vv^\top + \mN(0, 1)^{\otimes m'' \times m''} \right)$$
where $v = \mathcal{H} \mathbf{1}_{U'} \in \mathbb{R}^{m''}$. The last statement holds since the rows of $\mathcal{H}$ are orthogonal by an application of the isotropic property of independent Gaussians similar to Lemma \ref{lem:samplerotation}.

Now consider the vector $v$, which is the sum of the $k$ columns of $\mathcal{H}$ with indices in $U'$. Since $|U' \cap F_i'| \le 1$, the construction of $\mathcal{H}$ implies that these $k$ columns have disjoint support. By Lemma \ref{lem:orthogonalmatrices}, each column of $\mathcal{H}$ not corresponding to the point $P_i = 0$ in some block contains exactly $\frac{1}{2}( 3^{\ell - 1} - 1)$ entries equal to $1/\sqrt{2 \cdot 3^\ell}$, exactly $3^{\ell - 1}$ equal to $-2/\sqrt{2 \cdot 3^\ell}$ and the rest of its entries are zero. Step 3 ensures that all columns of $\mathcal{H}$ corresponding to $P_i = 0$ are in $W$ and thus not in $U'$. Thus it follows that $v$ contains exactly $\frac{k}{2}( 3^{\ell - 1} - 1)$ entries equal to $1/\sqrt{2 \cdot 3^\ell}$, exactly $3^{\ell - 1}k$ equal to $-2/\sqrt{2 \cdot 3^\ell}$ and the rest of its entries are zero. Define $S$ and $S'$ to be
$$S = \left\{ i \in [m''] : v_i = -2/\sqrt{2 \cdot 3^\ell} \right\} \quad \text{and} \quad S' = \left\{ i \in [m''] : v_i = 1/\sqrt{2 \cdot 3^\ell} \right\}$$
where $|S| = \frac{k}{2}( 3^{\ell - 1} - 1)$ and $|S'| = 3^{\ell - 1}k$. Therefore it follows that the entries of $M_{\text{R}}$ are independent and distributed as follows:
$$(M_{\text{R}})_{ij} \sim \left\{ \begin{array}{ll} \mN(2\mu \cdot 3^{-\ell}, 1) &\text{if } (i, j) \in S \times S \\ \mN(-\mu \cdot 3^{-\ell}, 1) &\text{if } (i, j) \in S \times S' \text{ or } (i, j) \in S' \times S \\ \mN\left(\frac{1}{2} \mu \cdot 3^{-\ell}, 1\right) &\text{if } (i, j) \in S' \times S' \\ \mN(0, 1) &\text{otherwise} \end{array} \right.$$
After thresholding, adding $n - m$ new vertices and permuting as in Step 5 therefore yields that $G' \sim \pr{tg}_{H_1}\left(n, \frac{1}{2}(3^{\ell - 1} - 1)k, 3^{\ell - 1}k, m, \mu_1, \mu_2, \mu_3\right)$ for the values of $\mu_1, \mu_2, \mu_3 \in (0, 1)$ defined in the lemma statement. This completes the proof of the first part of the lemma. Now consider the case where $M_{\textnormal{G}} \sim \mN(0, 1)^{\otimes m \times m}$. By an identical argument, we have that $M_{\text{R}} \sim \mN(0, 1)^{\otimes m'' \times m''}$. Then thresholding, adding vertices and permuting as in Step 5 yields $G' \sim \pr{tg}_{H_0}(n, m, \mu_1)$, completing the proof of the lemma.
\end{proof}

Using this lemma, we now prove our second main result showing the $\pr{pds}$ Recovery Conjecture under a semirandom adversary for constant ambient edge density. We begin with the case of $q = 1/2$ and will deduce the general $q = \Theta(1)$ case subsequently in a corollary.

\begin{theorem}[$k$\pr{-pc} Lower Bounds for Semirandom Community Recovery] \label{thm:semicrhardness}
Fix any constant $\beta \in [1/2, 1)$. Suppose that $\mathcal{B}$ is a randomized polynomial time test for $\pr{semi-cr}(n, k, 1/2 + \nu, 1/2)$ for all $(n, k, \nu)$ with $k = \Theta(n^\beta)$ and $\nu \ge \overline{\nu}$ where $\overline{\nu}^2 = o(n/k^2 \log n)$. Then $\mathcal{B}$ has asymptotic Type I$+$II error at least 1 assuming either the $k$\pr{-pc} conjecture or the $k$\pr{-pds} conjecture for some fixed edge densities $0 < q < p \le 1$.
\end{theorem}

\begin{proof}
Assume the $k$\pr{-pds} conjecture for some fixed edge densities $0 < q < p \le 1$ and let $Q = 1 - \sqrt{(1 - p)(1 - q)} + \mathbf{1}_{\{ p = 1\}} \left( \sqrt{q} - 1 \right) \in (0, 1)$. Let $w = w(k') = \omega(1)$ be a sufficiently slow-growing function and define the parameters $(k', N)$ to be such that $N = wk'^2$. Now define the following parameters:
\begin{itemize}
\item Blow up factor $\ell = \lceil \log_3 (N^\beta/k') \rceil$ and target subgraph size $k = \frac{1}{2} \left( 3^{\ell - 1} - 1\right) k'$
\item Target number of vertices $n = m$, the smallest multiple of $(3^\ell - 1)k$ larger than $\left( \frac{p}{Q} + 1 \right) N$
\item Target graph distribution parameters $\mu_1 = \Phi\left( \frac{1}{2} \mu \cdot 3^{-\ell} \right) - \frac{1}{2}$ and $\mu_2 = \mu_3 = \Phi\left( \frac{1}{2} \mu \cdot 3^{-\ell + 1} \right) - \frac{1}{2}$
\end{itemize}
Note that these parameters satisfy the given conditions since $k = \Theta(N^\beta)$ and $N = \Theta(n)$. As defined in Step 2 of $\mathcal{A}$, it holds that $\mu = \Theta(1/\sqrt{\log n})$. Let $\nu = \mu_3$ and observe that
$$\nu = \Phi\left( \frac{1}{2} \mu \cdot 3^{-\ell + 1} \right) - \frac{1}{2} \asymp \mu \cdot 3^{-\ell} \asymp \frac{1}{\sqrt{\log n}} \cdot \frac{1}{k} \cdot \sqrt{\frac{N}{w}} \asymp \sqrt{\frac{n}{wk^2 \log n}}$$
since $\Phi(x) - 1/2 = \Theta(x)$ for $x \in (0, 1)$. Therefore it follows that $\nu \ge \overline{\nu}$ for a sufficiently slow-growing choice of $w$.

We now will show that $\mathcal{A}$ maps $\mG(N, q)$ approximately to $\pr{tg}_{H_0}(n, m, \mu_1)$ and maps $\mG_E(N, k', p, q)$ approximately to $\pr{tg}_{H_1}\left(n, \frac{1}{2}(3^{\ell - 1} - 1)k, 3^{\ell - 1}k, m, \mu_1, \mu_2, \mu_3\right)$ in total variation. To prove this, we apply Lemma \ref{lem:tvacc} to the steps $\mathcal{A}_i$ of $\mathcal{A}$ in each of these two cases. Let $E$ be a partition of $[N]$ into $k'$ parts of size $N/k'$. Define the steps of $\mathcal{A}$ to map inputs to outputs as follows
$$(G, E) \xrightarrow{\mathcal{A}_1} (M_{\text{PD}}, F) \xrightarrow{\mathcal{A}_2} (M_{\text{G}}, F) \xrightarrow{\mathcal{A}_{\text{3-5}}} G'$$
In the first case, consider Lemma \ref{lem:tvacc} applied to the steps $\mathcal{A}_i$ above and the following sequence of distributions
\allowdisplaybreaks
\begin{align*}
\mathcal{P}_0 &= \mG_E(N, k', p, q) \\
\mathcal{P}_1 &= \mathcal{M}(m, \mU_m(F'), p, Q) \\
\mathcal{P}_2 &=\mathcal{M}\left(m, \mU_{m}(F), \mN\left(\mu, 1\right), \mN(0, 1) \right) \\
\mathcal{P}_{\text{3-5}} &= \pr{tg}_{H_1}\left(n, \frac{1}{2}(3^{\ell - 1} - 1)k', 3^{\ell - 1}k', m, \mu_1, \mu_2, \mu_3\right)
\end{align*}
As in the statement of Lemma \ref{lem:tvacc}, let $\epsilon_i$ be any real numbers satisfying $\TV\left( \mathcal{A}_i(\mP_{i-1}), \mP_i \right) \le \epsilon_i$ for each step $i$. Lemma \ref{lem:submatrix} implies that we can take $\epsilon_1 = 4k \cdot \exp\left( - Q^2N^2/48pk'm \right) + \sqrt{C_Q k'^2/2m}$ where $C_Q = \max\left\{ Q/(1 - Q), (1 - Q)/Q \right\}$. Applying Lemma \ref{lem:gaussianize} and averaging over $S = T \sim \mU_{m}(F)$ yields that we can take $\epsilon_2 = O(N^{-1})$. Lemma \ref{lem:rotthres} implies that Steps 3, 4 and 5 are exact and we can take $\epsilon_{\text{3-5}} = 0$. Since $\epsilon_1, \epsilon_2 = o(1)$, Lemma \ref{lem:tvacc}, implies that $\mathcal{A}$ takes $\mathcal{P}_0$ to $\mathcal{P}_{\text{3-5}}$ with $o(1)$ total variation, which proves the first part of the claim. Now consider the second case Lemma \ref{lem:tvacc} applied to the following sequence of distributions
$$\mathcal{P}_0 = \mG(N, q), \quad \mathcal{P}_1 = \text{Bern}(Q)^{\otimes m \times m}, \quad \mathcal{P}_2 = \mN(0, 1)^{\otimes m \times m}, \quad \mathcal{P}_{\text{3-5}} = \pr{tg}_{H_0}(n, m, \mu_1)$$
As above, Lemmas \ref{lem:submatrix}, \ref{lem:gaussianize} and \ref{lem:samplerotation} imply that we can take $\epsilon_1 = 4k \cdot \exp\left( - Q^2N^2/48pk'm \right)$, $\epsilon_2 = O(N^{-1})$ and $\epsilon_{\text{3-5}} = 0$. Applying Lemma \ref{lem:tvacc} again implies that $\mathcal{A}$ takes $\mathcal{P}_0$ to $\mathcal{P}_{\text{3-5}}$ with $o(1)$ total variation, which proves the second part of the claim.

We now will show that these two target distributions can be simulated by the $H_0$ and $H_1$ semirandom adversaries in $\pr{semi-cr}(n, k, 1/2 + \nu, 1/2)$. Consider the adversary that observes $G \sim \mG(n, k, 1/2 + \nu, 1/2)$ and modifies $G$ as follows:
\begin{enumerate}
\item samples $S'$ of size $3^{\ell - 1} k'$ uniformly at random from all $3^{\ell - 1} k'$-subsets of $[n] \backslash S$ where $S$ is the vertex set of the planted dense subgraph; and
\item if the edge $\{i, j \}$ is in $E(G)$, remove it from $G$ independently with probability $P_{ij}$ where
$$P_{ij} = \left\{ \begin{array}{ll} 0 &\text{if } (i, j) \in S^2 \\ 2\mu_1 &\text{if } (i, j) \not \in (S \cup S')^2 \\ 2\mu_2 &\text{if } (i, j) \in S \times S' \text{ or } (i, j) \in S' \times S \end{array} \right.$$
\end{enumerate}
This exactly simulates $\pr{tg}_{H_1}\left(n, \frac{1}{2}(3^{\ell - 1} - 1)k, 3^{\ell - 1}k, m, \mu_1, \mu_2, \mu_3\right)$ and shows that it is in the set of distributions $\pr{adv}(\mG(n, k, 1/2 + \nu, 1/2))$. Now consider the adversary that observes a graph $G \sim \mG(n, 1/2)$ and removes every present edge independently with probability $2\mu_1$. This similarly shows that $\pr{tg}_{H_0}(n, m, \mu_1) \in \pr{adv}(\mG(n, 1/2))$. By Lemma \ref{lem:3a} applied to the reduction $\mathcal{A}$, if $\mathcal{B}$ has asymptotic Type I$+$II error less than 1, t follows that there is a randomized polynomial time test for $k\pr{-pds}$ on the sequence of inputs $(N, k', p, q)$ with asymptotic Type I$+$II error less than 1. This contradicts the $k$\pr{-pds} conjecture and proves the theorem.
\end{proof}

An identical analysis and reduction modified to replace the threshold $\frac{1}{2} \mu \cdot 3^{-\ell}$ with $\frac{1}{2} \mu \cdot 3^{-\ell} + \Phi^{-1}(q)$ shows the same statistical-computational gap holds at ambient edge density $q$ instead of $1/2$, as long as $\min\{q, 1 - q \} = \Omega(1)$. The resulting generalization is formally stated below.

\begin{corollary}[Arbitrary Bounded $q$]
Fix any constant $\beta \in [1/2, 1)$. Suppose that $\mathcal{B}$ is a randomized polynomial time test for $\pr{semi-cr}(n, k, p, q)$ for all $(n, k, p, q)$ with $k = \Theta(n^\beta)$ and
$$\frac{(p - q)^2}{q(1 - q)} \ge \overline{\nu} \quad \textnormal{and} \quad \min\{q, 1 - q \} = \Omega(1) \quad \textnormal{where} \quad \overline{\nu} = o\left(\frac{n}{k^2 \log n}\right)$$
Then $\mathcal{B}$ has asymptotic Type I$+$II error at least 1 assuming either the $k$\pr{-pc} conjecture or the $k$\pr{-pds} conjecture for some fixed edge densities $0 < q' < p' \le 1$.
\end{corollary}

%
%
%
%
%

\section{Universality of Lower Bounds for Learning Sparse Mixtures}
\label{sec:universality}

In this section, we combine our reduction to $\pr{isgm}$ from Section \ref{sec:redisgm} with a new gadget performing an algorithmic change of measure in order to obtain a universality principle for computational lower bounds at the sample complexity of $n = \tilde{\Omega}(k^2)$ for learning sparse mixtures. This gadget, symmetric 3-ary rejection kernels, is introduced and analyzed in Section \ref{subsec:srk}. We remark that the $k$-partite promise in $k\pr{-pc}$ and $k\pr{-pds}$ is crucially used in our reduction to obtain this universality. In particular, the $k$-partite promise ensures that the entries of the intermediate $\pr{isgm}$ instance are from one of three distinct distributions, when conditioned on the part of the mixture the sample is from. This is necessary for our application of symmetric 3-ary rejection kernels.

Our lower bounds hold given tail bounds on the likelihood ratios between the planted and noise distributions. In particular, our universality result shows tight computational lower bounds for sparse PCA in the spiked covariance model and a wide range of natural distributional formulations of learning sparse mixtures. The results in this section can also be interpreted as a universality principle for computational lower bounds in sparse PCA. We prove total variation guarantees for our reduction to $\pr{glsm}$ in Section \ref{subsec:universalitybounds} and discuss the universality conditions needed for our lower bounds in Section \ref{subsec:universalitydiscussion}.

\subsection{Symmetric 3-ary Rejection Kernels and Truncating Gaussians}
\label{subsec:srk}

\begin{figure}[t!]
\begin{algbox}
\textbf{Algorithm} \textsc{3-srk}$(B, \mP_+, \mP_-, \mQ)$

\vspace{2mm}

\textit{Parameters}: Input $B \in \{-1, 0, 1\}$, number of iterations $N$, parameters $a \in (0, 1)$ and sufficiently small nonzero $\mu_1, \mu_2 \in \mathbb{R}$, distributions $\mP_+, \mP_-$ and $\mQ$ over a measurable space $(X, \mathcal{B})$ such that $(\mP_+, \mQ)$ and $(\mP_-, \mQ)$ are computable pairs
\begin{enumerate}
\item Initialize $z$ arbitrarily in the support of $\mQ$.
\item Until $z$ is set or $N$ iterations have elapsed:
\begin{enumerate}
\item[(1)] Sample $z' \sim \mQ$ independently and compute the two quantities
$$\mL_1(z') = \frac{d\mP_+}{d\mQ} (z') - \frac{d\mP_-}{d\mQ} (z') \quad \text{and} \quad \mL_2(z') = \frac{d\mP_+}{d\mQ} (z') + \frac{d\mP_-}{d\mQ} (z') - 2$$
\item[(2)] Proceed to the next iteration if it does not hold that
$$2|\mu_1| \ge \left| \mL_1(z') \right| \quad \text{and} \quad \frac{2|\mu_2|}{\max\{a, 1 - a\}} \ge |\mL_2(z')|$$
\item[(3)] Set $z \gets z'$ with probability $P_A(x, B)$ where
$$P_A(x, B) = \frac{1}{2} \cdot \left\{ \begin{array}{ll} 1+ \frac{a}{4\mu_2} \cdot \mL_2(z') + \frac{1}{4\mu_1} \cdot \mL_1(z') &\text{if } B = 1 \\ 1 - \frac{1 - a}{4\mu_2} \cdot \mL_2(z') &\text{if } B = 0 \\ 1+ \frac{1}{4\mu_2} \cdot \mL_2(z') - \frac{a}{4\mu_1} \cdot \mL_1(z') &\text{if } B = -1 \end{array} \right.$$
\end{enumerate}
\item Output $z$.
\end{enumerate}
\vspace{1mm}
\end{algbox}
\caption{3-ary symmetric rejection kernel algorithm.}
\label{fig:srej-kernel}
\end{figure}

In this section, we introduce symmetric 3-ary rejection kernels, which will be the key gadget in our reduction showing universality of lower bounds for learning sparse mixtures. Rejection kernels are a form of an algorithmic change of measure. Rejection kernels mapping a pair of Bernoulli distributions to a target pair of scalar distributions were introduced in \cite{brennan2018reducibility}. These were extended to arbitrary high-dimensional target distributions and applied to obtain universality results for submatrix detection in \cite{brennan2019universality}. A surprising and key feature of both of these rejection kernels is that they are not lossy in mapping one computational barrier to another. For instance, in \cite{brennan2019universality}, multivariate rejection kernels were applied to increase the relative size $k$ of the planted submatrix, faithfully mapping instances tight to the computational barrier at lower $k$ to tight instances at higher $k$. This feature is also true of the scalar rejection kernels applied in \cite{brennan2018reducibility}.

To faithfully map the $k\pr{-pc}$ computational barrier onto the computational barrier of sparse mixtures, it is important to produce multiple planted distributions. Since previous rejection kernels all begin with binary inputs, they do not have enough degrees of freedom to map to three output distributions. The symmetric 3-ary rejection kernels $3\pr{-srk}$ introduced in this section overcome this issue by mapping from distributions supported on $\{-1, 0, 1\}$ to three output distributions $\mP_+, \mP_-$ and $\mQ$. In order to produce clean total variation guarantees, these rejection kernels also exploit symmetry in their three input distributions on $\{-1, 0, 1\}$.

Let $\text{Tern}(a, \mu_1, \mu_2)$ where $a \in (0, 1)$ and $\mu_1, \mu_2 \in \mathbb{R}$ denote the probability distribution on $\{-1, 0, 1\}$ such that if $B \sim \text{Tern}(a, \mu_1, \mu_2)$ then
$$\bP[X = -1] = \frac{1 - a}{2} - \mu_1 + \mu_2, \quad \bP[X = 0] = a - 2\mu_2, \quad \bP[X = 1] = \frac{1 - a}{2} + \mu_1 + \mu_2$$
if all three of these probabilities are nonnegative. The map $3\pr{-srk}(B)$, shown in Figure \ref{fig:srej-kernel}, sends an input $B \in \{-1, 0, 1\}$ to a set $X$ simultaneously satisfying three Markov transition properties:
\begin{enumerate}
\item if $B \sim \text{Tern}(a, \mu_1, \mu_2)$, then $3\textsc{-srk}(B)$ is close to $\mP_+$ in total variation;
\item if $B \sim \text{Tern}(a, -\mu_1, \mu_2)$, then $3\textsc{-srk}(B)$ is close to $\mQ$ in total variation; and
\item if $B \sim \text{Tern}(a, 0, 0)$, then $3\textsc{-srk}(B)$ is close to $\mP_-$ in total variation.
\end{enumerate}
In order to state our main results for $3\pr{-srk}(B)$, we will need the notion of computable pairs from \cite{brennan2019universality}. The definition below is that given in \cite{brennan2019universality}, without the assumption of finiteness of KL divergences. This assumption was convenient for the Chernoff exponent analysis needed for multivariate rejection kernels in \cite{brennan2019universality}. Since our rejection kernels are univariate, we will be able to state our universality conditions directly in terms of tail bounds rather than Chernoff exponents.

\begin{definition}[Relaxed Computable Pair \cite{brennan2019universality}] \label{def:computable}
Define a pair of sequences of distributions $(\mP, \mQ)$ over a measurable space $(X, \mathcal{B})$ where $\mP = (\mP_n)$ and $\mQ = (\mQ_n)$ to be computable if:
\begin{enumerate}
\item there is an oracle producing a sample from $\mQ_n$ in $\textnormal{poly}(n)$ time;
\item $\mP_n$ and $\mQ_n$ are mutually absolutely continuous and the likelihood ratio satisfies
$$\bE_{x \sim \mQ} \left[\frac{d\mP}{d\mQ}(x) \right] = \bE_{x \sim \mP}\left[\left( \frac{d\mP}{d\mQ}(x) \right)^{-1} \right] = 1$$
where $\frac{d\mP_n}{d\mQ_n}$ is the Radon-Nikodym derivative; and
\item there is an oracle computing $\frac{d\mP_n}{d\mQ_n} (x)$ in $\textnormal{poly}(n)$ time for each $x \in X$.
\end{enumerate}
\end{definition}

We remark that the second condition above always holds for discrete distributions and generally for most well-behaved distributions $\mP$ and $\mQ$. We now prove our main total variation guarantees for $3\pr{-srk}$. The proof of the next lemma follows a similar structure to the analysis of rejection sampling as in Lemma 5.1 of \cite{brennan2018reducibility} and Lemma 5.1 of \cite{brennan2019universality}. However, the bounds that we obtain are different than those in \cite{brennan2018reducibility, brennan2019universality} because of the symmetry of the three input $\text{Tern}$ distributions.

\begin{lemma}[Symmetric 3-ary Rejection Kernels] \label{lem:srk}
Let $a \in (0, 1)$ and $\mu_1, \mu_2 \in \mathbb{R}$ be nonzero and such that $\textnormal{Tern}(a, \mu_1, \mu_2)$ is well-defined. Let $\mP_+, \mP_-$ and $\mQ$ be distributions over a measurable space $(X, \mathcal{B})$ such that $(\mP_+, \mQ)$ and $(\mP_-, \mQ)$ are computable pairs with respect to a parameter $n$. Let $S \subseteq X$ be the set
$$S = \left\{x \in X : 2|\mu_1| \ge \left| \frac{d\mP_+}{d\mQ} (x) - \frac{d\mP_-}{d\mQ} (x) \right| \quad \textnormal{and} \quad \frac{2|\mu_2|}{\max\{a, 1 - a\}} \ge \left|\frac{d\mP_+}{d\mQ} (x) + \frac{d\mP_-}{d\mQ} (x) - 2 \right| \right\}$$
Given a positive integer $N$, then the algorithm $3\textsc{-srk} : \{-1, 0, 1\} \to X$ can be computed in $\textnormal{poly}(n, N)$ time and satisfies that
$$\left. \begin{array}{r} \TV\left( 3\textsc{-srk}(\textnormal{Tern}(a, \mu_1, \mu_2)), \mP_+ \right) \\ \TV\left( 3\textsc{-srk}(\textnormal{Tern}(a, -\mu_1, \mu_2)), \mP_- \right) \\\TV\left( 3\textsc{-srk}(\textnormal{Tern}(a, 0, 0)), \mQ \right) \end{array} \right\} \le 2\delta \left(1 + |\mu_1|^{-1} + |\mu_2|^{-1} \right) + \left( \frac{1}{2} + \delta \left( 1 + |\mu_1|^{-1} + |\mu_2|^{-1} \right) \right)^N$$
where $\delta > 0$ is such that $\bP_{X \sim \mP_+}[X \not \in S]$, $\bP_{X \sim \mP_-}[X \not \in S]$ and $\bP_{X \sim \mQ}[X \not \in S]$ are upper bounded by $\delta$.
\end{lemma}

\begin{proof}
Define $\mL_1, \mL_2 : X \to \mathbb{R}$ to be
$$\mL_1(x) = \frac{d\mP_+}{d\mQ} (x) - \frac{d\mP_-}{d\mQ} (x) \quad \text{and} \quad \mL_2(x) = \frac{d\mP_+}{d\mQ} (x) + \frac{d\mP_-}{d\mQ} (x) - 2$$
Note that if $x \in S$, then the triangle inequality implies that
\begin{align*}
P_A(x, 1) &\le \frac{1}{2} \left( 1 + \frac{a}{4|\mu_2|} \cdot |\mL_2(x)| + \frac{1}{4|\mu_1|} \cdot |\mL_1(x)| \right) \le 1 \\
P_A(x, 1) &\ge \frac{1}{2} \left( 1 - \frac{a}{4|\mu_2|} \cdot |\mL_2(x)| - \frac{1}{4|\mu_1|} \cdot |\mL_1(x)| \right) \ge 0
\end{align*}
Similar computations show that $0 \le P_A(x, 0) \le 1$ and $0 \le P_A(x, -1) \le 1$, implying that each of these probabilities is well-defined. Now let $R_1 = \bP_{X \sim \mP_+}[X \in S]$, $R_0 = \bP_{X \sim \mQ}[X \in S]$ and $R_{-1} = \bP_{X \sim \mP_-}[X \in S]$ where $R_1, R_0, R_{-1} \ge 1 - \delta$ by assumption.

We now define several useful events. For the sake of analysis, consider continuing to iterate Step 2 even after $z$ is set for the first time for a total of $N$ iterations. Let $A_i^1$, $A_i^0$ and $A_i^{-1}$ be the events that $z$ is set in the $i$th iteration of Step 2 when $B = 1$, $B = 0$ and $B = -1$, respectively. Let $B_i^1 = (A_1^1)^C \cap (A_2^1)^C \cap \cdots \cap (A^1_{i - 1})^C \cap A_i^1$ be the event that $z$ is set for the first time in the $i$th iteration of Step 2. Let $C^1 = A_1^1 \cup A_2^1 \cup \cdots \cup A_N^1$ be the event that $z$ is set in some iteration of Step 2. Define $B_i^0$, $C^0$, $B_i^{-1}$ and $C^{-1}$ analogously. Let $z_0$ be the initialization of $z$ in Step 1.

Now let $Z_1 \sim \mD_1 = \mL(3\textsc{-srk}(1))$, $Z_0 \sim \mD_0 = \mL(3\textsc{-srk}(0))$ and $Z_{-1} \sim \mD_{-1} = \mL(3\textsc{-srk}(-1))$. Note that $\mL(Z_t|B_i^t) = \mL(Z_t|A_i^t)$ for each $t \in \{-1, 0, 1\}$ since $A_i^t$ is independent of $A_1^t, A_2^t, \dots, A_{i-1}^t$ and the sample $z'$ chosen in the $i$th iteration of Step 2. The independence between Steps 2.1 and 2.3 implies that
\allowdisplaybreaks
\begin{align*}
\bP\left[A_i^1\right] &= \bE_{x \sim \mQ}\left[ \frac{1}{2} \left( 1+ \frac{a}{4\mu_2} \cdot \mL_2(x) + \frac{1}{4\mu_1} \cdot \mL_1(x) \right) \cdot \mathbf{1}_{S}(x) \right] \\
&= \frac{1}{2} R_0 + \frac{a}{8\mu_2} \left( R_1 + R_{-1} - 2R_0 \right) + \frac{1}{8\mu_1} \left( R_1 - R_{-1} \right) \ge \frac{1}{2} - \frac{\delta}{2} \left( 1 + \frac{a}{2}|\mu_2|^{-1} + \frac{1}{4}|\mu_1|^{-1} \right) \\
\bP\left[A_i^0 \right] &= \bE_{x \sim \mQ}\left[ \frac{1}{2} \left( 1 - \frac{1 - a}{4\mu_2} \cdot \mL_2(x) \right) \cdot \mathbf{1}_{S}(x) \right] \\
&= \frac{1}{2} R_0 - \frac{1 - a}{8\mu_2} \left( R_1 + R_{-1} - 2R_0 \right) \ge \frac{1}{2} - \frac{\delta}{2} \left( 1 + \frac{1 - a}{4} \cdot |\mu_2|^{-1} \right) \\
\bP\left[A_i^{-1}\right] &= \bE_{x \sim \mQ}\left[ \frac{1}{2} \left( 1+ \frac{a}{4\mu_2} \cdot \mL_2(x) - \frac{1}{4\mu_1} \cdot \mL_1(x) \right) \cdot \mathbf{1}_{S}(x) \right] \\
&= \frac{1}{2} R_0 + \frac{a}{8\mu_2} \left( R_1 + R_{-1} - 2R_0 \right) - \frac{1}{4\mu_1} \left( R_1 - R_{-1} \right) \ge \frac{1}{2} - \frac{\delta}{2} \left( 1 + \frac{a}{2}|\mu_2|^{-1} + \frac{1}{4}|\mu_1|^{-1} \right)
\end{align*}
The independence of the $A_i^t$ for each $t \in \{-1, 0, 1\}$ implies that
$$1 - \bP\left[ C^t \right] = \prod_{i = 1}^N \left( 1 - \bP\left[A_i^t\right] \right) \le \left( \frac{1}{2} + \frac{\delta}{2} \left( 1 + \frac{1}{2}|\mu_2|^{-1} + |\mu_1|^{-1} \right) \right)^N$$
Note that $\mL(Z_t|A_i^t)$ are each absolutely continuous with respect to $\mQ$ or each $t \in \{-1, 0, 1\}$, with Radon-Nikodym derivatives given by
\allowdisplaybreaks
\begin{align*}
\frac{d\mL(Z_1|B_i^1)}{d\mQ} (x) = \frac{d\mL(Z_1|A_i^1)}{d\mQ} (x) &= \frac{1}{2\cdot \bP\left[A_i^1\right]} \left( 1+ \frac{a}{4\mu_2} \cdot \mL_2(x) + \frac{1}{4\mu_1} \cdot \mL_1(x) \right) \cdot \mathbf{1}_S(x) \\
\frac{d\mL(Z_0|B_i^0)}{d\mQ} (x) = \frac{d\mL(Z_0|A_i^0)}{d\mQ} (x) &= \frac{1}{2\cdot \bP\left[A_i^1\right]} \left( 1 - \frac{1 - a}{4\mu_2} \cdot \mL_2(x) \right) \cdot \mathbf{1}_S(x) \\
\frac{d\mL(Z_{-1}|B_i^{-1})}{d\mQ} (x) = \frac{d\mL(Z_{-1}|A_i^{-1})}{d\mQ} (x) &= \frac{1}{2\cdot \bP\left[A_i^1\right]} \left( 1+ \frac{a}{4\mu_2} \cdot \mL_2(x) - \frac{1}{4\mu_1} \cdot \mL_1(x) \right) \cdot \mathbf{1}_S(x)
\end{align*}
Fix one of $t \in \{-1, 0, 1\}$ and note that since the conditional laws $\mL(Z_t|B_i^t)$ are all identical, we have that
$$\frac{d\mD_t}{d\mQ} (x) = \bP\left[C^t \right] \cdot \frac{d\mL(Z_t|B_1^t)}{d\mQ} (x) + \left( 1 - \bP\left[C^t \right] \right) \cdot \mathbf{1}_{z_0}(x)$$
Therefore it follows that
\begin{align*}
\TV\left( \mD_t, \mL(Z_t|B_1^t) \right) &= \frac{1}{2} \cdot \bE_{x \sim \mQ} \left[\left| \frac{d\mD_t}{d\mQ} (x) - \frac{d\mL(Z_t|B_1^t)}{d\mQ} (x) \right| \right] \\
&\le \frac{1}{2} \left( 1 - \bP\left[ C^t \right] \right) \cdot \bE_{x \sim \mQ} \left[ \mathbf{1}_{z_0}(x) + \frac{d\mL(Z_t|B_1^t)}{d\mQ} (x) \right] = 1 - \bP\left[ C^t \right]
\end{align*}
by the triangle inequality. Since $1+ \frac{a}{4\mu_2} \cdot \mL_2(x) + \frac{1}{4\mu_1} \cdot \mL_1(x) \ge 0$ for $x \in S$, we have that
\allowdisplaybreaks
\begin{align*}
&\bE_{x \sim \mQ} \left[\left| \frac{d\mL(Z_1|B_1^1)}{d\mQ} (x) - \left( 1+ \frac{a}{4\mu_2} \cdot \mL_2(x) + \frac{1}{4\mu_1} \cdot \mL_1(x) \right) \right| \right] \\ 
&\quad \quad = \left|\frac{1}{2\cdot \bP\left[A_i^1\right]} - 1 \right| \cdot \bE_{x \sim \mQ^*_n} \left[\left( 1+ \frac{a}{4\mu_2} \cdot \mL_2(x) + \frac{1}{4\mu_1} \cdot \mL_1(x) \right) \cdot \mathbf{1}_S(x) \right] \\
&\quad \quad \quad \quad + \bE_{x \sim \mQ} \left[ \left| 1+ \frac{a}{4\mu_2} \cdot \mL_2(x) + \frac{1}{4|\mu_1|} \cdot \mL_1(x) \right| \cdot \mathbf{1}_{S^C}(x) \right] \\
&\quad \quad \le \left| \frac{1}{2} - \bP[A_i^1] \right| + \bE_{x \sim \mQ} \left[ \left( 1+ \frac{a}{4|\mu_2|} \cdot \left( \frac{d\mP_+}{d\mQ} (x) + \frac{d\mP_-}{d\mQ} (x) +2 \right) \right) \cdot \mathbf{1}_{S^C}(x) \right] \\
&\quad \quad \quad \quad + \bE_{x \sim \mQ} \left[ \frac{1}{4|\mu_1|} \cdot \left( \frac{d\mP_+}{d\mQ} (x) + \frac{d\mP_-}{d\mQ} (x) \right) \cdot \mathbf{1}_{S^C}(x) \right] \\
&\quad \quad \le \frac{\delta}{2} \left( 1 + \frac{a}{2}|\mu_2|^{-1} + \frac{1}{4}|\mu_1|^{-1} \right) + \delta \left( 1 + a|\mu_2|^{-1} + \frac{1}{2}|\mu_1|^{-1} \right) = \delta \left( \frac{3}{2} + \frac{5}{4} |\mu_2|^{-1} + \frac{5}{8} |\mu_1|^{-1} \right)
\end{align*}
By analogous computations, we have that
\allowdisplaybreaks
\begin{align*}
\bE_{x \sim \mQ} \left[\left| \frac{d\mL(Z_0|B_1^0)}{d\mQ} (x) - \left( 1 - \frac{1 - a}{4\mu_2} \cdot \mL_2(x) \right) \right| \right] &\le 2\delta \left(1 + |\mu_1|^{-1} + |\mu_2|^{-1} \right) \\ 
\bE_{x \sim \mQ} \left[\left| \frac{d\mL(Z_{-1}|B_1^{-1})}{d\mQ} (x) - \left( 1+ \frac{a}{4\mu_2} \cdot \mL_2(x) - \frac{1}{4\mu_1} \cdot \mL_1(x) \right) \right| \right] &\le 2\delta \left(1 + |\mu_1|^{-1} + |\mu_2|^{-1} \right) 
\end{align*}
Now observe that
\allowdisplaybreaks
\begin{align*}
\frac{d\mP_+}{d\mQ}(x) &= \left( \frac{1 - a}{2} + \mu_1 + \mu_2 \right) \cdot \left( 1+ \frac{a}{4\mu_2} \cdot \mL_2(x) + \frac{1}{4\mu_1} \cdot \mL_1(x) \right) + (a - 2\mu_2) \cdot \left( 1 - \frac{1 - a}{4\mu_2} \cdot \mL_2(x) \right) \\
&\quad \quad + \left( \frac{1 - a}{2} - \mu_1 + \mu_2 \right) \cdot \left( 1+ \frac{a}{4\mu_2} \cdot \mL_2(x) - \frac{1}{4\mu_1} \cdot \mL_1(x) \right) \\
1 &= \frac{1 - a}{2} \cdot \left( 1+ \frac{a}{4\mu_2} \cdot \mL_2(x) + \frac{1}{4\mu_1} \cdot \mL_1(x) \right) + a \cdot \left( 1 - \frac{1 - a}{4\mu_2} \cdot \mL_2(x) \right) \\
&\quad \quad +\frac{1 - a}{2} \cdot \left( 1+ \frac{a}{4\mu_2} \cdot \mL_2(x) - \frac{1}{4\mu_1} \cdot \mL_1(x) \right) \\
\frac{d\mP_-}{d\mQ}(x) &= \left( \frac{1 - a}{2} - \mu_1 + \mu_2 \right) \cdot \left( 1+ \frac{a}{4\mu_2} \cdot \mL_2(x) + \frac{1}{4\mu_1} \cdot \mL_1(x) \right) + (a - 2\mu_2) \cdot \left( 1 - \frac{1 - a}{4\mu_2} \cdot \mL_2(x) \right) \\
&\quad \quad + \left( \frac{1 - a}{2} + \mu_1 + \mu_2 \right) \cdot \left( 1+ \frac{a}{4\mu_2} \cdot \mL_2(x) - \frac{1}{4\mu_1} \cdot \mL_1(x) \right)
\end{align*}
Let $\mD^*$ be the mixture of $\mL(Z_1 | B_1^1), \mL(Z_0 | B_1^0)$ and $\mL(Z_{-1} | B_1^{-1})$ with weights $\frac{1 - a}{2} + \mu_1 + \mu_2, a - 2\mu_2$ and $\frac{1 - a}{2} - \mu_1 + \mu_2$, respectively. It then follows by the triangle inequality that
\allowdisplaybreaks
\begin{align*}
&\TV\left( 3\textsc{-srk}(\textnormal{Tern}(a, \mu_1, \mu_2)), \mP_+ \right) \\
&\quad \quad \le \TV\left( \mD^*, \mP_+ \right) + \TV\left( \mD^*, 3\textsc{-srk}(\textnormal{Tern}(a, \mu_1, \mu_2)) \right) \\
&\quad \quad \le \left( \frac{1 - a}{2} + \mu_1 + \mu_2 \right) \cdot \bE_{x \sim \mQ} \left[\left| \frac{d\mL(Z_1|B_1^1)}{d\mQ} (x) - \left( 1+ \frac{a}{4\mu_2} \cdot \mL_2(x) + \frac{1}{4\mu_1} \cdot \mL_1(x) \right) \right| \right] \\
&\quad \quad \quad \quad + \left( a - 2\mu_2 \right) \cdot \bE_{x \sim \mQ} \left[\left| \frac{d\mL(Z_0|B_1^0)}{d\mQ} (x) - \left( 1 - \frac{1 - a}{4\mu_2} \cdot \mL_2(x) \right) \right| \right] \\
&\quad \quad \quad \quad + \left( \frac{1 - a}{2} - \mu_1 + \mu_2 \right) \cdot \bE_{x \sim \mQ} \left[\left| \frac{d\mL(Z_{-1}|B_1^{-1})}{d\mQ} (x) - \left( 1+ \frac{a}{4\mu_2} \cdot \mL_2(x) - \frac{1}{4\mu_1} \cdot \mL_1(x) \right) \right| \right] \\
&\quad \quad \quad \quad + \left( \frac{1 - a}{2} + \mu_1 + \mu_2 \right) \cdot \TV\left( \mD_1, \mL(Z_1|B_1^1) \right)  + \left( a - 2\mu_2 \right) \cdot \TV\left( \mD_1, \mL(Z_0|B_1^0) \right) \\
&\quad \quad \quad \quad + \left( \frac{1 - a}{2} - \mu_1 + \mu_2 \right) \cdot \TV\left( \mD_{-1}, \mL(Z_{-1}|B_1^{-1}) \right) \\
&\quad \quad \le 2\delta \left(1 + |\mu_1|^{-1} + |\mu_2|^{-1} \right) + \left( \frac{1}{2} + \delta \left( 1 + |\mu_1|^{-1} + |\mu_2|^{-1} \right) \right)^N
\end{align*}
A symmetric argument shows analogous upper bounds on $\TV\left( 3\textsc{-srk}(\textnormal{Tern}(a, -\mu_1, \mu_2)), \mP_- \right)$ and $\TV\left( 3\textsc{-srk}(\textnormal{Tern}(a, 0, 0)), \mQ \right)$. This completes the proof of the lemma.
\end{proof}

In our reduction showing universality, we will truncate Gaussians to generate the input distributions $\text{Tern}$. Let $\pr{tr}_{\tau} : \mathbb{R} \to \{-1, 0, 1\}$ be the following truncation map
$$\pr{tr}_{\tau}(x) = \left\{ \begin{array}{ll} 1 &\text{if } x > |\tau| \\ 0 &\text{if } -|\tau| \le x \le |\tau| \\ -1 &\text{if } x < -|\tau| \end{array} \right.$$
We conclude this section with the following simple lemma on truncating symmetric triples of Gaussian distributions.

\begin{lemma}[Truncating Gaussians] \label{lem:truncgauss}
Let $\tau > 0$ be constant, $\mu > 0$ be tending to zero and let $a, \mu_1, \mu_2$ be such that
\begin{align*}
&\pr{tr}_\tau(\mN(\mu, 1)) \sim \textnormal{Tern}(a, \mu_1, \mu_2) \\
&\pr{tr}_\tau(\mN(-\mu, 1)) \sim \textnormal{Tern}(a, -\mu_1, \mu_2) \\
&\pr{tr}_\tau(\mN(0, 1)) \sim \textnormal{Tern}(a, 0, 0)
\end{align*}
Then it follows that $a > 0$ is constant, $0 < \mu_1 = \Theta(\mu)$ and $0 < \mu_2 = \Theta(\mu^2)$.
\end{lemma}

\begin{proof}
The parameters $a, \mu_1, \mu_2$ for which these distributional statements are true are given by
\allowdisplaybreaks
\begin{align*}
a &= \Phi(\tau) - \Phi(-\tau) \\
\mu_1 &= \frac{1}{2} \left( (1 - \Phi(\tau - \mu)) - \Phi(-\tau - \mu) \right) = \frac{1}{2} \left( \Phi(\tau + \mu) - \Phi(\tau - \mu) \right) \\
\mu_2 &= \frac{1}{2} \left( \Phi(\tau) - \Phi(-\tau) \right) - \frac{1}{2} \left( \Phi(\tau + \mu) - \Phi(-\tau + \mu) \right) = \frac{1}{2} \left( 2 \cdot \Phi(\tau) - \Phi(\tau + \mu) - \Phi(\tau - \mu) \right)
\end{align*}
Now note that
$$\mu_1 = \frac{1}{2} \left( \Phi(\tau + \mu) - \Phi(\tau - \mu) \right) = \frac{1}{2\sqrt{2\pi}} \int_{\tau - \mu}^{\tau + \mu} e^{-t^2/2} dt = \Theta(\mu)$$
and is positive since $e^{-t^2/2}$ is bounded on $[\tau - \mu, \tau + \mu]$ as $\tau$ is constant and $\mu \to 0$. Furthermore, note that
\begin{align*}
\mu_2 &= \frac{1}{2} \left( 2 \cdot \Phi(\tau) - \Phi(\tau + \mu) - \Phi(\tau - \mu) \right) = \frac{1}{2\sqrt{2\pi}} \int_{\tau - \mu}^{\tau} e^{-t^2/2} dt - \frac{1}{2\sqrt{2\pi}} \int_{\tau}^{\tau + \mu} e^{-t^2/2} dt \\
&= \frac{1}{2\sqrt{2\pi}} \int_{\tau}^{\tau + \mu} \left( e^{-(t - \mu)^2/2} - e^{-t^2/2}\right) dt = \frac{1}{2\sqrt{2\pi}} \int_{\tau}^{\tau + \mu}  e^{-t^2/2} \left(e^{t\mu - \mu^2/2} - 1 \right) dt 
\end{align*}
Now note that as $\mu \to 0$ and for $t \in [\tau, \tau + \mu]$, it follows that $0 < e^{t\mu - \mu^2/2} - 1= \Theta(\mu)$. This implies that $0 < \mu_2 = \Theta(\mu^2)$, as claimed.
\end{proof}

\subsection{Universality of the $n = \tilde{\Theta}(k^2)$ Computational Barrier in Sparse Mixtures}
\label{subsec:universalitybounds}

\begin{figure}[t!]
\begin{algbox}
\textbf{Algorithm} $k$\textsc{-pds-to-glsm}

\vspace{1mm}

\textit{Inputs}: $k$\pr{-pds} instance $G \in \mG_N$ with dense subgraph size $k$ that divides $N$, partition $E$ of $[N]$ and edge probabilities $0 < q < p \le 1$, constant threshold $\tau > 0$, slow-growing function $w(N) = \omega(1)$, target $\pr{glsm}$ parameters $(n, k, d)$ with $wn \le cN$ and $d \ge c^{-1} N$ for a sufficiently small constant $c > 0$, mixture distribution $\mD$ and target distributions $\{ \mP_{\nu} \}_{\nu \in \mathbb{R}}$ and $\mQ$

\begin{enumerate}
\item \textit{Map to Gaussian Sparse Mixtures}: Form the sample $Z_1, Z_2, \dots, Z_n \in \mathbb{R}^d$ by setting
$$(Z_1, Z_2, \dots, Z_n) \gets k\pr{-pds-to-isgm}(G, E)$$
where $k\pr{-pds-to-isgm}$ is applied with $r = 2$, slow-growing function $w(N) = \omega(1)$, $t = \lceil \log_2(c^{-1} N/k) \rceil$, target parameters $n, k, d$, $\epsilon = 1/2$ and $\mu = c\sqrt{\frac{k}{N \log n}}$.
\item \textit{Truncate and 3-ary Rejection Kernels}: Sample $\nu_1, \nu_2, \dots, \nu_n \sim_{\text{i.i.d.}} \mD$, truncate the $\nu_i$ to lie within $[-1, 1]$ and form the vectors $X_1, X_2, \dots, X_n \in \mathbb{R}^d$ by setting
$$X_{ij} \gets 3\pr{-srk}(\pr{tr}_{\tau}(Z_{ij}), \mP_{\nu_i}, \mP_{-\nu_i}, \mQ)$$
for each $i \in [n]$ and $j \in [d]$. Here $3\pr{-srk}$ is applied with $N = \lceil 4 \log (dn) \rceil$ iterations and with the parameters
\begin{align*}
a &= \Phi(\tau) - \Phi(-\tau), \quad \mu_1 = \frac{1}{2} \left( \Phi(\tau + \mu) - \Phi(\tau - \mu) \right), \\
\mu_2 &= \frac{1}{2} \left( 2 \cdot \Phi(\tau) - \Phi(\tau + \mu) - \Phi(\tau - \mu) \right)
\end{align*}
\item \textit{Output}: The vectors $(X_1, X_2, \dots, X_n)$.
\end{enumerate}
\vspace{0.5mm}

\end{algbox}
\caption{Reduction from $k$-partite planted dense subgraph to general learning sparse mixtures.}
\label{fig:universalityreduction}
\end{figure}

In this section, we combine symmetric 3-ary rejection kernels with the reduction $k\pr{-pds-to-isgm}$ to map from $k\pr{-pds}$ to generalized sparse mixtures. The details of this reduction $k$\textsc{-pds-to-glsm} are shown in Figure \ref{fig:universalityreduction}. Throughout this section, we will denote $\mathcal{A} = k\textsc{-pds-to-glsm}$. In order to apply symmetric 3-ary rejection kernels, we will need a set of conditions on the target distributions $\mD, \mQ$ and $\{ \mP_{\nu} \}_{\nu \in \mathbb{R}}$. These conditions will also be the conditions for our universality result. As will be discussed in Section \ref{subsec:universalitydiscussion}, these conditions turn out to faithfully map the computational barrier of $k\pr{-pds}$ to learning sparse mixtures in a number of natural cases, including learning sparse Gaussian mixtures and sparse PCA in the spiked covariance model. Our universality conditions are as follows.

\begin{definition}[Universality Conditions]
Given parameters $(n, k, d)$, define the collection of distributions $(\mD, \mQ, \{ \mP_{\nu} \}_{\nu \in \mathbb{R}})$ to be in $\pr{uc}(n, k, d)$ if:
\begin{itemize}
\item $\mD$ is a symmetric distribution about zero and $\bP_{\nu \sim \mD}[\nu \in [-1, 1]] = 1 - o(n^{-1})$; and
\item for all $\nu \in [-1, 1]$, it holds that
$$\frac{1}{\sqrt{k \log n}} \gg \left| \frac{d\mP_{\nu}}{d\mQ} (x) - \frac{d\mP_{-\nu}}{d\mQ} (x) \right| \quad \textnormal{and} \quad \frac{1}{k \log n} \gg \left|\frac{d\mP_{\nu}}{d\mQ} (x) + \frac{d\mP_{-\nu}}{d\mQ} (x) - 2 \right|$$
with probability at least $1 - o(n^{-3} d^{-1})$ over each of $\mP_{\nu}, \mP_{-\nu}$ and $\mQ$.
\end{itemize}
\end{definition}

Let $\mathcal{A}_2$ denote Step 2 of $\mathcal{A}$ with input $(Z_1, Z_2, \dots, Z_n)$ and output $(X_1, X_2, \dots, X_n)$. We now prove total variation guarantees for $\mathcal{A}_2$, which follow from an application of tensorization of $\TV$.

\begin{lemma}[$\pr{isgm}$ to $\pr{glsm}$] \label{lem:univlem}
Suppose that $\tau > 0$ is a fixed constant and $\mu = \Omega(1/\sqrt{wk\log n})$ for a sufficiently slow-growing function $w$. If $(\mD, \mQ, \{ \mP_{\nu} \}_{\nu \in \mathbb{R}}) \in \pr{uc}(n, k, d)$, then
$$\TV\left( \mathcal{A}_2\left( \pr{isgm}(n, k, d, \mu, 1/2) \right), \pr{glsm}\left(n, k, d, \{ \mP_{\nu} \}_{\nu \in \mathbb{R}}, \mQ, \mD \right) \right) = o(1)$$
under both $H_0$ and $H_1$.
\end{lemma}

\begin{proof}
Let $(Z_1, Z_2, \dots, Z_n)$ be an instance of $\pr{isgm}(n, k, d, \mu, 1/2)$ under $H_1$. In other words, $Z_i \sim_{\text{i.i.d.}} \pr{mix}_{1/2}\left( \mN( \mu \cdot \mathbf{1}_S, I_d), \mN( -\mu \cdot \mathbf{1}_S, I_d) \right)$ where $S$ is a $k$-subset of $[d]$ chosen uniformly at random. For the next part of this argument, we condition on: (1) the entire vector $(\nu_1, \nu_2, \dots, \nu_n)$; (2) the latent support $S \subseteq [d]$ with $|S| = k$ of the planted indices of the $Z_i$; and (3) the subset $P \subseteq [n]$ of sample indices corresponding to the positive part $\mN(\mu \cdot \mathbf{1}_S, I_d)$ of the mixture. Let $\mathcal{C}$ denote the event corresponding to this conditioning. After truncating according to $\pr{tr}_{\tau}$, by Lemma \ref{lem:truncgauss} the resulting entries are distributed as
$$\pr{tr}_{\tau}(Z_{ij}) \sim \left\{ \begin{array}{ll} \text{Tern}(a, \mu_1, \mu_2) &\text{if } (i, j) \in S \times P \\ \text{Tern}(a, -\mu_1, \mu_2) &\text{if } (i, j) \in S \times P^C \\ \text{Tern}(a, 0, 0) &\text{if } i \not \in S \end{array} \right.$$
Furthermore, given our conditioning, these entries are all independent. Since $\tau$ is constant, Lemma \ref{lem:truncgauss} also implies that $a \in (0, 1)$ is constant, $\mu_1 = \Theta(\mu)$ and $\mu_2 = \Theta(\mu^2)$. Let $S_\nu$ be
$$S_\nu = \left\{ x \in X : 2|\mu_1| \ge \left| \frac{d\mP_{\nu_i}}{d\mQ} (x) - \frac{d\mP_{-\nu_i}}{d\mQ} (x) \right| \quad \textnormal{and} \quad \frac{2|\mu_2|}{\max\{a, 1 - a\}} \ge \left|\frac{d\mP_{\nu_i}}{d\mQ} (x) + \frac{d\mP_{-\nu_i}}{d\mQ} (x) - 2 \right| \right\}$$
as in Lemma \ref{lem:srk}. The second implication of $(\mD, \mQ, \{ \mP_{\nu} \}_{\nu \in \mathbb{R}}) \in \pr{uc}(n, k, d)$ gives that $\{x \in S_{\nu_i}\}$ occurs with probability at least $1 - \delta$ over each of $\mP_{\nu_i}, \mP_{-\nu_i}$ and $\mQ$, where $\delta = o(n^{-3} d^{-1})$, for each $i \in [n]$. This holds for a sufficiently slow-growing choice of $w$. Therefore we can apply Lemma \ref{lem:srk} to each application of $3\pr{-srk}$ in Step 2 of $\mathcal{A}$. Note that $|\mu_1|^{-1} = O(\sqrt{n \log n})$ and $|\mu_2|^{-1} = O(n \log n)$ since $\mu = \Omega(1/\sqrt{wk\log n})$ and $k \ge 1$. Now consider the $d$-dimensional vectors $X_1', X_2', \dots, X_n'$ with independent entries distributed as
$$X'_{ij} \sim \left\{ \begin{array}{ll} \mP_{\nu_i} &\text{if } (i, j) \in S \times P \\ \mP_{-\nu_i} &\text{if } (i, j) \in S \times P^C \\ \mQ &\text{if } i \not \in S \end{array} \right.$$
The tensorization property of $\TV$ from Fact \ref{tvfacts} implies that
\begin{align*}
&\TV\left( \mL(X_1, X_2, \dots, X_n | \mathcal{C}), \mL(X_1', X_2', \dots, X_n'| \mathcal{C}) \right) \\
&\quad \quad \le \sum_{i = 1}^n \sum_{j = 1}^d \TV\left( \mL(X_{ij} | \mathcal{C}), \mL(X_{ij}' | \mathcal{C}) \right) \\
&\quad \quad \le \sum_{i = 1}^n \sum_{j = 1}^d \TV\left( 3\pr{-srk}(\pr{tr}_{\tau}(Z_{ij}), \mP_{\nu_i}, \mP_{-\nu_i}, \mQ), \mL(X_{ij}' | \mathcal{C}) \right) \\
&\quad \quad \le nd \left[ 2\delta \left(1 + |\mu_1|^{-1} + |\mu_2|^{-1} \right) + \left( \frac{1}{2} + \delta \left( 1 + |\mu_1|^{-1} + |\mu_2|^{-1} \right) \right)^N \right] = o(1)
\end{align*}
by the total variation upper bounds in Lemma \ref{lem:srk}. Now consider dropping the conditioning on $\mathcal{C}$. It follows by the definition of $\pr{glsm}$ that $(X_1', X_2', \dots, X_n')$, when no longer conditioned on $\mathcal{C}$, is distributed as $\pr{glsm}\left(n, k, d, \{ \mP_{\nu} \}_{\nu \in \mathbb{R}}, \mQ, \mD' \right)$ where $\mD'$ is the distribution sampled by first sampling $x \sim \mD$, truncating $x$ to lie in $[-1, 1]$ and then multiplying $x$ by $-1$ with probability $1/2$. It therefore follows by the conditioning property of $\TV$ in Fact \ref{tvfacts} that
$$\TV\left( \mathcal{A}_2\left( \pr{isgm}(n, k, d, \mu, 1/2) \right), \pr{glsm}\left(n, k, d, \{ \mP_{\nu} \}_{\nu \in \mathbb{R}}, \mQ, \mD' \right) \right) = o(1)$$
Now note that since $\mD$ is symmetric, it follows that $\TV( \mD, \mD') = o(n^{-1})$ since $x \in \mD$ lies in $[-1, 1]$ with probability $1 - o(n^{-1})$. Another application of tensorization yields that $\TV( \mD^{\otimes n}, \mD'^{\otimes n}) = o(1)$. Coupling the latent $\nu_1, \nu_2, \dots, \nu_n$ sampled from $\mD'$ and $\mD$ in $\pr{glsm}\left(n, k, d, \{ \mP_{\nu} \}_{\nu \in \mathbb{R}}, \mQ, \mD' \right)$ and $\pr{glsm}\left(n, k, d, \{ \mP_{\nu} \}_{\nu \in \mathbb{R}}, \mQ, \mD \right)$, respectively, yields by the conditioning property that their total variation distance is $o(1)$. The desired result in the $H_1$ case then follows from the triangle inequality. Under $H_0$, an identical argument shows the desired result without conditioning on $\mathcal{C}$ being necessary. This completes the proof of the lemma.
\end{proof}

We now use this lemma, Theorem \ref{thm:isgmreduction} and Lemma \ref{lem:3a} to formally deduce our universality principle for lower bounds at the threshold $n = \tilde{\Theta}(k^2)$ in $\pr{glsm}$. The proof follows a similar structure to that of Theorems \ref{thm:rsmeopt} and \ref{thm:semicrhardness}. We omit details where they are the same.

\begin{theorem}[Lower Bounds for General Sparse Mixtures] \label{thm:univlowerbounds}
Let $(n, k, d)$ be parameters such that $n = o(k^2)$ and $k^2 = o(d)$. Suppose that $(\mD, \mQ, \{ \mP_{\nu} \}_{\nu \in \mathbb{R}}) \in \pr{uc}(n, k, d)$ and $\mathcal{B}$ is a randomized polynomial time test for $\pr{glsm}\left(n, k, d, \{ \mP_{\nu} \}_{\nu \in \mathbb{R}}, \mQ, \mD \right)$. Then $\mathcal{B}$ has asymptotic Type I$+$II error at least 1 assuming either the $k$\pr{-pc} conjecture or the $k$\pr{-pds} conjecture for some fixed edge densities $0 < q < p \le 1$.
\end{theorem}

\begin{proof}
Assume the $k$\pr{-pds} conjecture for some fixed edge densities $0 < q < p \le 1$ and let $Q = 1 - \sqrt{(1 - p)(1 - q)} + \mathbf{1}_{\{ p = 1\}} \left( \sqrt{q} - 1 \right) \in (0, 1)$. Let $w = w(n) = \omega(1)$ be a sufficiently slow-growing function and let $t$ be such that $2^t$ is the smallest power of two larger than $wk$. Note that $wn = o(k(2^t - 1))$. Now let $N$ be the largest multiple of $k$ less than $(\frac{p}{Q} + 1)^{-1} k \cdot 2^t$. By construction, we have that $N = \Theta(k \cdot 2^t) = \Theta(wk^2)$. For a sufficiently slow-growing choice of $w$, it follows that $N \le d/2$. Now consider $\mathcal{A}_1$, which applies $k\pr{-pds-to-isgm}$ to map from $k\pr{-pds}(N, k, p, q)$ to $\pr{isgm}(n, k, d, \mu, 1/2)$ with
$$\mu = \frac{c'\delta}{2 \sqrt{6\log (k \cdot 2^t) + 2\log (p - Q)^{-1}}} \cdot \frac{1}{\sqrt{2^t}} \asymp \sqrt{\frac{1}{wk \log n}}$$
where $\delta = \min \left\{ \log \left( \frac{p}{Q} \right), \log \left( \frac{1 - Q}{1 - p} \right) \right\}$ and $c' > 0$ is a sufficiently small constant. These parameters satisfy the conditions needed to apply Theorem \ref{thm:isgmreduction} and thus $\mathcal{A}_1$ performs this mapping within $o(1)$ total variation error. Combining this with Lemma \ref{lem:univlem} and applying Lemma \ref{lem:tvacc} implies that $\mathcal{A}$ maps $k\pr{-pds}(N, k, p, q)$ to $\pr{glsm}\left(n, k, d, \{ \mP_{\nu} \}_{\nu \in \mathbb{R}}, \mQ, \mD \right)$ within $o(1)$ total variation error under both $H_0$ and $H_1$. Since $N = \omega(k^2)$, if $\mathcal{B}$ were to have an asymptotic Type I$+$II error less than 1, then this would contradict the $k$\pr{-pds} conjecture at edge densities $0 < q < p \le 1$ by Lemma \ref{lem:3a}. This completes the proof of the theorem.
\end{proof}

\subsection{The Universality Class UC$(n, k, d)$}
\label{subsec:universalitydiscussion}

The result in Theorem \ref{thm:univlowerbounds} shows universality of the computational sample complexity of $n = \Omega(k^2)$ for learning sparse mixtures under the conditions of $\pr{uc}(n, k, d)$. In this section, we make several remarks on the implications of showing hardness for $\pr{glsm}$ under our conditions $\pr{uc}(n, k, d)$.

\paragraph{Remarks on UC$(n, k, d)$.} The conditions for $(\mD, \mQ, \{ \mP_{\nu} \}_{\nu \in \mathbb{R}}) \in \pr{uc}(n, k, d)$ have the following two notable properties.
\begin{itemize}
\item \textit{They are conditions on marginals}: The conditions in $\pr{uc}(n, k, d)$ are entirely in terms of the likelihood ratios $d\mP_\nu/d\mQ$ between the planted and non-planted marginals. In particular, they do not depend on any properties of high-dimensional distributions constructed from the $\mP_{\nu}$ and $\mQ$. Thus Theorem \ref{thm:univlowerbounds} extracts the high-dimensional structure leading to statistical-computational gaps in instances of $\pr{glsm}$ as a condition on the marginals $\mP_{\nu}$ and $\mQ$.
\item \textit{Their dependence on $n$ and $d$ is negligible}: As noted in Section \ref{sec:summary}, when the likelihood ratios are relatively concentrated the dependence of the conditions in $\pr{uc}(n, k, d)$ on $n$ and $d$ is nearly negligible. Specifically, the upper bounds on the functions of the likelihood ratios $d\mP_\nu/d\mQ$ only depend logarithmically on $n$. If the ratios $d\mP_\nu/d\mQ$ are concentrated under $\mP_{\nu}$ and $\mQ$ with exponentially decaying tails, then the tail probability bounds of $o(n^{-3} d^{-1})$ in $\pr{uc}(n, k, d)$ only impose a mild condition on the $\mP_{\nu}$ and $\mQ$. Instead, the conditions in $\pr{uc}(n, k, d)$ almost exclusively depend on $k$, implying that they will not implicitly require a stronger dependence between $n$ and $k$ to produce hardness than the $n = \tilde{o}(k^2)$ condition that arises from our reductions. Thus Theorem \ref{thm:univlowerbounds} does show a universality principle for the computational sample complexity of $n = \tilde{\Theta}(k^2)$.
\end{itemize}

\paragraph{$\mD$ and Parameterization over $[-1,1]$.} As remarked in Section \ref{sec:summary}, $\mD$ and the indices of $\mP_{\nu}$ can be reparameterized without changing the underlying problem. The assumption that $\mD$ is symmetric and mostly supported on $[-1, 1]$ is for notational convenience.

While the output vectors $(X_1, X_2, \dots, X_n)$ of our reduction $k$\textsc{-pds-to-glsm} are independent, their coordinates have dependence induced by the mixture $\mD$. The fact that our reduction samples the $\nu_i$ implies that if these values were revealed to the algorithm, the problem would still remain hard: an algorithm for the latter could be used together with the reduction to solve \kpc. However, even given the $\nu_i$ for the $i$th sample, our reduction is such that whether the planted marginals in the $i$th sample are distributed according to $\mP_{\nu_i}$ or $\mP_{-\nu_i}$ remains unknown to the algorithm. Intuitively, our setup chooses to parameterize the distribution $\mD$ over $[-1, 1]$ such that the sign ambiguity between $\mP_{\nu_i}$ or $\mP_{-\nu_i}$ is what is producing hardness below the sample complexity of $n = \tilde{\Omega}(k^2)$.

\paragraph{Implications for Concentrated LLR.} We now give several remarks on the conditions $\pr{uc}(n, k, d)$ in the case that the log-likelihood ratios (LLR) $\log d\mP_{\nu}/d\mQ (x)$ are sufficiently well-concentrated if $x \sim \mQ$ or $x \sim \mP_{\nu}$. Suppose that $(\mD, \mQ, \{ \mP_{\nu} \}_{\nu \in \mathbb{R}}) \in \pr{uc}(n, k, d)$, fix some function $w(n) \to \infty$ as $n \to \infty$ and fix some $\nu \in [-1,1]$. If $S_{\mQ}$ is the common support of the $\mP_{\nu}$ and $\mQ$, define $S$ to be
$$S = \left\{ x \in S_{\mQ} : \frac{1}{\sqrt{wk \log n}} \ge \left| \frac{d\mP_{\nu}}{d\mQ} (x) - \frac{d\mP_{-\nu}}{d\mQ} (x) \right| \quad \textnormal{and} \quad \frac{1}{wk \log n} \ge \left|\frac{d\mP_{\nu}}{d\mQ} (x) + \frac{d\mP_{-\nu}}{d\mQ} (x) - 2 \right| \right\}$$
Now note that
\begin{align*}
\TV\left( \mP_{\nu}, \mP_{-\nu} \right) &= \frac{1}{2} \cdot \bE_{x \in \mQ} \left[ \left| \frac{d\mP_{\nu}}{d\mQ} (x) - \frac{d\mP_{-\nu}}{d\mQ} (x) \right| \right] \\
&\le \frac{1}{2} \cdot \bE_{x \in \mQ} \left[ \left| \frac{d\mP_{\nu}}{d\mQ} (x) - \frac{d\mP_{-\nu}}{d\mQ} (x) \right| \cdot \mathbf{1}_S(x) \right] + \frac{1}{2} \cdot \mP_{\nu}\left[S^C\right] + \frac{1}{2} \cdot \mP_{-\nu}\left[S^C\right] \\
&\le \frac{1}{2\sqrt{wk\log n}} + o(n^{-3} d^{-1}) \lesssim \frac{1}{\sqrt{k \log n}}
\end{align*}
A similar calculation with the second condition defining $S$ shows that
$$\TV\left( \pr{mix}_{1/2}\left(\mP_{\nu}, \mP_{-\nu} \right), \mQ \right) \lesssim \frac{1}{k \log n}$$
If the LLRs $\log d\mP_{\nu}/d\mQ$ are sufficiently well-concentrated, then the random variables
$$\left| \frac{d\mP_{\nu}}{d\mQ} (x) - \frac{d\mP_{-\nu}}{d\mQ} (x) \right| \quad \text{and} \quad \left|\frac{d\mP_{\nu}}{d\mQ} (x) + \frac{d\mP_{-\nu}}{d\mQ} (x) - 2 \right|$$
will also concentrate around their means if $x \sim \mQ$. LLR concentration also implies that this is true if $x \sim \mP_\nu$ or $x \sim \mP_{-\nu}$. Thus, under sufficient concentration, the conditions in $\pr{uc}(n, k, d)$ reduce to the much more interpretable conditions
$$\TV\left( \mP_{\nu}, \mP_{-\nu} \right) = \tilde{o}(k^{-1/2}) \quad \text{and} \quad \TV\left( \pr{mix}_{1/2}\left(\mP_{\nu}, \mP_{-\nu} \right), \mQ \right) = \tilde{o}(k^{-1})$$
These conditions directly measure the amount of statistical signal present in the planted marginals $\mP_{\nu}$. The relevant calculations for an example application of Theorem \ref{thm:univlowerbounds} when the LLR concentrates is shown below for sparse PCA. In \cite{brennan2019universality}, various assumptions of concentration of the LLR and analogous implications for computational lower bounds in submatrix detection are analyzed in detail. We refer the reader to Sections 3 and 9 of \cite{brennan2019universality} for the calculations needed to make the discussion here precise.

We remark that, assuming sufficient concentration on the LLR, the analysis of the $k$-sparse eigenvalue statistic from \cite{berthet2013complexity} yields an information-theoretic upper bound for $\pr{glsm}$. Given samples $(X_1, X_2, \dots, X_n)$, consider forming the LLR-processed samples $Z_i$ with
$$Z_{ij} = \bE_{\nu \sim \mD} \left[ \log \frac{d\mP_{\nu}}{d\mQ} (X_{ij}) \right]$$
for each $i \in [n]$ and $j \in [d]$. Now consider taking the $k$-sparse eigenvalue of the samples $Z_1, Z_2, \dots, Z_n$. Under sub-Gaussianity assumptions on the $Z_{ij}$, the analysis in Theorem 2 of \cite{berthet2013complexity} applies. Similarly, the analysis in Theorem 5 of \cite{berthet2013complexity} continues to hold, showing that the semidefinite programming algorithm for sparse PCA yields an algorithmic upper bound for $\pr{glsm}$. In many setups captured by $\pr{glsm}$ such as sparse PCA, learning sparse mixtures of Gaussians and learning sparse mixtures of Rademachers, these analyses and our lower bounds together confirm a $k$-to-$k^2$ gap. As information-theoretic limits and algorithms are not the focus of this paper, we omit these details.

\paragraph{Sparse PCA and Specific Distributions.} One specific example captured by our universality principle and that falls under the concentrated LLR setup discussed above is sparse PCA in the spiked covariance model. The statistical-computational gaps of sparse PCA have been characterized based on the planted clique conjecture in a line of work \cite{berthet2013optimal, berthet2013complexity, wang2016statistical, gao2017sparse, brennan2018reducibility, brennan2019optimal}. We show that our universality principle faithfully recovers the $k$-to-$k^2$ gap for sparse PCA shown in \cite{berthet2013optimal, berthet2013complexity, wang2016statistical, gao2017sparse, brennan2018reducibility} assuming the $k\pr{-pc}$ conjecture up to $k = o(\sqrt{n})$. We remark that \cite{brennan2019optimal} shows stronger hardness based on weaker forms of the $\pr{pc}$ conjecture.

We show that sparse PCA corresponds to $\pr{glsm}\left(n, k, d, \{ \mP_{\nu} \}_{\nu \in \mathbb{R}}, \mQ, \mD \right)$ for the proper choice of $(\mD, \mQ, \{ \mP_{\nu} \}_{\nu \in \mathbb{R}}) \in \pr{uc}(n, k, d)$ exactly at its conjectured computational barrier. In particular, we have the following corollary of Theorem \ref{thm:univlowerbounds}.

\begin{corollary}[Lower Bounds for Sparse PCA]
Let $\pr{spca}(n, k, d, \theta)$ be the testing problem
$$H_0 : (X_1, X_2, \dots, X_n) \sim_{\textnormal{i.i.d.}} \mN(0, I_d)^{\otimes n} \quad H_1 : (X_1, X_2, \dots, X_n) \sim_{\textnormal{i.i.d.}} \mN\left(0, I_d + \theta vv^\top\right)^{\otimes n}$$
where $v$ is a $k$-sparse unit vector in $\mathbb{R}^d$ chosen uniformly at random among all such vectors with nonzero entries equal to $1/\sqrt{k}$. If $n = o(k^2)$, $k^2 = o(d)$ and $\theta = o(1/(\log n)^4)$, then $\pr{spca}(n, k, d, \theta)$ can be expressed as $\pr{glsm}\left(n, k, d, \{ \mP_{\nu} \}_{\nu \in \mathbb{R}}, \mQ, \mD \right)$ for some choice $(\mD, \mQ, \{ \mP_{\nu} \}_{\nu \in \mathbb{R}}) \in \pr{uc}(n, k, d)$.
\end{corollary}

\begin{proof}
Note that if $X \sim \mN(0, I_d + \theta vv^\top)$ then $X$ can be written as $X = \sqrt{\theta} \cdot gv + G$ where $G \sim \mN(0, I_d)$ and $g \sim \mN(0, 1)$; $X$ can be rewritten as $X = \sqrt{3\theta \log n} \cdot g'v + G$ where $g' \sim \mN(0, 1/\sqrt{3\log n})$. Now consider setting
$$\mD = \mN(0, 1/\sqrt{3\log n}), \quad \mP_{\nu} = \mN\left( \nu \sqrt{\frac{3\theta \log n}{k}}, 1 \right), \quad \mQ = \mN(0, 1)$$
Note that the probability that $x \sim \mD$ satisfies $x \in [-1, 1]$ is $1 - o(n^{-1})$ by standard Gaussian tail bounds. Fix some $\nu$ and let $t = \nu \sqrt{\frac{3\theta \log n}{k}}$. Note that if $\theta = o(1/(\log n)^4)$, we have that
$$\left|\frac{d\mP_\nu}{d\mQ}(x) - \frac{d\mP_{-\nu}}{d\mQ}(x) \right| = \left| e^{tx - t^2/2} - e^{-tx - t^2/2}\right| = \Theta(tx)$$
ift $tx = o(1)$. Now note that this quantity is $o(1/\sqrt{k \log n})$ as long as $x = o(\log n)$. Note that $x = o(\log n)$ occurs with probability at least $1 - (n^{-3} d^{-1})$ as long as $d = \text{poly}(n)$ by standard Gaussian tail bounds. Now note that
$$\left|\frac{d\mP_\nu}{d\mQ}(x) + \frac{d\mP_{-\nu}}{d\mQ}(x) - 2\right| = \left| e^{tx - t^2/2} + e^{-tx - t^2/2} - 2\right| = \Theta(t^2)$$
holds if $tx = o(1)$, which follows from $x = o(\log n)$. As shown above, this occurs with probability at least $1 - (n^{-3} d^{-1})$. Since $t^2 = o(1/k \log n)$, we have that these $(\mD, \mQ, \{ \mP_{\nu} \}_{\nu \in \mathbb{R}})$ are in $\pr{uc}(n, k, d)$, completing the proof.
\end{proof}

Combining this lower bound with the subsampling internal reduction in Section 8.1 of \cite{brennan2019optimal} extends this reduction to the small signal regime of $\theta = \tilde{\Theta}(n^{-\alpha})$, recovering the optimal computational threshold of $n = \tilde{\Omega}(k^2/\theta^2)$ for all $(n, k, d, \theta)$ with $\theta = \tilde{o}(1)$ and $k^2 = o(d)$. Similar calculations to those in the above corollary can be used to show many other choices of $(\mD, \mQ, \{ \mP_{\nu} \}_{\nu \in \mathbb{R}})$ are in $\pr{uc}(n, k, d)$. Some examples are:
\begin{itemize}
\item Balanced sparse Gaussian mixtures where $\mQ = \mN(0, 1)$ and $\mP_{\nu} = \mN(\theta \nu, 1)$ where $\theta = \tilde{o}(k^{-1/2})$ and $\mD$ is any distribution over $[-1, 1]$.
\item The Bernoulli case where $\mQ = \text{Bern}(1/2)$ and $\mP_{\nu} = \text{Bern}(1/2 + \theta \nu)$ where $\theta = \tilde{o}(k^{-1/2})$ and $\mD$ is any distribution over $[-1, 1]$.
\item Sparse mixtures of exponential distributions where $\mQ = \text{Exp}(\lambda)$ and $\mP_{\nu} = \text{Exp}(\lambda + \theta \nu)$, $\mD$ is any distribution over $[-1, 1]$ and
$$\theta = \tilde{o}\left( k^{-1/2} \cdot \min \left\{ \lambda^{1/2}, (1 + \lambda)^{-1} \right\} \right)$$
\item Sparse mixtures of centered Gaussians with difference variances where $\mQ = \mN(0, 1)$ and $\mP_{\nu} = \mN(0, 1 + \theta \nu)$ where $\theta = \tilde{o}\left( k^{-1/2} \right)$ and $\mD$ is any distribution over $[-1, 1]$.
\end{itemize}
We remark that the conditions of $\pr{uc}(n, k, d)$ can be verified for many more choices of $\mD, \mQ$ and $\mP_{\nu}$ using the computations outlined in the discussion above on the implications of our result for concentrated LLR. Furthermore, tradeoffs between $n$, $k$ and $\theta$ at smaller levels of signal $\theta$ can be obtained from the lower bounds above through subsampling internal reductions, analogously to sparse PCA.

\section{Evidence for the $k$-Partite Planted Clique Conjecture}
\label{sec:evidencekpc}

Our $k$-partite versions of planted clique and planted dense subgraph, $k$-\textsc{pc} and $k$-\textsc{pds}, seem to be just as hard as the standard versions. While the partition $E$ in their definitions contains a slight amount of information about the position of the clique, in this section we provide evidence that the hardness threshold for $k$ is unchanged by considering two restricted classes of algorithms: tests based on low-degree polynomials as well as statistical query algorithms. For the expert on either of these, the high-level message substantiated below is that $k$-\textsc{pc} and ordinary \textsc{pc} are virtually identical in both their Fourier spectrum and statistical dimension (and the same is true for $k$-\textsc{pds} versus ordinary \textsc{pds}). We emphasize that this section contains no new ideas, and repeats the arguments of \cite{hopkinsThesis} and \cite{feldman2013statistical} with a few tiny changes. If anything, however, this only supports the goal of the section which is to substantiate our claim that {\kpc } is extremely similar to standard \textsc{pc}.

We remark here that whenever we refer to {\kpds } in this section we will have $q=1/2$ and $p=1/2 + n^{-\delta}$ for arbitrary $\delta>0$. Furthermore, the analysis of low-degree polynomial and statistical query algorithms for {\kpds} where $p - q = \Omega(1)$ is essentially the same as for {\kpc}. Also, these analyses remain qualitatively the same when $q$ is replaced by a constant other than $1/2$.

\subsection{Low-degree Likelihood Ratio}

This subsection draws heavily from the paper by Hopkins and Steurer \cite{hopkins2017efficient} and on Hopkins's thesis \cite{hopkinsThesis}. Based their understanding of the sum-of-squares (SOS) hierarchy applied to statistical inference problems, they conjecture that \emph{low-degree polynomial tests} capture the full power of SOS and, more generally, all efficient hypothesis testing algorithms. 

%

\begin{conjecture} \label{c:lowDeg}
	For a broad class of hypothesis testing problems $H_0$ versus $H_1$, there is a test running in time $n^{\Ot(D)}$ with Type I + II errors tending to zero if and only if there is a successful $D$-simple statistic, i.e. a polynomial $f$ of degree at most $D$ such that $\E_{H_0}f(X)=0$ and $\E_{H_0}f(X)^2=1$ yet $\E_{H_1}f(X)\to \infty$.
\end{conjecture}

By the Neyman-Pearson lemma, the likelihood ratio test is the optimal test with respect to Type I + II errors: given a sample $X$, declare $H_1$ if $\lr(X) = \frac{\P_{H_1}(X)}{\P_{H_0}(X)}>1$ and $H_0$ otherwise. Of course, computing the likelihood ratio is intractable for the problems of interest. The low-degree likelihood ratio $\lrd$ is the orthogonal projection of the likelihood ratio onto the subspace of polynomials of degree at most $D$, and as stated in the following theorem is the optimal test of a given degree. We take $H_0$ to be the uniform distribution on the appropriate dimension hypercube $\{-1,+1\}^N$, and here the projection is with respect to the inner product $\la f,g\ra = \E_{H_0} f(X) g(X)$, which also defines a norm $\|f\|_2 = \la f,f\ra$. 

\begin{theorem}[Page 35 of \cite{hopkinsThesis}] The optimal $D$-simple statistic is the low-degree likelihood ratio, i.e. it holds that
$$
\max_{{f\in \bR[x]_{{\leq D}}}\atop \E_{H_0}f(X)=0} \frac{\E_{H_1} f(X)}{\sqrt{\E_{H_0} f(X)^2}} = \|\lrd - 1\|_2
$$	
\end{theorem}

Thus existence of low-degree tests for a given problem boils down to computing the norm of the low-degree likelihood ratio.
In order to bound the norm on the right-hand side it is useful to re-express it in terms of the standard Boolean Fourier basis. The collection of functions $\{\chi_\alpha(X) = \prod_{e\in \alpha} X_e: \alpha \subseteq [N]\}$ is an orthonormal basis over the space $\{-1,1\}^{N}$ with inner product defined above.
By orthonormality of the basis, for any basis function $\chi_\alpha$ with $1\leq |\alpha|\leq D$,
$$
\la \chi_\alpha, \lrd -1\ra = \la \chi_\alpha, \lr \ra = \E_{H_0} \chi_\alpha(X) \lr(X) = \E_{H_1}\chi_\alpha(X)
$$
and $\E_{H_0} \lrd=\E_{H_1} 1=1$ so that  $\la 1, \lrd -1\ra=0$. 
It then follows by Parseval's identity that 
\begin{equation}\label{e:energy}
	\|\lrd - 1\|_2 = \left( \sum_{1\leq|\alpha|\leq D} \big(\E_{H_1}\chi_\alpha(X)\big)^2\right)^{1/2}
\end{equation}
which is exactly the Fourier energy up to degree $D$. By directly computing the Fourier coefficients of $\lr^{\le D}$ for {\kpc} and {\kpds}, we arrive at the following proposition.

\begin{proposition}[Failure of low-degree tests] \label{prop:lowdegree}
Consider {\kpc } and {\kpds } with $k= n^{1/2-\epsilon}$. Tests of degree at most $D$ fail, i.e., $\|\lrd - 1\|=O(1)$ in the following cases:
\begin{enumerate}
\item[(i)] in {\kpc }, if $D \leq C\log n $ for a sufficiently small constant $C$
\item[(ii)] in {\kpds}, if $D \leq n^{\delta / 10}$.
	\end{enumerate}
\end{proposition}

Combining the proposition with Conjecture~\ref{c:lowDeg} implies that in the case of {\kpc } there is no polynomial time algorithm and for {\kpds } there is no algorithm with runtime less than $\exp(n^{\delta/10})$. A proof of this proposition can be found in Appendix \ref{sec:appendix}, which is similar to the computation of the Fourier spectrum of $\pr{pc}$ in \cite{hopkinsThesis}. We also make the following remark.

\begin{remark}
The computational threshold for $k$ in both $k\textsc{-pds}(n,k,\frac12 + n^{-\delta},\frac12)$ and $\textsc{pds}(n,k,\frac12 + n^{-\delta},\frac12)$ are no longer at $k=\sqrt{n}$, but rather at $k=n^{1/2+\delta'}$ for some $\delta'>0$. For this reason we see failure of low-degree tests up to degree (say) $n^{\delta/10}$ even for $\epsilon=0$, i.e. $k=\sqrt{n}$. 
\end{remark}

%

\subsection{Statistical Query Algorithms and Statistical Dimension}

In this section we verify that the lower bounds shown by \cite{feldman2013statistical} for \textsc{pc} for a generalization of statistical query algorithms hold essentially unchanged for \kpc. The same approach results in lower bounds for {\kpds } that are essentially the same as for \textsc{pds}, which we omit to avoid redundancy. 

The Statistical Algorithm framework of \cite{feldman2013statistical} applies to distributional problems, where the input is a sequence of i.i.d. observations from a distribution $D$. We therefore define distributional versions of {\kpds } and \kpc. Just as in \cite{feldman2013statistical}, we first define a bipartite version of each problem, and will then define corresponding distributional versions by thinking of samples as generated by observing neighborhoods of right-hand side vertices. 

We consider the following bipartite versions of $k$-\textsc{pds} and {\kpc } with $n$ vertices per side. Let $k$ divide $n$ and $E=E_1\cup E_2\cup \dots \cup E_k$ be a known partition of left-hand side vertices $[n]$ with $|E_i|=n/k$ for each $i$. A $k\times k$ bipartite clique $S$ is formed by selecting a single vertex u.a.r. from each $E_i$ on the LHS and including each of the RHS vertices independently with probability $k/n$ each. Note that only the LHS respects the partition $E$. Now in {\kpc } all edges between LHS and RHS vertices in $S$ are included, and remaining edges are included with probability $1/2$ each, with the obvious generalization to edge probabilities $p$ and $q$ for \kpds. We now define the corresponding distributional version of \kpc.

\begin{definition}
Let $k$ divide $n$ and fix a known partition of $E=E_1\cup E_2\cup \dots \cup E_k$ of $[n]$ with $|E_i|=n/k$. Let $S\subset [n]$ be a subset of indices with $|S\cap E_i|=1$ for each $i\in [k]$. The distribution $D_S$ over $\{0,1\}^n$ produces with probability $1-k/n$ a uniform point $X\sim \mathrm{Unif}(\{0,1\}^n)$ and with probability $k/n$ a point $X$ with $X_i=1$ for all $i\in S$ and $X_{S^c}\sim \mathrm{Unif}(\{0,1\})^{n-k}$. The \textbf{distributional bipartite $k$-\textsc{PC}} problem is to find the subset $S$ given some number of independent samples $m$ from $D_S$. 
\end{definition} 

The correspondence between this distributional problem where samples are observed and the bipartite version can be seen by considering algorithms that sequentially examine the neighborhoods of the RHS vertices in the bipartite graph. Because there are only $n$ RHS vertices, meaningful conclusions regard number of samples $m\leq n$. We now make an important remark on the relationship between these formulations and our reductions.

\begin{remark}
Our main reductions from {\kpc} and {\kpds} to $\pr{rsme}$ and $\pr{glsm}$ both only require the partition structure of $E$ along one axis of the adjacency matrix of the input graph. Therefore both of these reductions can easily be adapted to begin with the distributional bipartite variants of {\kpc} and {\kpds}. However, the reductions to $\pr{semi-cr}$ require the partition structure along both the rows and columns of the adjacency matrix of the input graph.
\end{remark}

 Let $\cX=\{-1,+1\}^n$ denote the space of configurations and let
 $\cD$ be a set of distributions over $\cX$. Let $\cF$ be a set of solutions (in our case, clique positions) and $\cZ:\cD\to 2^\cF$ be a map taking each distribution $D\in \cD$ to a subset of solutions $\cZ(D)\subseteq \cF$ that are defined to be valid solutions for $D$. In our setting, since each clique position is in one-to-one correspondence with distributions, there is a single clique $\cZ(D)$ corresponding to each distribution $D$. For $m>0$, the \emph{distributional search problem} $\cZ$ over $\cD$ and $\cF$ using $m$ samples is to find a valid solution $f\in \cZ(D)$ given access to $m$ random samples from an unknown $D\in \cD$. 

One class of algorithms we are interested in are called \emph{unbiased statistical algorithms}, defined by access to an unbiased oracle.

\begin{definition}[Unbiased Oracle]
Let $D$ be the true distribution. A query to the oracle consists of any function $h:\cX\to \{0,1\}$, and the oracle then takes an independent random sample $X\sim D$ and returns $h(X)$.
\end{definition}

These algorithms access the sampled data only through the oracle: unbiased statistical algorithms outsource the computation. Because the data is accessed through the oracle, it is possible to prove \emph{unconditional} lower bounds using information-theoretic methods. Another oracle, VSTAT, is similar but the oracle is allowed to make an adversarial perturbation of the function evaluation. It is shown in \cite{feldman2013statistical} via a simulation argument that the two oracles are approximately equivalent. 

\begin{definition}[VSTAT Oracle]
Let $D$ be the true distribution and $t>0$ a sample size parameter. A query to the oracle again consists of any function $h:\cX\to \{0,1\}$, and the oracle returns an arbitrary value $v\in[\E_{ D}h(X)-\tau, \E_{ D}h(X)+\tau]$, where $\tau = \max\{1/t,\sqrt{\E_{ D}h(X)(1-\E_{ D}h(X))/t}\}$.
\end{definition}

%
%
%

We borrow some definitions from \cite{feldman2013statistical}. For a distribution $D$ we define the inner product $\la f,g\ra_D = \E_{X\sim D}f(X) g(X)$ and the corresponding norm $\|f\|_D = \sqrt{\la f,f\ra_D}$. For distributions $D_1$ and $D_2$ both absolutely continuous with respect to $D$, the pairwise correlation is defined to be 
$$
\chi_D(D_1,D_2) = \Big|\Big\la \frac{D_1}D-1,\frac{D_2}D-1\Big\ra_D \Big|=|\la \Dh_1, \Dh_2\ra_D|\,.
$$ 
Here we defined $ \Dh_1 = \frac{D_1}D-1$.
The \emph{average correlation} $\rho(\cD, D)$ of a set of distributions $\cD$ relative to distribution $D$ is defined as 
$$
\rho(\cD, D) = \frac1{|\cD|^2} \sum_{D_1,D_2\in \cD}\chi_D(D_1,D_2) = \frac1{|\cD|^2} \sum_{D_1,D_2\in \cD}\Big|\Big\la \frac{D_1}D-1,\frac{D_2}D-1\Big\ra_D \Big|\,.
$$

We quote the definition of statistical dimension with average correlation from \cite{feldman2013statistical}, and then state a lower bound on the number of queries needed by any statistical algorithm.
\begin{definition}[Statistical dimension]\label{d:statDimProblem}
Fix $\gamma>0,\eta>0$, and search problem $\cZ$ over set of solutions $\cF$ and class of distributions $\cD$ over $\cX$. 
We consider pairs $(D,\cD_D)$ consisting of a ``reference distribution" $D$ over $\cX$ and a finite set of distributions $\cD_D\subseteq \cD$ with the following property: for any solution $f\in \cF$, the set $\cD_f = \cD_D\setminus \cZ\inv (f)$ has size at least $(1-\eta)\cdot |\cD_D|$.
Let $\ell(D,\cD_D)$ be the largest integer $\ell$ so that for any subset $\cD'\subseteq \cD_f$ with $|\cD'|\geq |\cD_f|/\ell$, the average correlation is $|\rho(\cD',D)|<\gamma$ (if there is no such $\ell$ one can take $\ell=0$). The \emph{statistical dimension} with average correlation $\gamma$ and solution set bound $\eta$ is defined to be the largest $\ell(D,\cD_D)$ for valid pairs $(D,\cD_D)$ as described, and is  denoted by $\mathrm{SDA}(\cZ,\gamma,\eta)$.
\end{definition}

\begin{theorem}[Theorems 2.7 and 3.17 of \cite{feldman2013statistical}]\label{t:sampleBound}
  Let $\cX$ be a domain and $\cZ$ a search problem over a set of solutions $\cF$ and a class of distributions $\cD$ over $\cX$. For $\gamma>0$ and $\eta\in (0,1)$, let $\ell = \mathrm{SDA}(\cZ,\gamma,\eta)$. Any (possibly randomized) statistical query algorithm that solves $\cZ$ with probability $\delta>\eta$ requires at least $\ell$ calls to the $VSTAT(1/(3\gamma))$ oracle to solve $\cZ$. 
  
Moreover, any statistical query algorithm requires at least  $m$ calls to the Unbiased Oracle for
  $
  m = \min\left\{ \frac{\ell(\delta- \eta)}{2(1-\eta)},\frac{(\delta-\eta)^2}{12\gamma}\right\}
  $.
  In particular, if $\eta \leq 1/6$, then any algorithm with success probability at least $2/3$ requires at least $\min\{ \ell/4,1/48\gamma\}$ samples from the Unbiased Oracle.
\end{theorem}

%
%
%
%

%


Let $\cS$ be the set of all size $k$ subsets of $[n]$ respecting the partition $E$, i.e., $\cS = \{S:|S|=k\text{ and } |S\cap E_i|=1\text{ for }i\in [k]\}$, and note that $|\cS| = (n/k)^k$. We henceforth use $D$ to denote the uniform distribution on $\{0,1\}^n$. The following lemma is exactly the same as in \cite{feldman2013statistical}, except that we further restrict $S$ and $T$ to be in $\cS$ rather than arbitrary size $k$ subsets of $[n]$, which does not change the bound.

\begin{lemma}[Lemma 5.1 in \cite{feldman2013statistical}]\label{l:avgCorr}
For $S,T\in \cS$, $\chi_D(D_S,D_T) = |\la \Dh_S, \Dh_T\ra_D|\leq 2^{|S\cap T|} k^2 / n^2$.  	
\end{lemma}

We now proceed to derive the main statistical query lower bounds for the bipartite formulations of {\kpc} and {\kpds}. The lemma below is similar to Lemma 5.2 in \cite{feldman2013statistical}. Its proof is deferred to Appendix \ref{sec:appendix}.

\begin{lemma}[Modification of Lemma 5.2 in \cite{feldman2013statistical}]\label{l:avgCorrLargeSets}
	Let $\delta \geq 1/\log n$ and $k\leq n^{1/2 - \delta}$. For any integer $\ell \leq k$, $S\in \cS$, and set $A\subseteq \cS$ with $|A|\geq 2|\cS|/ n^{2\ell \delta}$, 
	$$
	\frac1{|A|} \sum_{T\in A} \big| \la \Dh_S, \Dh_T\ra_D \big| \leq 2^{\ell + 3}\frac{k^2}{n^2}\,.
	$$
\end{lemma}

This lemma now implies the following statistical dimension lower bound.

\begin{theorem}[Analogue of Theorem 5.3 of \cite{feldman2013statistical}]\label{t:SQdim}
	For $\delta\geq 1/\log n$ and $k\leq n^{1/2-\delta}$ let $\cZ$ denote the distributional bipartite {\kpc } problem. If $\ell\leq k$ then $SDA(\cZ, 2^{\ell+3} k^2/n^2,\big(\frac nk\big) ^{-k} ) \geq n^{2\ell \delta}/8$. 
\end{theorem}
\begin{proof}
For each clique position $S$ let $\cD_S = \cD\setminus\{D_S\}$. Then $|\cD_S| = \big(\frac nk\big) ^k -1=\big(1-\big(\frac nk\big) ^{-k}\big)|\cD|$. Now for any $\cD'$ with $|\cD'|\geq 2|\cS|/ n^{2\ell \delta}$ we can apply Lemma~\ref{l:avgCorrLargeSets} to conclude that $\rho(\cD',D)\leq 2^{\ell + 3}k^2/n^2$. By Definition~\ref{d:statDimProblem} of statistical dimension this implies the bound stated in the theorem. 
\end{proof}

Applying Theorem~\ref{t:sampleBound} to this statistical dimension lower bound yields the following hardness for statistical query algorithms.

\begin{corollary}[SQ lower bound for recovery in \kpc]
	For any constant $\delta>0$ and $k\leq n^{1/2-\delta}$, any SQ algorithm that solves the distributional bipartite {\kpc } problem requires $\Omega(n^2/k^2\log n)=\tilde \Omega(n^{1+2\delta})$ queries to the Unbiased Oracle.
\end{corollary}

This is to be interpreted as impossible, as there are only $n$ RHS vertices/samples available in the actual bipartite graph. Because all the quantities in Theorem~\ref{t:SQdim} are the same as in \cite{feldman2013statistical} up to constants, the same logic as used there allows to deduce a statement regarding the hypothesis testing version, stated there as Theorems 2.9 and 2.10.

\begin{corollary}[SQ lower bound for decision version of \kpc]
	For any constant $\delta>0$, suppose $k\leq n^{1/2-\delta}$. Let $D=\mathrm{Unif}(\{0,1\}^n)$ and let $\cD$ be the set of all planted bipartite {\kpc } distributions (one for each clique position). Any SQ algorithm that solves the hypothesis testing problem between $\cD$ and $D$ with probability better than $2/3$ requires $\Omega(n^2/k^2)$ queries to the Unbiased Oracle. 
A similar statement holds for VSTAT. There is a $t = n^{\Omega(\log n)}$ such that any randomized SQ algorithm that solves the hypothesis testing problem between $\cD$ and $D$ with probability better than $2/3$ requires at least $t$ queries to $VSTAT(n^{2-\delta}/k^2)$. 
\end{corollary}

\section*{Acknowledgements}

We are greatly indebted to Jerry Li for introducing the conjectured statistical-computational gap for robust sparse mean estimation and for discussions that helped lead to this work. We thank Ilias Diakonikolas for pointing out the statistical query model construction in \cite{diakonikolas2017statistical}. We also thank Frederic Koehler, Sam Hopkins, Philippe Rigollet, Enric Boix Adser\`{a}, Dheeraj Nagaraj and Austin Stromme for inspiring discussions on related topics.

\bibliography{GB_BIB.bib}
\bibliographystyle{alpha}

\appendix

\section{Deferred Proofs}
\label{sec:appendix}

In this section, we present the deferred proofs from the body of the paper. We first present the proof of Lemma \ref{lem:tvacc}. 

\begin{proof}[Proof of Lemma \ref{lem:tvacc}]
This follows from a simple induction on $m$. Note that the case when $m = 1$ follows by definition. Now observe that by the data-processing and triangle inequalities of total variation, we have that if $\mathcal{B} = \mathcal{A}_{m-1} \circ \mathcal{A}_{m-2} \circ \cdots \circ \mathcal{A}_1$ then
\begin{align*}
\TV\left( \mathcal{A}(\mP_0), \mP_m \right) &\le \TV\left( \mathcal{A}_m \circ \mathcal{B}(\mP_0), \mathcal{A}_m(\mP_{m - 1}) \right) + \TV\left(\mathcal{A}_m(\mP_{m - 1}), \mP_m \right) \\
&\le \TV\left( \mathcal{B}(\mP_0), \mP_{m - 1} \right) + \epsilon_m \\
&\le \sum_{i = 1}^m \epsilon_i
\end{align*}
where the last inequality follows from the induction hypothesis applied with $m - 1$ to $\mathcal{B}$. This completes the induction and proves the lemma.
\end{proof}

We now present the proof of Proposition \ref{prop:lowdegree}, which is similar to the computation of the Fourier spectrum of $\pr{pc}$ in \cite{hopkinsThesis}. We only provide a sketch of details similar to \cite{hopkinsThesis} for brevity.

\begin{proof}[Proof of Proposition \ref{prop:lowdegree}]
Recall that in $k$-\textsc{pc} the nodes are partitioned into $k$ sets $E_1,\dots, E_k$ of size $n/k$ each. Denote by $S$ the clique vertices. We are guaranteed that $|S\cap E_i|=1$ for all $1\leq i\leq k$, and thus the edges between nodes in any given $E_i$ contain no information and can be removed without changing the clique. We take the set of possible edges $\mathcal{E}_0\subset {[n]\choose 2}$ in an instance of $k$-\textsc{pc} to be pairs $ij$ with $i$ and $j$ from different partitions. 
Let $\mathcal{S}=\{S:|S\cap E_i|=1\}$ be the collection of all size $k$ subsets respecting the given partition $E$. Note that choosing an $S$ uniformly from $\mathcal{S}$ amounts to selecting a single node uniformly at random from each set in the partition.  Let $P_S$ be the distribution on graphs such that $X_{ij}=1$ if $i\in S$ and $j\in S$ and otherwise $X_{ij}=\pm 1$ with probability half each. The uniform mixture over valid $S$ is denoted by $P=\E_{S\sim \mathrm{Unif}(\mathcal{S})} P_S$.

Now let $\alpha\subseteq \mathcal{E}_0$ be a subset of possible edges. The set of functions $\{\chi_\alpha(X) = \prod_{e\in \alpha} X_e: \alpha \subseteq \mathcal{E}_0\}$ comprises the standard Fourier basis on $\{\pm 1\}^{\mathcal{E}_0}$. Consider a fixed clique $S$. Just as for standard \textsc{pc}, because $\E_{P_S} X_e=0$ if $e\notin {S\choose 2}$ and non-clique edges are independent, we see that $\E_{P_S} [\chi_\alpha (X) ]= \ind \{V(\alpha)\subseteq S\}$ where $V(\alpha)$ is the set of nodes covered by edges in $\alpha$. Thus, if $V(\alpha)$ has at most one node per size $n/k$ set, then the Fourier coefficients are
$$
\E_P[\chi_\alpha (X)] = \P_{S\sim \mathrm{Unif}}(V(\alpha)\subseteq S) = \Big(\frac{1}{n/k}\Big)^{|V(\alpha)|} = \Big(\frac{k}n\Big)^{|V(\alpha)|} \,,
$$
and otherwise $\E_P[\chi_\alpha (X)] =0$. 

Remarkably, as can be seen from \cite{hopkinsThesis} or \cite{barak2016nearly} this is precisely the same Fourier coefficient as for the version of planted clique where each node is independently included in $S$ with probability $k/n$. Because the set of Fourier coefficients is indexed by $\mathcal{E}_0$ in $k$-\textsc{pc} and this is a subset of the set of Fourier coefficients in standard \textsc{pc}, it immediately follows that the quantity of interest in \eqref{e:energy} is smaller in $k$-\textsc{pc} relative to \textsc{pc}. Thus $k$-\textsc{pc} is at least as hard as \textsc{pc} from the perspective of low-degree polynomial tests. 

We briefly sketch the calculation showing a constant bound on the Fourier energy of $k$-\textsc{pc} for sets of size $|\alpha|\leq D$ for $D = C\log n$, following the calculation for \textsc{pc} in \cite{hopkinsThesis}. 
Note that if $|\alpha|\leq D$, then $|V(\alpha)|\leq 2D $ and for every $t\leq 2D$ we may bound the number of sets $\alpha\in \mathcal{E}_0$ with $|V(\alpha)|=t$ and $|\alpha| \leq D$ as 
\begin{equation}\label{e:FourierSets}
{k\choose t} \Big(\frac nk\Big)^t {{t\choose 2}\choose |\alpha|}
\leq
{k\choose t} \Big(\frac nk\Big)^t t^{\min(2D, 2t^2)}\leq n^t t^{\min(2D, 2t^2)}\,.
\end{equation}
The total Fourier energy for $|\alpha|\leq D=C\log n$ is
\begin{align*}
\sum_{\alpha \subseteq \mathcal{E}_0\atop 0< |\alpha|\leq D} (\E_{H_0} \chi_\alpha(X))^2
\leq
 \sum_{t\leq \sqrt{C\log n}} \Big(\frac kn\Big)^{2t} n^{t}t^{2t^2} + \sum_{\sqrt{C\log} n< t\leq 2C\log n}\Big(\frac kn\Big)^{2t}n^{t}t^{2C\log n}
\end{align*}
and if $k = n^{1/2 - \epsilon}$ then this is at most
$$
\sum_{t\leq \sqrt{C\log n}} n^{-2\epsilon t} t^{2t^2} + \sum_{\sqrt{C\log} n< t\leq 2C\log n} n^{-2\epsilon t}
t^{2C\log n} = O(1)\,.
$$

To compute the Fourier coefficients for $k$-$\textsc{pds}(n,k,p,q)$ with $p= 1/2 + n^{-\delta}$ and $q=1/2$, we express $\mathrm{Bern}(p)$ as the mixture $\mathrm{Bern}(p) = (2-2p)\cdot \mathrm{Bern}(1/2) + (2p-1)\cdot \mathrm{Bern}(1)$. The Fourier coefficient corresponding to set $\alpha$ with $V(\alpha)\subseteq S$ is then nonzero only if each of the edges selected the $\mathrm{Bern}(1)$ component of the mixture, so $\E \chi_\alpha(X) = (2p-1)^{|\alpha|} = (2n^{-\delta})^{|\alpha|}$. We will now take $D=n^{\delta/10}$ and again $k = n^{1/2 - \epsilon}$. 
By \eqref{e:FourierSets} the number of sets with $|\alpha|=r$ and $|V(\alpha)|=t$ is bounded by $n^t t^{2r}$, so 
\begin{align*}
\sum_{\alpha \subseteq \mathcal{E}_0\atop 0< |\alpha|\leq D} (\E_{H_0} \chi_\alpha(X))^2
&\leq 
\sum_{ t\leq 2D}\sum_{\alpha: |V(A)|=t\atop 0< |\alpha|\leq D} \Big(\frac k n\Big)^{2t} (2n^{-\delta})^{2|\alpha|}
\\&\leq
 \sum_{t\leq 2D}\sum_{\frac t2 \leq r\leq t^2/2}
n^t t^{2r} \Big(\frac k n\Big)^{2t}(2n^{-\delta})^{2r} 
\\&\leq
\sum_{t\leq 2D}\sum_{\frac t2 \leq r\leq t^2/2} n^{-2\epsilon t} (4n^{-\delta/2})^{2r}
\end{align*}
where the last inequality used $t\leq 2n^{\delta/10}$. This last quantity is $O(1)$.
\end{proof}

We now present the proof of Lemma \ref{l:avgCorrLargeSets}, which is similar to Lemma 5.2 in \cite{feldman2013statistical}.

\begin{proof}[Proof of Lemma \ref{l:avgCorrLargeSets}]
The proof is almost identical to Lemma 5.2 in \cite{feldman2013statistical} and we give a sketch here. Lemma~\ref{l:avgCorr} implies that $\sum_{T\in A}\big| \la \Dh_S, \Dh_T\ra_D \big| \leq \sum_{T\in A} 2^{|S\cap T|} k^2 / n^2$. If the only constraint on $A$ is its cardinality, then the maximum value for the RHS is obtained by adding $S$ to $A$, next $\{T:|T\cap S|=k-1\}$, and so forth with decreasing size of $|T\cap S|$, and we assume that $A$ is defined in this manner. Letting $T_\lambda = \{T: |T\cap S|=\lambda\}|$, set $\lambda_0 = \min\{\lambda: T_\lambda\neq \varnothing\}$ so that $T_\lambda\subseteq A$ for $\lambda>\lambda_0$. We bound the ratio
$$
\frac{|T_j|}{|T_{j+1}|} = \frac{{k\choose j}\big(\frac nk\big)^{k-j}}{{k\choose j+1}\big(\frac nk\big)^{k-j-1}}\geq \frac{jn}{k^2}=j n^{2\delta}\quad \text{hence}\quad |T_j|\leq \frac{|T_0|}{(j-1)! n^{2\delta j}}\leq \frac{|\cS|}{(j-1)! n^{2\delta j}}\,.
$$ 
Now
$$
|A|\leq \sum_{j\geq \lambda_0} |T_j| \leq |\cS|n^{-2\delta \lambda_0}\sum_{j\geq \lambda_0} \frac1{(j-1)!n^{2\delta(j-\lambda_0)}}\leq 2 |\cS|n^{-2\delta \lambda_0}
$$ for $n$ greater than some constant. Thus if $|A|\geq 2|\cS|/ n^{2\ell \delta}$, we must conclude that $\ell \geq \lambda_0$. We bound the quantity $\sum_{T\in A} 2^{|S\cap T|} \leq \sum_{j=\lambda_0}^k 2^j|T_j\cap A|\leq 2^{\lambda_0}|T_{\lambda_0}\cap A|+\sum_{j=\lambda_0+1}^k 2^j|T_j|\leq 2^{\lambda_0}|A| + 2^{\lambda_0+2}|T_{\lambda_0+1}|\leq 2^{\lambda_0+3}|A|\leq 2^{\ell +3}|A|$. Here we used that $|T_{j+1}|\leq |T_j| n^{-2\delta}$ to bound by a geometric series and also that $T_{\lambda_0+1}\subseteq A$. Rearranging and combining with the inequality at the start of the proof concludes the argument.
\end{proof}

\end{document}